\documentclass{theoretics}

\usepackage{ifthen}

\usepackage[english]{babel}
\usepackage{csquotes}

\usepackage{bussproofs}

\usepackage{mathtools}

\usepackage{bm}

\usepackage{varioref}

\usepackage{xspace}

\newcommand{\formuladots}{\cdots}

\newcommand{\BIGOH}[1]{\mathrm{O} \left( #1 \right)}
\newcommand{\Bigoh}[1]{\mathrm{O} \bigl( #1 \bigr)}
\newcommand{\bigoh}[1]{\mathrm{O} ( #1 )}
\newcommand{\LITTLEOH}[1]{\mathrm{o} \left( #1 \right)}
\newcommand{\Littleoh}[1]{\mathrm{o} \bigl( #1 \bigr)}
\newcommand{\littleoh}[1]{\mathrm{o} ( #1 )}

\newcommand{\Bigtheta}[1]{\Theta \bigl( #1 \bigr)}

\newcommand{\BIGOMEGA}[1]{\Omega \left( #1 \right)}
\newcommand{\Bigomega}[1]{\Omega \bigl( #1 \bigr)}
\newcommand{\bigomega}[1]{\Omega ( #1 )}

\newcommand{\Rplus}     {\mathbb{R}^{+}}
\newcommand{\N}         {\mathbb{N}}
\newcommand{\Nplus}     {\mathbb{N}^{+}}

\providecommand{\abs}[1]{\lvert#1\rvert}
\providecommand{\Abs}[1]{\bigl\lvert#1\bigr\rvert}

\newcommand{\ceiling}[1]{\lceil #1 \rceil}

\newcommand{\Floor}[1]{\bigl \lfloor #1 \bigr \rfloor}
\newcommand{\FLOOR}[1]{\left \lfloor #1 \right \rfloor}

\newcommand{\MAXOFEXPR}[2][]{\max_{#1} \left\{ #2 \right\}}
\newcommand{\MINOFEXPR}[2][]{\min_{#1} \left\{ #2 \right\}}
\newcommand{\Maxofexpr}[2][]{\max_{#1} \bigl\{ #2 \bigr\}}
\newcommand{\Minofexpr}[2][]{\min_{#1} \bigl\{ #2 \bigr\}}
\newcommand{\maxofexpr}[2][]{\max_{#1} \{ #2 \}}

\newcommand{\minofset}[3][:]{\min \{ #2 #1 #3 \}}
 
\newcommand{\MAXOFSET}[3][:]{\ifthenelse{\equal{#1}{;}}{\MAXOFEXPR{ #2 \,;\, #3 }}
     {\ifthenelse{\equal{#1}{:}}{\MAXOFEXPR{ #2 \,:\, #3 }}
     {\max \twincommandJN{\left\{}{#2}{\left#1}{\right}{\,#3}{\right\}}}}}
\newcommand{\MINOFSET}[3][:]{\ifthenelse{\equal{#1}{;}}{\MINOFEXPR{ #2 \,;\, #3 }}
     {\ifthenelse{\equal{#1}{:}}{\MINOFEXPR{ #2 \,:\, #3 }}
     {\min \twincommandJN{\left\{}{#2}{\left#1}{\right}{\,#3}{\right\}}}}}

\newcommand{\Maxofset}[3][:]{\ifthenelse{\equal{#1}{;}}{\Maxofexpr{ #2 \,;\, #3 }}
     {\ifthenelse{\equal{#1}{:}}{\Maxofexpr{ #2 \,:\, #3 }}
     {\max \twincommandJN{\bigl\{}{#2}{\bigl#1}{\bigr}{\,#3}{\bigr\}}}}}
\newcommand{\Minofset}[3][:]{\ifthenelse{\equal{#1}{;}}{\Minofexpr{ #2 \,;\, #3 }}
     {\ifthenelse{\equal{#1}{:}}{\Minofexpr{ #2 \,:\, #3 }}
     {\min \twincommandJN{\bigl\{}{#2}{\bigl#1}{\bigr}{\,#3}{\bigr\}}}}}

\DeclareMathOperator{\Expop}{E}

\newcommand{\PROB}[2][]{\Pr_{#1} \left[ #2 \right]}
\newcommand{\Prob}[2][]{\Pr_{#1} \bigl[ #2 \bigr]}
\newcommand{\prob}[2][]{\Pr_{#1} [ #2 ]}

\newcommand{\Expectation}[2][]{\Expop_{#1} \bigl[ #2 \bigr]}
\newcommand{\expectation}[2][]{\Expop_{#1} [ #2 ]}

\newcommand{\twincommandJN}[6]{#1#2#3\vphantom{#2#5}\mspace{-2.05mu}#4.#5#6}

\newcommand{\funcdescr}[3]{\ensuremath{ #1 : #2 \to #3}}

\newcommand{\domainof}[1]{\ensuremath{\mathrm{dom} ( #1 )}}

\newcommand{\boundary}[1]{\ensuremath{\partial #1}}

\newcommand{\set}[1]{\{ #1 \}}
\newcommand{\Set}[1]{\bigl\{ #1 \bigr\}}

\newcommand{\setdescr}[3][\mid]{\set{ #2 #1 #3 }}
\newcommand{\Setdescr}[3][|]{\ifthenelse{\equal{#1}{;}}{\Set{ #2 \,;\, #3 }}
     {\ifthenelse{\equal{#1}{:}}{\Set{ #2 \,:\, #3 }}
     {\twincommandJN{\bigl\{}{#2\,}{\bigl#1}{\bigr}{\,#3}{\bigr\}}}}}

\newcommand{\Setsize}[1]{\bigl\lvert#1\bigr\rvert}
\newcommand{\setsize}[1]{\lvert#1\rvert}

\newcommand{\intersection}{\cap}

\newcommand{\union}{\cup}

\newcommand{\disjointunion}{\overset{.}{\cup}}

\newcommand{\Lor}{\bigvee}

\newcommand{\olnot}[1]{\overline{#1}}

\newcommand{\complclassformat}[1]{\textrm{\upshape{\textsf{#1}}}\xspace}
\newcommand{\cocomplclass}[1]{\textrm{\upshape{\textsf{co#1}}}\xspace}

\newcommand{\NP}{\complclassformat{NP}}

\newcommand{\coNP}{\cocomplclass{NP}}

\newcommand{\Ppoly}{\complclassformat{P/poly}}

\newcommand{\introduceterm}[1]{{\emph{#1}}}

\newcommand{\eqperiod}{\enspace .}
\newcommand{\eqcomma}{\enspace ,}

\newcommand{\wrt}{with respect to\xspace}

\newcommand{\ie}{i.e.,\ }

\ifthenelse{\isundefined{\st}}
{\newcommand{\st}{such that\xspace}}
{}

\newcommand{\wolog}{without loss of generality\xspace}

\newcommand{\aas}{asymptotically almost surely\xspace}

\newcommand{\refsec}[1]{Section~\ref{#1}}

\newcommand{\reffig}[1]{Figure~\ref{#1}}

\newcommand{\refth}[1]{Theorem~\ref{#1}}

\newcommand{\refthm}[1]{Theorem~\ref{#1}}

\newcommand{\reflem}[1]{Lemma~\ref{#1}}
\newcommand{\reftwolems}[2]{Lemmas~\ref{#1} and~\ref{#2}}

\newcommand{\refcor}[1]{Corollary~\ref{#1}}

\newcommand{\refdef}[1]{Definition~\ref{#1}}

\newcommand{\Reflem}[1]{Lemma~\ref{#1}}

\ifthenelse
{\isundefined{\refeq}}
{\newcommand{\refeq}[1]{\eqref{#1}}}
{\renewcommand{\refeq}[1]{\eqref{#1}}}

\usepackage{ifpdf}

\ifpdf         
\fi

\newcommand{\proofstd}{\pi}
\newcommand{\proofpi}{\pi}
\newcommand{\refpi}{\pi}

\newcommand{\derivof}[4][\derives]
        {{\ensuremath{{#2} : {#3} \, {#1}\, {#4}}}}
\newcommand{\refof}[2]{\derivof{#1}{#2}{\bot}}

\newcommand{\emptycl}{\bot}

\newcommand{\formf}{\ensuremath{F}}

\newcommand{\varx}{\ensuremath{x}}

\newcommand{\clb}{\ensuremath{B}}
\newcommand{\clc}{\ensuremath{C}}
\newcommand{\cld}{\ensuremath{D}}

\newcommand{\setsofvarsorlit}[2]{\mathit{#1}({#2})}

\newcommand{\vars}[1]{\setsofvarsorlit{Vars}{#1}}

\newcommand{\restrrho}{\rho}

\newcommand{\restrict}[2]{{{#1}\!\!\upharpoonright_{#2}}}

\newcommand{\genericformsmall}[2]{\mathit{#1}( #2 )}

\newcommand{\lengthstd}{L}
\newcommand{\length}[1]{\genericformsmall{L}{#1}}

\newcommand{\width}[2][]{\genericformsmall{W_{#1}}{#2}}

\newcommand{\formulaformat}[1]{\mathit{#1}}

\newcommand{\graphfphpnot}[1][G]{\formulaformat{FPHP}({#1})}

\newcommand{\pigeonclause}[1]{P^{#1}}
\newcommand{\holeclause}[3]{H^{#1,#2}_{#3}}
\newcommand{\functionclause}[3]{F^{#1}_{#2,#3}}
\newcommand{\ontoclause}[1]{S_{#1}}

\newcommand{\fakeaxiomintro}{fake axiom\xspace}
\newcommand{\quotepaper}[1]{\emph{``#1''}}

\newcommand{\numpigeons}{m}
\newcommand{\numholes}{n}

\newcommand{\heavy}{heavy\xspace}

\newcommand{\sheavy}{super-heavy\xspace}
\newcommand{\superheavy}{\sheavy}
\newcommand{\light}{light\xspace}

\newcommand{\GFPHP}{G\text{-}FPHP}

\newcommand{\PM}{PM(G)}

\renewcommand{\GFPHP}{\graphfphpnot[G]}

\renewcommand{\PM}{\formulaformat{PM}(G)}

\newcommand{\neigh}[2][]{N_{{#1}}(#2)}

\newcommand{\uniqueNeigh}[2][]{\boundary_{{#1}} ({#2})}
\newcommand{\uniqueneigh}[2][]{\boundary_{{#1}} ({#2})}

\renewcommand{\deg}[2][XXXXXXXXXXXX]{\Delta_{#1}(#2)}

\newcommand{\diff}{\xi}
\newcommand{\cdiff}{4}
\newcommand{\ontocdiff}{64}

\newcommand{\dimLi}{\deg[G]{i} - d_i + \delta_i/4}

\newcommand{\ontodimLi}{1/2(\deg[G]{v} - d_v + \delta_v/2)}

\newcommand{\GdegMax}{\Delta}
\newcommand{\gdegmax}{\GdegMax}

\newcommand{\Gdist}{\mathcal{G}}

\newcommand{\constFPHPnew}{16}

\newcommand{\constonto}{128}

\newcommand{\V}[1][]{V_{#1}}

\newcommand{\Vl}{\V[L]}
\newcommand{\Vr}{\V[R]}

\newcommand{\Vp}{\V[P]}
\newcommand{\Vh}{\V[H]}

\newcommand{\Gprime}{G'}
\newcommand{\GprimeConstantDeg}{8}
\newcommand{\GprimeConstantExp}{12}

\newcommand{\Vbad}[2][]{\widetilde V^{#1}(#2)}

\newcommand{\VbadThick}[1]{\overline V(#1)}

\newcommand{\exponentalpha}{\varepsilon}

\newcommand{\Ax}{\ensuremath\mathcal{A}}

\newcommand{\thr}[1]{ \ensuremath d_{{#1}} }
\newcommand{\vecthr}{\vec{d}}
\newcommand{\vecthrdef}[1]{\vecthr = \left(\thr{1}, \ldots, \thr{#1}\right)}

\newcommand{\adv}[1]{\delta_{{#1}}}
\newcommand{\vecadv}{\vec{\delta}}
\newcommand{\vecadvdef}[1]{\vecadv = \left(\adv{1}, \ldots, \adv{#1}\right)}

\newcommand{\Pthick}[1]{\ensuremath P_{\vecthr,\vecadv}(#1)}
\newcommand{\Pfat}[1]{\ensuremath P_{\vecthr}(#1)}

\newcommand{\linspace}[1][]{\ensuremath \Lambda_{{#1}}}

\newcommand{\linmap}[1][]{\ensuremath \lambda_{{#1}}}
\newcommand{\linmapof}[2][]{\linmap[#1](#2)}

\newcommand{\linmapl}[2][]{\ensuremath \lambda^{{#2}}_{{#1}}}
\newcommand{\linmaplof}[3][]{\linmapl[#1]{#2}(#3)}

\newcommand{\Vfat}[1]{\ensuremath V_{\vecthr}(#1)}
\newcommand{\Vthick}[1]{\ensuremath V_{\vecthr,\vecadv}(#1)}

\newcommand{\fakeaxiom}{$\bigl( w_0, \vecthr \bigr)$-axiom\xspace}

\newcommand{\pseudowidth}{pseudo-width\xspace}

\newcommand{\PSEUDOWIDTH}{Pseudo-Width\xspace}

\renewcommand{\width}[1]{\ensuremath w_{\vecthr,\vecadv}( #1 )}

\newcommand{\cgamma}{\ensuremath \gamma}
\newcommand{\cgammap}{\ensuremath \gamma'}
\newcommand{\indexvec}{\ell}
\newcommand{\nvecs}{L}

\newcommand{\ri}[1]{r_{#1}}
\newcommand{\riof}[2]{r_{#2}(#1)}
\newcommand{\vecr}{\vec{r}}
\newcommand{\vecrof}[1]{\vec{r}(#1)}

\newcommand{\vecrdef}[1]{\vecr = \left(\ri{1}, \ldots, \ri{#1}\right)}

\newcommand{\vecrofdef}[2]{\vecrof{#1} = 
\left(\riof{#1}{1}, \ldots, \riof{#1}{#2}\right)}

\newcommand{\limval}{t}
\newcommand{\distfilter}{\mu}

\newcommand{\vecrb}{\pmb{\vecr}}
\newcommand{\rb}{\pmb{r}}
\newcommand{\rbi}[1]{\pmb{r_#1}}

\newcommand{\goodset}[2][\vecrb]{P_{#1}(#2)}
\newcommand{\betterset}[2][\vecrb + 1]{P_{#1}(#2)}
\renewcommand{\betterset}[2][\vecrb]{Q_{#1}(#2)}

\newcommand{\contained}[3]{$(#1, #2, #3)$-contained}
\newcommand{\closure}[2][]{\mathsf{closure}_{{#1}}(#2)}

\newcommand{\expansionconst}{c}

\newcommand{\almostcontained}{\nu}
\newcommand{\setsizemax}{k}

\newcommand{\diffclosure}[1][]{D_{{#1}}}
\renewcommand{\diffclosure}[1][]{\mathcal{D}_{{#1}}}

\newcommand{\matchbool}[1]{\rho_{#1}}
\newcommand{\matchboolof}[2]{\rho_{#1}(#2)}

\newcommand{\vmatch}[1]{V(#1)}

\newcommand{\matcha}{\varphi}
\newcommand{\matchaof}[1]{\matcha(#1)}

\newcommand{\matchb}{\varphi'}

\newcommand{\matchc}{\psi}

\newcommand{\matchd}{\psi'}

\newcommand{\matche}{\eta}

\newcommand{\matchf}{\eta'}

\newcommand{\matchings}[1][]{\ensuremath \mathcal{M}_{{#1}}}

\newcommand{\Null}[1]{\ensuremath Z(#1)}

\newcommand{\spans}[1]{\mathrm{span} ( #1 )}
\newcommand{\dom}[1]{\mathrm{dom} ( #1 )}

\newcommand{\tensor}{\ensuremath \otimes}
\newcommand{\bigtensor}{\ensuremath \bigotimes}

\newcommand{\setalg}{T}

\newcommand{\Xrand}{\boldsymbol{X}}

\newcommand{\sample}[2]{#1\sim#2}

\newcommand{\lengthsize}{length\xspace}

\usepackage{numname}

\definecolor{airforceblue}{rgb}{0.36, 0.54, 0.66}
\definecolor{amethyst}{rgb}{0.6, 0.4, 0.8}
\definecolor{asparagus}{rgb}{0.53, 0.66, 0.42}
\definecolor{brass}{rgb}{0.71, 0.65, 0.26}
\definecolor{brown}{rgb}{0.59, 0.29, 0.0}
\definecolor{darkolivegreen}{rgb}{0.33, 0.42, 0.18}
\definecolor{darkorange}{rgb}{1.0, 0.55, 0.0}

\newcommand{\squeezesubsection}[1]{\subsection{#1}}

\numberwithin{equation}{section}

\makeatletter
\newcommand\IfRestateTF{%
  \ifx\label\thmt@gobble@label 
    \expandafter\@firstoftwo
  \else
    \expandafter\@secondoftwo
  \fi
}
\makeatother
\newcommand{\RestateRemark}{\IfRestateTF{{\normalfont\bfseries (Restated) }}{}}

\title{Exponential 
  Resolution
  Lower Bounds for Weak Pigeonhole Principle 
  and Perfect Matching Formulas over Sparse Graphs}

\ThCSauthor[Lund]{Susanna F.~de Rezende}{susanna.rezende@cs.lth.se}[0000-0001-8923-1240]
\ThCSauthor[CopenhagenLund]{Jakob Nordström}{jn@di.ku.dk}[0000-0002-2700-4285]
\ThCSauthor[EPFL]{Kilian Risse}{rissekil@gmail.com}[0000-0002-6913-3341]
\ThCSauthor[EPFL]{Dmitry Sokolov}{sokolov.dmt@gmail.com}[0000-0003-2809-3467]
\ThCSaffil[Lund]{Lund University, Sweden}
\ThCSaffil[CopenhagenLund]{University of Copenhagen, Denmark, and Lund University, Sweden}
\ThCSaffil[EPFL]{EPFL, Lausanne, Switzerland}
\ThCSthanks{A preliminary version of this article appeared at CCC 2020 \cite{RezendeNR020}. 
The authors were funded by 
the Knut and Alice Wallenberg grant
\mbox{KAW 2016.0066}.
In addition, the second author
was supported by
the  Swedish Research Council grant
\mbox{2016-00782}
and the Independent Research Fund Denmark (DFF)  grant
\mbox{9040-00389B},
and the fourth author by
the Knut and Alice Wallenberg grant
\mbox{KAW 2016.0433}.
Most of this work was carried out while the authors were affiliated to
KTH Royal Institute of Technology. Some parts were done while visiting
the Simons Institute for the Theory of Computing in association with
the DIMACS/Simons Collaboration on Lower Bounds in Computational
Complexity, which is conducted with support from the National Science
Foundation.}
\ThCSshortnames{S.\ de Rezende, J.\ Nordström, K.\ Risse, D.\ Sokolov}  
\ThCSshorttitle{Exponential 
  Resolution
  Lower Bounds for Weak Pigeonhole Principle}
\ThCSyear{2025}
\ThCSarticlenum{9}
\ThCSreceived{Feb 19, 2024}
\ThCSrevised{Dec 12, 2024}
\ThCSaccepted{Jan 20, 2025}
\ThCSpublished{Mar 25, 2025}
\ThCSdoicreatedtrue
\ThCSkeywords{Computational Complexity, Proof Complexity, Resolution, Pigeonhole Principle}

\begin{document}
\maketitle





\begin{abstract}
  We show exponential lower bounds on resolution proof length for
  pigeonhole principle (PHP) formulas and perfect matching formulas
  over highly unbalanced, sparse expander graphs, thus answering the
  challenge to establish strong lower bounds in the regime between
  balanced constant-degree expanders as in [Ben-Sasson and
  Wigderson~'01] and highly unbalanced, dense graphs as in [Raz~'04]
  and [Razborov~'03,~'04].  We obtain our results by revisiting
  Razborov's pseudo-width method for PHP formulas over dense graphs
  and extending it to sparse graphs. This further demonstrates the
  power of the pseudo-width method, and we believe it could
  potentially be useful for attacking also other longstanding open
  problems for resolution and other proof systems.
\end{abstract}

\section{Introduction}
\label{sec:intro}

In one sentence, 
proof complexity 
is the study
of efficient certificates of unsatisfiability for formulas in
conjunctive normal form (CNF).  In its most general form, this is the
question of whether \coNP~can be separated from \NP~or not, and as
such appears out of reach for current techniques. 
However, if one instead focuses on concrete proof systems, which can
be thought of as restricted models of nondeterministic computation,
this opens up the view to a rich landscape of results.

One line of research in proof complexity has been to prove
superpolynomial lower bounds for stronger and stronger proof systems,
as a way of approaching 
the distant goal of establishing
$\NP \neq \coNP$. 
A perhaps even more fruitful direction, however, has been to study
different combinatorial principles and investigate what kind of
reasoning is needed to efficiently establish the validity of these
principles.  In this way, one can quantify the ``depth'' of
different mathematical truths, measured in terms of how strong a proof
system is required to prove them.

In this paper, we consider the proof system
\introduceterm{resolution}~\cite{Blake37Thesis}, in which one derives
new disjunctive clauses 
from the
formula until an
explicit contradiction is reached.
This is arguably the most well-studied proof system
in proof complexity, for which numerous exponential lower
bounds on proof size have been shown (starting with
\cite{CS88ManyHard,Haken85Intractability,Urquhart87HardExamples}).
Yet many basic questions about resolution remain stubbornly open.
One such set of questions concerns the
\introduceterm{pigeonhole principle (PHP)}
stating that there is no injective mapping of $m$ pigeons into $n$
holes if $m > n$.  This is one of the simplest, 
and yet most useful, combinatorial principles in mathematics, and it
has been topic of extensive study in proof complexity with deep
connections to TFNP classes.

When studying the pigeonhole principle,
it is convenient to think of it in terms of a bipartite graph
$G = (U \disjointunion  V, E)$  with 
pigeons $U = [m]$ and holes $V = [n]$ for
$m \geq n+1$.
Every pigeon~$i$ can fly to its neighbouring pigeonholes~$N(i)$ as
specified 
by~$G$,
which for now we 
fix to be the 
complete bipartite graph~$K_{m,n}$ with $\neigh{i} = [n]$ for all $i \in [m]$.
Since we wish to study unsatisfiable formulas, 
we encode the claim that there does in fact exist an injective
mapping of pigeons to holes as a CNF formula consisting of
\introduceterm{pigeon axioms}
\begin{subequations}
\begin{align}
  \label{eq:axiom-pigeon}
  \pigeonclause{i}
  &= \Lor_{j \in 
    \neigh{i}
    } x_{ij} 
  &&\text{for } i \in [m]
  \\
\shortintertext{and \introduceterm{hole axioms}}
  \label{eq:axiom-hole}
  \holeclause{i}{i'}{j}
  &= (\olnot{x}_{ij} \vee \olnot{x}_{i'j}) 
  &&\text{for }i \neq i' \in [m], j \in \neigh{i} \intersection \neigh{i'} 
  \\
\intertext{(where the intended meaning of the variables is that
  $x_{i,j}$ is true if pigeon~$i$ flies to hole~$j$).
  To rule out multi-valued mappings one can also add
  \introduceterm{functionality axioms}
}
  \label{eq:axiom-functionality}
  \functionclause{i}{j}{j'}
    &= (\olnot{x}_{ij} \vee \olnot{x}_{ij'}) 
  &&\text{for }i \in [m], j \neq j' \in \neigh{i} \eqcomma\\
  \intertext{and a further restriction is to include
  \introduceterm{surjectivity} or
  \introduceterm{onto axioms}
}
  \label{eq:axiom-onto}
  \ontoclause{j}
  &= \Lor_{i \in 
    \neigh{j}
    } x_{ij} 
  &&\text{for } j \in [n] 
\end{align}
\end{subequations}
requiring that every hole should get a pigeon.
Clearly, the ``basic''
\introduceterm{pigeonhole principle (PHP) formulas}
with clauses
\refeq{eq:axiom-pigeon} and~\refeq{eq:axiom-hole}
are the least constrained.
As one adds clauses~\refeq{eq:axiom-functionality} to obtain
the
\introduceterm{functional pigeonhole principle (FPHP)}
and also clauses~\refeq{eq:axiom-onto} to get
the
\introduceterm{onto functional pigeonhole principle (onto-FPHP)},
the formulas become more overconstrained and thus (potentially) easier
to disprove, meaning that establishing lower bounds becomes harder.
A moment of reflection reveals that onto-FPHP formulas are just saying
that complete bipartite graphs with $m$~left vertices and
$n$~right vertices have perfect matchings, and so these formulas are
also referred to as
\introduceterm{perfect matching formulas}.

Another way of varying the hardness of PHP formulas is  by letting the
number of pigeons~$m$ grow larger as a function of the number of
holes~$n$.  
What this means is that it is not necessary to
count exactly to refute the formulas. Instead, it is sufficient 
to provide
a precise enough estimate 
to show that $m>n$
must hold (where the hardness of this task depends on how much larger
$m$ is than~$n$).
Studying the hardness of such so-called
\introduceterm{weak PHP formulas}
gives a way of measuring how good different proof systems
are at approximate counting.
A second
application of  lower bounds
for weak PHP  formulas is that they  can be used to show that proof
systems cannot produce efficient proofs of the claim that
\mbox{$\NP \nsubseteq \Ppoly$
\cite{Razborov98LowerBound,Razborov04ResolutionLowerBoundsPM}.}

Yet another version of more constrained formulas is obtained by
restricting what choices the pigeons have for flying into holes, by
defining the formulas not over~$K_{m,n}$ but
sparse bipartite
graphs with bounded left degree---such instances are usually called
\introduceterm{graph PHP formulas}.
Again, this makes the formulas easier to disprove in the sense that
pigeons are more constrained, and it also removes the symmetry in the
formulas that plays an essential role in many lower bound proofs.

Our work focuses on the most challenging setting in terms of lower
bounds, when all of these restrictions apply: 
the PHP formulas contain both functionality and onto axioms,
the number of  pigeons~$m$ is very large
compared to the number of holes~$n$,  
and the choices of holes are restricted by a sparse graph. 
But before discussing our contributions, let us review what has been
known about resolution and pigeonhole principle formulas.
We emphasize that what will follow is a brief and selective overview
focusing on resolution only---see Razborov's beautiful survey
paper~\cite{Razborov02ProofComplexityPHP} 
for a discussion of upper and lower bounds on PHP formulas in other
proof systems.

\squeezesubsection{Previous Work}
   
In a breakthrough result, which
served as a strong impetus for further developments in proof complexity,
Haken~\cite{Haken85Intractability} proved a lower bound
$\exp ( \bigomega{n}) $
on resolution proof length for $m=n+1$ pigeons.
Haken's proof was for the basic PHP formulas, but easily extends to
onto-FPHP formulas. 
This result was simplified and improved in a sequence of works
\cite{BP96Simplified,BW01ShortProofs,BT88ResolutionProofs,Urq03}
to a lower bound of the form
$\exp \bigl( n^2 / m \bigr)$,
which, unfortunately, does not yield anything nontrivial for 
$m = \Bigomega{n^2}$ pigeons.

Buss and Pitassi~\cite{BP97Weak} showed that the pigeonhole principle
does in fact get easier for resolution when $m$ becomes sufficiently
large: namely, for 
$m = \exp \bigl( \Bigomega{\sqrt{n \log n}} \bigr)$
PHP formulas can be refuted in length
$\exp \bigl( \Bigoh{ \sqrt{n \log n} } \bigr)$.
This is in contrast to 
what holds for the  
weaker subsystem
\introduceterm{tree-like resolution},
for which 
the
formulas remain equally hard as the number of pigeons
increases, and where the 
complexity was even sharpened in  
\cite{BGL10,BP97Weak,Dant02,DantRiis01}
to an
$\exp ( \bigomega{ n \log n} )$ 
length
lower bound.

Obtaining lower bounds beyond $m = n^2$~pigeons for non-tree-like
resolution turned out to be quite challenging.
Haken's bottleneck counting method fundamentally breaks down when the
number of pigeons is quadratic in the number of holes, and the same
holds for the celebrated length-width lower bound
in~\cite{BW01ShortProofs}. 
Some progress was made for 
restricted forms of
resolution
in~\cite{RazbWigYao02} and~\cite{PitRaz04},
leading up to an
$\exp \bigl( n^\varepsilon \bigr)$
lower bound for so-called \introduceterm{regular resolution}.
In a technical tour de force, 
Raz~\cite{Raz04Resolution} finally proved 
that general, unrestricted resolution requires length
$\exp \bigl( n^\varepsilon \bigr)$
to refute the basic PHP formulas even with arbitrary many pigeons.
Razborov followed up on this in three papers where he first 
simplified and slightly strengthened 
Raz's result in~\cite{Razborov01ImprovedResolutionLowerBoundsWPHP},
then extended it to FPHP formulas
in~\cite{Razborov03ResolutionLowerBoundsWFPHP} 
and lastly established an analogous lower bound for onto-FPHP formulas
in~\cite{Razborov04ResolutionLowerBoundsPM}.

More precisely, what Razborov showed is that for any version of the PHP
formula with $m$~pigeons and $n$~holes, the minimal proof length
required in resolution is
$
\exp \bigl( \Bigomega{n / \log^2 m } \bigr)
$.
It is easy to see that this implies a lower bound
$
\exp \bigl( \Bigomega{\sqrt[3]{n} } \bigr)
$
for any number of pigeons---for 
$m = \exp \bigl( \Bigoh{\sqrt[3]{n}} \bigr)$
we can appeal directly to the bound above, and if a resolution proof
would use 
$\exp\bigl(\Bigomega{\sqrt[3]{n}} \bigr)$~pigeons, 
then just mentioning all these
different pigeons already requires
$\exp\bigl( \Bigomega{ \sqrt[3]{n} } \bigr)$ distinct clauses.
It is also clear that considering complexity in terms of the number of
holes~$n$ is the right measure. Since any formula contains a basic PHP
subformula with $n+1$~pigeons that can be refuted in
length~$\exp(\bigoh{n})$,
we can never hope for exponential lower bounds in terms of formula
size as the number of pigeons~$m$ grows to exponential.

So far we have stated results only for the standard PHP formulas
over~$K_{m,n}$, where any pigeon can fly to any hole. 
However, the way Ben-Sasson and Wigderson~\cite{BW01ShortProofs}
obtained their result was by considering graph PHP formulas over balanced
bipartite expander graphs of constant left degree, from which the
lower bound for~$K_{m,n}$ easily follows by a restriction argument.
It was shown in~\cite{IOSS16} that an analogous bound holds for
onto-FPHP formulas, \ie perfect matching formulas, on bipartite expanders.
In this context is also relevant to mention the exponential lower
bounds in~\cite{Alekhnovich04Mutilated,DR01Planar} on 
\introduceterm{mutilated chessboard formulas}, which can be
viewed as perfect matching formulas on balanced, sparse bipartite
graphs with very bad expansion.
At the other end of the spectrum, Razborov's PHP lower bound
in~\cite{Razborov04ResolutionLowerBoundsPM} for highly unbalanced
bipartite graphs also applies in a more general  setting
than~$K_{m,n}$: namely, for 
any graph where the minimal degree of any left vertex
is~$\delta$, the minimal length of any resolution proof is
$\exp \bigl( \Bigomega{\delta / \log^2 m} \bigr)$.
Thus, for graph PHP formulas we have exponential lower bounds
on the one hand~\cite{BW01ShortProofs}
for $m \ll n^2$ pigeons, where each pigeon is
adjacent to a constant number of holes,
and on the other
hand~\cite{Razborov04ResolutionLowerBoundsPM} for
any number of pigeons
given that each pigeon is adjacent to a polynomial~$n^{\bigomega{1}}$
number of holes,
but nothing has been known in between these extremes.
In~\cite{Razborov04ResolutionLowerBoundsPM}, Razborov asks whether a
\quotepaper{common generalization} of the techniques in
\cite{BW01ShortProofs}
and
\cite{Razborov04ResolutionLowerBoundsPM,Razborov03ResolutionLowerBoundsWFPHP}
can be found
\quotepaper{that would uniformly cover both cases?}
Urquhart~\cite{Urquhart07WidthVsSize} also discusses Razborov's lower
bound technique, but notes that
\quotepaper{the search for a yet more general point of view remains a topic for
further research.}

\squeezesubsection{Our Results}

In this work, we give an answer to the questions raised
in~\cite{Razborov04ResolutionLowerBoundsPM,Urquhart07WidthVsSize} by
presenting a general technique that applies for any number of
pigeons~$m$ all the way from linear to weakly exponential, and that
establishes exponential lower bounds on resolution proof length for
all flavours of graph PHP formulas (including perfect matching
formulas) even over sparse graphs.

Let us state below three examples of the kind of lower bounds we
obtain---the full, formal statements will follow in later sections.
Our first theorem is an average-case lower bound for 
onto-FPHP formulas with slightly superpolynomial number of pigeons.

\begin{theorem}[Informal]
  \label{th:informal-quasipoly-npigeons}  
  Let
  $G$ be a randomly sampled bipartite graph with
  $n$~right vertices, \mbox{$m = n^{\littleoh{\log n}}$} left vertices, and left
  degree
  $\Bigtheta{\log^2 m}$.
  Then refuting
  the onto-FPHP formula (a.k.a.~perfect matching formula) over~$G$
  in resolution
  requires length
  $\exp \bigl( \Bigomega{n^{1-\littleoh{1}}} \bigr)$ 
  asymptotically almost surely.
\end{theorem}

Note that as the number of pigeons grow larger, it is clear that
the left degree also has to grow---otherwise we will get a small
number of pigeons constrained to fly to a small number of holes by a
birthday paradox argument, 
yielding
a small unsatisfiable
subformula that can easily be refuted by brute force.

If the number of pigeons increases further to weakly exponential, then
randomly sampled graphs no longer have good enough expansion for our
technique to work, but there are explicit constructions of unbalanced
expanders for which we can still get lower bounds.

\begin{theorem}[Informal]
  \label{th:informal-exponential-npigeons}
  There are explicitly constructible bipartite graphs~$G$ with
  $n$~right vertices, \mbox{$m = \exp \bigl( \Bigoh{n^{1/16}} \bigr)$}
  left vertices, and left degree $\Bigtheta{\log^4 m}$ such that
  resolution requires length
  $\exp \bigl( \Bigomega{n^{1/8 - \exponentalpha}} \bigr)$ to refute
  the perfect matching formula over~$G$.
\end{theorem}

Finally, for functional pigeonhole principle formulas we can also
prove an exponential lower bound for \emph{constant} left degree even if the
number of pigeons is a large polynomial.

\begin{theorem}[Informal]
  \label{th:informal-poly-npigeons}
  Let
  $G$ be a randomly sampled bipartite graph with
  $n$~right vertices, \mbox{$m = n^{k}$} left vertices, and left
  degree
  $\Bigtheta{(k / \varepsilon)^2}$.
  Then refuting the functional pigeonhole principle formula
  over~$G$
  in resolution
  requires length
  $\exp \bigl( \Bigomega{n^{1-\varepsilon}} \bigr)$ 
  asymptotically almost surely.
   \end{theorem}

\squeezesubsection{Techniques}
\label{sec:intro_techniques}

At a very high level, what we do in terms of techniques is to revisit
the pseudo-width method introduced by Razborov for functional PHP
formulas in~\cite{Razborov03ResolutionLowerBoundsWFPHP}.  We
strengthen this method to work in the setting of sparse graphs by
combining it with the closure operation on expander
graphs in~\cite{ABRW04Pseudorandom,AR03LowerBounds}, which is a way to
restore expansion after a small set of (potentially adversarially
chosen) vertices have been removed.
To extend the results further to perfect matching formulas, we apply a
``preprocessing step'' on the formulas as
in~\cite{Razborov04ResolutionLowerBoundsPM}. 
In what remains of this section, we focus on graph FPHP formulas and
give an informal overview of the lower bound proof in this setting,
which already contains most of the interesting ideas
(although the extension to onto-FPHP also raises significant
additional challenges).

Let
$\GFPHP$
denote the functional 
pigeonhole principle
formula over the graph~$G$ consisting of 
\mbox{clauses
  \refeq{eq:axiom-pigeon}--\refeq{eq:axiom-functionality}.}
A first, quite naive (and incorrect), description of the proof
structure is that we start by defining a 
\introduceterm{\pseudowidth{}} measure on clauses~$\clc$ that counts
pigeons~$i$ that 
appear  in~$\clc$ in many variables~$x_{ij}$
for distinct~$j$.
We then show that
any short resolution refutation of
$\GFPHP$ can be transformed into a refutation where all clauses have
small \pseudowidth. By a separate argument, 
we establish that any refutation of
$\GFPHP$ 
requires large \pseudowidth.
Hence,  no short refutations can exist, which is precisely what we
were aiming to prove.

To fill in the details (and correct) this argument, let us start by
making clear what we mean by \pseudowidth. 
Suppose that the graph~$G$ has left degree~$\gdegmax$.
In what follows, we 
identify a mapping of pigeon~$i$ to a neighbouring hole~$j$ with the partial
assignment~$\restrrho$ such that
$\restrrho (x_{i,j}) = 1$ and
$\restrrho (x_{i,j'}) = 0$ for all $j' \in \neigh{i} \setminus \set{j}$.
We denote by 
$\thr{i}(\clc)$ 
the number of mappings of pigeon $i$ that
satisfy $\clc$.
Note that if
$\clc$ contains at least
one negated literal~$\olnot{x}_{i,j}$,
then
$\thr{i}(\clc) \geq \gdegmax - 1$,
and otherwise
$\thr{i}(\clc)$
is the number of positive literals
$x_{i,j}$ for $j \in \neigh{i}$.
Given a judiciously  chosen
``filter vector''
$\vecthr =(\thr{1}, \ldots, \thr{m})$
for
$\thr{i} \approx \gdegmax$
and a ``slack''
$\delta \approx \gdegmax / \log m$,
we 
say that
pigeon $i$ is \introduceterm{\heavy{}} in 
$\clc$
if $\thr{i}(\clc) \ge \thr{i} - \delta$ and
\introduceterm{\superheavy{}} 
\mbox{if  $\thr{i}(\clc) \ge \thr{i}$.}
We define the \introduceterm{\pseudowidth{}} of 
a clause~$\clc$ to be
the number of \heavy pigeons in $\clc$.

With these definitions in hand, we can give a description of the
actual proof:
\begin{enumerate}
\item 
\label{item:overview-filter-lemma}
  Given any resolution refutation~$\proofpi$ of~$\GFPHP$ in small
  length~$\lengthstd$, we argue that all clauses can be classified as
  having either low or high \pseudowidth, where 
  an important additional guarantee is that 
  the high-width clauses not
  only have many \heavy pigeons but actually many \superheavy pigeons.
  
\item
  \label{item:overview-substitution}
  We replace all clauses~$\clc$ with many \superheavy pigeons with
  ``\fakeaxiomintro{}s''
  $\clc' \subseteq \clc$
  obtained by throwing away literals from~$\clc$
  until we have nothing left but a medium number of \superheavy pigeons.
  By construction, the set $\Ax$ of such \fakeaxiomintro{}s is of size
  $\setsize{\Ax} \leq \lengthstd$, and after making the replacement we
  have a resolution refutation $\proofpi'$ of~$\GFPHP \union \Ax$
  in low \pseudowidth.
  
\item
\label{item:overview-pseudowidth-lower-bound}
  However, since  $\Ax$ is not too large, we are able to show that
  any resolution refutation of
  $\GFPHP \union \Ax$
  must still require   large pseudo-width. Hence, $\lengthstd$ cannot
  be small, and the lower bound follows.
\end{enumerate}

Part~\ref{item:overview-filter-lemma}
is similar to~\cite{Razborov03ResolutionLowerBoundsWFPHP},
but with a slight twist.
We show that if the length of~$\proofpi$ is  $\lengthstd < 2^{w_0}$
and if we choose
$
\delta \leq 
\varepsilon
\gdegmax \log n / \log m
$, 
then there exists  a vector
$\vecthr = (\thr{1}, \ldots, \thr{m})$
such that for all clauses 
in~$\proofpi$  either the number of \superheavy pigeons is
at least $w_0$
or else the number of \heavy pigeons is at 
\mbox{most
$
\Bigoh{w_0 \cdot n^{\varepsilon}}        
$.}
The proof of this is by sampling the coordinates~$\thr{i}$
independently from a suitable probability distribution and then
applying a union bound argument.
Once this has been established,
part~\ref{item:overview-substitution}
follows easily:
we just replace all clauses with at least
$w_0$ \superheavy pigeons by  
(stronger) \fakeaxiomintro{}s.
Including all \fakeaxiomintro{}s~$\Ax$ yields a
refutation~$\proofpi'$ of
$ \GFPHP \union \Ax $ (since we can add a weakening rule
deriving
$\clc$ from   $\clc' \subseteq \clc$
to resolution \wolog)
and clearly all clauses 
in~$\proofpi'$
have \pseudowidth
$\Bigoh{w_0 \cdot n^\varepsilon}$.

Part~\ref{item:overview-pseudowidth-lower-bound} is where most of the
hard work is.
Suppose that
$G$ is an excellent 
expander graph, so that for some value $r$ 
all left vertex sets $U'$ of size $\Setsize{U'} \leq r$
have at least
\mbox{$
(1 - \varepsilon\, {\log n}/{\log m})\gdegmax \setsize{U'}
$}
unique neighbours on the right-hand side.
We show that, under the assumptions above, refuting
$\GFPHP \union \Ax$
requires pseudo-width
$
\Bigomega{
  {r \cdot \log n} / {\log m}
}
$.
Tuning the parameters appropriately, this yields a contradiction with
part~\ref{item:overview-substitution}.

Before outlining how the proof of
part~\ref{item:overview-pseudowidth-lower-bound} goes,
we remark that the requirements we place on the expansion of~$G$ are
quite severe. Clearly, any left vertex set~$U$ can have at most
$\gdegmax \setsize{U'}$ neighbours in total, and we are asking for
all except a vanishingly small fraction of these neighbours to be
unique. This is why we can establish
\refth{th:informal-quasipoly-npigeons}  
but not
\refth{th:informal-exponential-npigeons}
for randomly sampled graphs. We see no reason to believe that the
latter theorem would not hold also for random graphs, 
but the expansion properties required for our
proof
are so stringent that they are not satisfied 
in this parameter regime.  This seems to be a fundamental
shortcoming of our technique, and it appears that new ideas would be
required to circumvent this problem.

In order to argue that refuting
$\GFPHP \union \Ax$
in resolution requires large pseudo-width, we
want to
estimate how much
progress the resolution derivation has made up to the point when it
derives some clause~$\clc$. Following Razborov's lead, we measure this
by looking at what fraction of partial matchings of all the \heavy
pigeons in~$\clc$ do not satisfy~$\clc$ (meaning, intuitively, that
the derivation has managed to rule out this part of the search space).
It is immediate by inspection that 
all pigeons mentioned in the
real
axiom clauses~\refeq{eq:axiom-pigeon}--\refeq{eq:axiom-functionality}
are \heavy, 
and any matching of such pigeons satisfies the clauses. Thus, 
the original axioms in~$\GFPHP$ do not rule out any matchings.
Also, it is easy to show that \fakeaxiomintro{}s rule out only an
exponentially small fraction of matchings, since they contain many
\superheavy pigeons and it is hard to match all of these pigeons
without satisfying the clause.
However, the contradictory empty clause~$\emptycl$ rules out
100\% of partial matchings, since it contains no \heavy pigeons to
match in the first place.

What we would like to prove now is that for any 
derivation
in small \pseudowidth it holds that
the derived clause cannot rule out any matching other than those
already 
eliminated
by the clauses used to derive it.
This means that the \fakeaxiomintro{}s together need to rule out all
partial matchings,
but since every \fakeaxiomintro contributes only an exponentially
small fraction 
they are too few
to achieve this. Hence, it
is not
possible to derive contradiction in small \pseudowidth,
which completes 
part~\ref{item:overview-pseudowidth-lower-bound} 
of our proof outline.

There is one problem, however: the last claim above is not true, and
so what is outlined above is only a fake proof. While we have to defer
the discussion of what the full proof actually looks like in detail,
we conclude this section by attempting to hint at a couple of
technical issues and how to resolve them.

Firstly, it does not hold that a derived clause~$\clc$ eliminates only
those matchings that are also forbidden by one of the predecessor clauses
used to derive~$\clc$.  The issue is that a pigeon~$i$ that is \heavy
in both predecessors might cease to be \heavy in~$\clc$---for
instance, if $\clc$ was derived by a resolution step over a
variable~$x_{i,j}$. If this is so, then we would need to show that any
matching of the \heavy pigeons in~$\clc$ can be extended to match
also pigeon~$i$ to any of its neighbouring holes without satisfying
both predecessor clauses. But this will not be true, because a
non-\heavy pigeon can still have some variable~$x_{i,j}$ occurring in both
predecessors. The solution to this, introduced
in~\cite{Razborov03ResolutionLowerBoundsWFPHP}, is to do a  
``lossy counting'' of matchings by associating each partial matching
with a linear subspace of some suitable vector space, and then 
to consider
the span of all matchings ruled out
by~$\clc$. When we accumulate a ``large enough'' number of matchings for a
pigeon~$i$, then the whole subspace associated to~$i$ is spanned and
we can stop counting.

But this leads to a second problem: when studying matchings of the
\heavy pigeons in~$\clc$ we might already have assigned pigeons
$i'_1, \ldots, i'_w$ to occupy holes where pigeon~$i$ might want to fly.  For standard PHP
formulas over complete bipartite graphs this is not a problem, since
at least $n - w$ holes are still available and this number is ``large
enough'' in the sense described above. But for a sparse graph it will
typically be the case that $w \gg \gdegmax$, and so it might well be
the case that
pigeons $i'_1, \ldots, i'_w$ 
are already occupying all the~$\gdegmax$ holes available for pigeon~$i$
according to~$G$. Although it is perhaps hard to see from our
(admittedly somewhat informal) discussion, this turns out to be a very
serious problem, and indeed it is one of the main technical challenges
we need to overcome.

To address this problem we consider not only the \heavy pigeons
in~$\clc$, but also any other pigeons in~$G$ that risk becoming far too
constrained when the 
\heavy pigeons of~$\clc$ 
are matched.
Inspired by~\cite{ABRW04Pseudorandom,AR03LowerBounds},
we define the \introduceterm{closure} to be a
superset~$S$ of the \heavy pigeons such that when $S$ and the
neighbouring holes of~$S$ are removed it holds that the residual graph is
still guaranteed to be a good expander. Provided that $G$ is an
excellent expander to begin with, and that the number of
\heavy pigeons in~$\clc$ is not too large, it can then be shown that an
analogue of the original argument outlined above goes through.

\squeezesubsection{Outline of This Paper}

We review the necessary preliminaries in 
\refsec{sec:prelims}
and introduce two crucial technical tools in
\refsec{sec:filterandclosure}.
The lower bounds for weak graph FPHP formulas are then presented in
\refsec{sec:GFPHP},
after which the perfect matching lower bounds follow in
\refsec{sec:pm}.
We conclude with a discussion of questions for future research in
\refsec{sec:conclusion}.

\section{Preliminaries}
\label{sec:prelims}

We denote natural logarithms (base $\mathrm{e}$) by $\ln$, and base
$2$ logarithms by $\log$.  For positive integers $n \in \Nplus$ we
write $[n] = \set{1,\ldots,n}$.

\subsection{Proof Complexity}

A \introduceterm{literal} over a Boolean variable $\varx$ is either
the variable $\varx$ itself (a \introduceterm{positive literal}) or
its negation $\olnot{\varx}$ (a
\introduceterm{negative literal}).
A \introduceterm{clause}
$\clc = \ell_1 \lor \formuladots \lor \ell_{w}$ is a disjunction
of literals.
We write $\emptycl$ to denote the empty clause without any literals.
A \introduceterm{CNF formula}
$\formf = \clc_1 \land \formuladots \land \clc_m$ is a conjunction
of clauses.
We think of clauses and CNF formulas as sets:
order is irrelevant and there are no repetitions.
We let $\vars{\formf}$ denote the set of variables of~$\formf$.

A \introduceterm{resolution  refutation} $\proofstd$ of an unsatisfiable CNF formula~$\formf$,
or
\introduceterm{resolution proof} for (the unsatisfiability
of)~$\formf$, 
is an ordered sequence of clauses
\mbox{$\proofstd = (\cld_1, \dotsc, \cld_{L})$}
such that $\cld_{L} = \emptycl$ and
for each $i \in [L]$
either 
$\cld_i$ is a clause in $\formf$
(an \introduceterm{axiom})
or there exist $j < i$ and $k < i$ such that
$\cld_i$ is derived from $\cld_j$ and $\cld_k$
by the
\introduceterm{resolution rule}
\begin{equation}
\label{eq:resolution-rule}
\AxiomC{$\clb \lor x$}
\AxiomC{$\clc \lor \olnot{x}$}
\BinaryInfC{$\clb \lor \clc$}
\DisplayProof
\eqperiod
\end{equation}
We refer to $\clb \lor \clc$ as the \introduceterm{resolvent} of $\clb
\lor x$ and~$\clc \lor \olnot{x}$ over $x$, and to $x$ as the
\emph{resolved variable}. 
For technical reasons it is 
sometimes convenient to also allow
clauses to be derived 
by the \introduceterm{weakening} rule
\begin{equation}
\label{eq:weakening}
\AxiomC{$\clc$}
\UnaryInfC{$\cld$}
\DisplayProof
\ \ [ \clc \subseteq \cld ]
\end{equation}
(and for two clauses 
$\clc \subseteq \cld$
we will sometimes refer to $\clc$
as a \introduceterm{strengthening} of~$\cld$).

The \introduceterm{\lengthsize}
$\length{\proofstd}$   
of a refutation $\proofstd = (\cld_1, \dotsc, \cld_{L})$
is~$L$ and the length of refuting 
$\formf$ is
$\min_{\refof{\proofstd}{\formf}}\set{\length{\proofstd}}$, 
where the minimum is taken over all resolution refutations~$\proofstd$
of~$\formf$.
It is easy to show that
removing the weakening rule~\refeq{eq:weakening}
does not increase the refutation length.

A \introduceterm{partial assignment} or a \introduceterm{restriction} on a 
formula $\formf$ is a partial function $\rho : \vars{F} \rightarrow \set{0,1}$. 
The clause $C$
\introduceterm{restricted by~$\rho$}, denoted $\restrict{C}{\rho}$, is
the trivial $1$-clause if any of the literals in~$C$ is satisfied by~$\rho$
and otherwise it is $C$ with all falsified literals removed. We extend
this definition  
to CNF formulas in the obvious way by taking unions.
For a variable $x\in \vars{F}$ we 
write
$\rho(x)=*$ if $x\notin \dom{\rho}$,
\ie if $\rho$ does not assign a value to~$x$.

\subsection{Concentration Inequalities}

We often denote random variables in boldface and write
$\sample{\Xrand}{\mathcal{D}}$ to denote that $\Xrand$ is sampled from
the distribution $\mathcal{D}$. We will use the following standard
forms of the multiplicative Chernoff bounds.
\begin{theorem}
  \label{thm:Chernoff}
  Let 
  $\boldsymbol{S}$ 
  be the sum of independent 0-1 random variables 
  (not necessarily equidistributed)
  with expectation $\mu = \expectation{\boldsymbol{S}}$.
  Then for $\delta \geq 0$ it holds that
  \begin{align*}
    \Prob{
    \mu - \boldsymbol{S} \geq \delta
    }
    \leq
      \exp
      \left(
      - \frac{\delta^2}{2\mu}
      \right)~\text{and}~
    \Prob{
    \boldsymbol{S} - \mu \geq \delta
    }
    \leq 
      \exp
      \left(
      - \frac{\delta^2}{2\mu + \delta}
      \right)\eqperiod
  \end{align*}
\end{theorem}

\subsection{Graph Theory}

We write  $G=(\V,E)$ to denote a graph with vertices~$\V$ and
edges~$E$, where $G$ is always undirected and without loops or
multiple edges.
Moreover, 
for bipartite graphs we write
$G = (U \disjointunion V, E)$, where edges in $E$ have one endpoint in
the left vertex set~$U$ and the other in the right vertex set~$V$. 
A \introduceterm{partial matching} $\matcha$ in $G$ is a subset of edges that 
are vertex-disjoint.
Let $V(\matcha) = \setdescr{v}{\exists e \in \matcha: v \in e}$ be the
vertices of $\matcha$ and for $v\in V(\matcha)$ denote by 
$\matcha_v$ the unique vertex $u$
such that $\set{u,v} \in \matcha$. A vertex $v$ 
is \introduceterm{covered} by $\matcha$ if $v \in V(\matcha)$.
If $\matcha$ is a partial matching in a bipartite graph 
$G = (U \disjointunion V, E)$, 
we identify it with a partial mapping of
$U$ to~$V$.
When referring to the pigeonhole formula,
this mapping will also be identified with an assignment
$\matchbool{\matcha}$ to the  
variables defined by
\begin{equation}
  \label{eq:matching-to-assignment}
  \matchboolof{\matcha}{x_{i,j}}
  =
  \begin{cases}
    * 
    & \text{if $i \notin \domainof{\matcha}$,} 
    \\
    0 
    & \text{if $i \in \domainof{\matcha}$ and $\matchaof{i} \neq j$,} 
    \\
    1 
    & \text{if $i \in \domainof{\matcha}$ and $\matchaof{i} = j$.} 
  \end{cases}
\end{equation}
For simplicity we interchangeably think of partial matchings as
partial functions and as Boolean assignments as defined in
\refeq{eq:matching-to-assignment}.

Given a vertex $v \in \V(G)$,
we write 
$\neigh[G]{v}$
to denote the set of \emph{neighbours of $v$} in the graph $G$
and $\deg[G]{v} = \setsize{\neigh[G]{v}}$ to denote the degree of $v$.
We extend this notion to sets and denote by 
$\neigh[G]{S} = \setdescr{v}{\exists\, (u,v)\in E \text { for } u\in S}$ 
the \emph{neighbourhood}
of a set of vertices $S \subseteq V$.
The \emph{boundary}, 
or \emph{unique  neighbourhood},
$\uniqueNeigh[G]{S}
=
\setdescr[:]
{v \in V \setminus S}
{\setsize{\neigh[G]{v} \intersection S} = 1}
$
of a set of vertices
$S \subseteq V$ 
contains all vertices in
$V \setminus  S$ that have a single neighbour 
in~$S$.
  If the graph is bipartite, there is of course  no need to
  subtract~$S$ from the neighbour set.
We will sometimes drop the subscript $G$ when the graph is clear from
context.
For a set $U \subseteq \V$ we
denote by $G\setminus U$ the subgraph of $G$ induced by the vertex set
$\V \setminus U$.  

A graph $G = (V, E)$ is an
\introduceterm{$(r, \gdegmax, c)$-expander} 
if all vertices $v \in V$ have degree at most $\gdegmax$ and for all
sets $S \subseteq V$, $\setsize{S} \le r$, it holds that
$\setsize{ \neigh{S} \setminus S } \geq c \cdot \setsize{S}$.
Similarly,
$G = (V, E)$
is an \emph{$(r, \gdegmax, c)$-boundary expander} if all vertices 
$v \in V$ 
have degree at most $\gdegmax$ and for all sets  
$S \subseteq V$,  $\setsize{S} \leq r$,
it holds that $\setsize{ \uniqueNeigh{S} } \geq c \cdot \setsize{S}$. 
For bipartite graphs, 
the degree and expansion requirements only apply to the left vertex set:
$G = (U \disjointunion V, E)$
is an
\introduceterm{$(r, \gdegmax, c)$-bipartite expander} 
if all vertices 
$u \in U$ 
have degree at most $\gdegmax$ and for all sets  
$S \subseteq U$,  $\setsize{S} \leq r$,
it holds that $\setsize{ \neigh{S} } \geq c \cdot \setsize{S}$,
and an
\emph{$(r, \gdegmax, c)$-bipartite boundary expander} if 
for all sets  
$S \subseteq U$,  $\setsize{S} \leq r$,
it holds that $\setsize{ \uniqueNeigh{S} } \geq c \cdot \setsize{S}$. 
For bipartite graphs we will only ever be interested in bipartite
notions of expansions, and so which kind of expansion is meant will
always be clear from context.
A simple
but useful observation is that
\begin{equation}
  \label{eq:neighbour-leq-unique}
  \setsize{ N(S) \setminus S } \le 
  \setsize{ \uniqueNeigh{S}}  + 
  \frac{\GdegMax \setsize{ S } - \setsize{ \uniqueNeigh{S} }}{2}
  = \frac{\GdegMax \setsize{ S } + \setsize{ \uniqueNeigh{S} }}{2}
  \eqcomma
\end{equation}
since all non-unique neighbours in $N(S) \setminus S$ have at least two incident edges.
This implies that if a graph $G$ is an 
\mbox{$(r, \gdegmax, (1-\diff)\gdegmax)$}-expander then it 
is
also
an $(r, \gdegmax, (1-2\diff)\gdegmax)$-boundary expander.

For $\numholes, \numpigeons, \GdegMax \in \N$, we
denote by $\Gdist(\numpigeons, \numholes, \GdegMax)$ the distribution over
bipartite graphs
with disjoint vertex sets $U = \set{u_1, \ldots, u_{\numpigeons}}$ and 
$V = \set{v_1, \ldots, v_{\numholes}}$
where the neighbourhood of a vertex $u \in U$ is chosen by sampling 
a subset of size $\GdegMax$ uniformly at random from $V$.
A property is said to hold
\emph{\aas} on $\Gdist(f(\numholes), \numholes, \GdegMax)$ if it
holds with probability that approaches $1$ as $\numholes$ approaches infinity.

For the right parameters, a randomly sampled graph
$G \sim \Gdist(\numpigeons, \numholes, \GdegMax)$ is asymptotically
almost surely a good boundary expander as stated next.

\begin{lemma}\label{lem:expander}
  Let $\numpigeons,\numholes$ and $\GdegMax$ be large enough integers such that 
  $\numpigeons > \numholes \ge \GdegMax $.
  Let $\diff, \chi \in \Rplus$
  be such that 
  $\diff < 1/2$,
  $\diff \ln \chi \ge 2$
  and $\diff \GdegMax \ln \chi \ge 4 \ln \numpigeons$.
  Then for
  $r =  \numholes /(\GdegMax \cdot \chi )$ 
  and $c = (1 - 2 \diff) \GdegMax$
  it holds
  asymptotically almost surely for a randomly sampled graph
  $G \sim \Gdist(\numpigeons,\numholes,\GdegMax)$
  that
  $G$ is an 
  $(r, \GdegMax, c)$-boundary expander.
\end{lemma}

\begin{proof}
  Let $G= (U \disjointunion V, E)$. We 
  first estimate the probability 
  that a set $S \subseteq U$ of size at most $r$ violates 
  the boundary expansion. 
  For brevity, let us write
  $s = \setsize{ S }$
  and $c' = (1 - \diff)\GdegMax$.
  In view of~\refeq{eq:neighbour-leq-unique}, 
  the probability that $S$ violates the boundary expansion
  can be bounded by
  \begin{subequations}
  \begin{align}
    \Prob{ \setsize{ \uniqueNeigh{S} } < cs } 
    &\le \PROB{ \setsize{ N(S) } < \frac{\GdegMax s + cs}{2}  }\\
    &= \Prob{ \setsize{ N(S) } < c's  }\\
    &\leq
      \binom{\numholes}{ c's } \cdot
      \left(
      \frac
      {\binom{c' s}{\gdegmax}}
      {\binom{\numholes}{\gdegmax}}
      \right)^s \\
    &\le  \binom{\numholes}{c's}  \cdot
    \left( \frac{ c's }{\numholes} \right)^{\GdegMax s}\\
    &\le 
     \left[
     \left( \frac{e\numholes}{c's} \right)^{c'} \cdot
     \left( \frac{c's}{\numholes} \right)^{\GdegMax}
     \right]^s
    \\
    &= 
    \left[
    {e}^{(1-\diff)\GdegMax} \cdot
    \left( \frac{\numholes}{c' s} \right)^{-\diff \GdegMax }
    \right]^s\\
    &\le \exp \left(\GdegMax s \left(1 -\diff \ln \left(\frac{\numholes}{c's}\right)\right) \right)\\ \label{eq:exp_first}
    &\le \exp \left(\GdegMax s \left(1 -\diff \ln \left(\frac{\chi}{1-\diff}\right)\right) \right)\\
    &\le \exp \left(\GdegMax s (1 -\diff \ln {\chi}) \right)\\ \label{eq:exp_last}
    &\le \exp \left(- (\GdegMax s \diff \ln {\chi}) /2 \right) \eqcomma
  \end{align}
  \end{subequations}
  where \eqref{eq:exp_first} holds since
    $
    s \leq r \leq \numholes / (\gdegmax \chi)
    $
    and
    \eqref{eq:exp_last} holds since $\diff \ln \chi \ge 2$.
  Hence, the probability that $G$ is not a boundary expander
  can be bounded by
  \begin{align}
\nonumber
    \Prob{G \text{ is not an expander}} &\le 
    \sum_{s \in [r]}\binom{\numpigeons}{s} \exp( - (\GdegMax s \diff \ln \chi)/2)\\
    &\le \sum_{s \in [r]} \exp (-s ((\diff \GdegMax \ln \chi)/2 - \ln
      \numpigeons))\\
\nonumber
    &\le \sum_{s \in [r]} \exp( -s \ln \numpigeons) \le \frac{1}{\numpigeons - 1} \eqcomma 
  \end{align}
where the second-to-last inequality holds since $\diff \GdegMax \ln \chi \ge 4 \ln \numpigeons$.
\end{proof}

We will also consider some parameter settings where randomly
sampled graphs do not have strong enough expansion for our purposes, 
but where we can resort to explicit constructions as follows.

\begin{theorem}[\cite{GUV09Unbalanced}]\label{thm:expander}
  For all positive integers $\numpigeons$, $r \le \numpigeons$, all $\diff > 0$,
  and all constant $\nu > 0$, there is an explicit 
  $( r, \gdegmax, (1- \diff)\GdegMax)$-expander
  $G=(U\disjointunion V, E)$,
  with $\setsize{U} = \numpigeons$,
  $\setsize{V} = \numholes$, 
  $\GdegMax = \BIGOH{ ((\log \numpigeons)(\log r) / \diff )^{1+1/\nu} }$
  and $\numholes \le \GdegMax^2 \cdot r^{1+\nu}$.
\end{theorem}

\begin{corollary}\label{cor:boundary_subexp}
  Let $\kappa, \varepsilon, \nu$ be positive constants,
  $\kappa <\frac{1}{8}$, and let $\numholes$ be a large enough
  integer.  Then there is an explicit graph
  $G=(U \disjointunion V, E)$, with
  $\setsize{U} = \numpigeons = 2^{\bigomega{\numholes^\kappa}}$ and
  $\setsize{V} \leq \numholes$, that is an
  $(\numholes^{\frac{1}{1 + \nu} - \frac{4\kappa}{\nu} }, \GdegMax, (1
  - 2\diff)\GdegMax)$-boundary expander for
  $\diff = \frac{ \varepsilon \log \numholes }{ \log \numpigeons }$
  and $\GdegMax = O(\log^{2 (1 + 1/\nu) }\numpigeons)$.
\end{corollary}

  \begin{proof}
    Let $G$ be the expander from Theorem~\ref{thm:expander} for the
    parameters
    $\numpigeons = 2^{\varepsilon' \numholes^\kappa}$, 
    $r = \numholes^{\frac{1}{1 + \nu} - \frac{4\kappa}{\nu} }$, and
    $\diff = \frac{ \varepsilon \log \numholes }{ \log \numpigeons }$,
    where $\varepsilon'$ is chosen to be a small enough constant so that 
    $\GdegMax^2 \cdot r^{1+\nu} \leq  \numholes$.
    Such a graph $G$ is an $(r, \gdegmax, (1-\diff)\GdegMax)$-expander
    for $\GdegMax$ as in the
    Corollary. By \refeq{eq:neighbour-leq-unique} it follows that an 
    $(r, \gdegmax, c)$-expander is an
    $(r, \gdegmax, 2c - \gdegmax)$-boundary expander, 
    and hence $G$ is an 
    $\bigl(r, \GdegMax, (1 -  2\diff)\GdegMax\bigr)$-boundary expander. Note
    that Theorem~\ref{thm:expander} guarantees that the right side of $G$ has size at most 
    $\GdegMax^2 \cdot r^{1+\nu} \leq \numholes$. 
  \end{proof}

\section{Two Key Technical Tools}
\label{sec:filterandclosure}

In this section we review two crucial technical ingredients of the
resolution lower bound proofs.

\squeezesubsection{Pigeon Filtering}

The following lemma is a generalization of
\cite[Lemma~6]{Razborov03ResolutionLowerBoundsWFPHP} and states that
for every small collection~$\vecr(1), \ldots, \vecr(\nvecs)$ of
integer vectors there is a vector~$\vecr$ such that each
vector~$\vecr(\ell)$ has either many coordinates~$i$ smaller
than~$\vecr$ (\ie~$r_i(\ell) \leq r_i$), or has few coordinates~$i$
which are by~$1$ larger than~$\vecr$, that is, there are few
coordinates~$i$ such that~$r_i(\ell) \leq r_i+1$.

The proof follows by sampling each coordinate of~$\vecr$ independently
from a geometric probability distribution and applying a union bound
over the vectors $\vecr(1), \ldots, \vecr(\nvecs)$. The difference
to~\cite{Razborov03ResolutionLowerBoundsWFPHP} is that our geometric
distribution depends on an additional parameter~$\alpha$ (which is
implicitly fixed to $\alpha = 2$ in
\cite{Razborov03ResolutionLowerBoundsWFPHP}) that allows us to get a
better upper bound on the numbers~$\ri{i}$.  This turns out to be
crucial for us---we discuss this in more detail in \refsec{sec:GFPHP}
where we explain how the lemma relates to the weak pigeonhole
principle.

\begin{lemma}[Filter lemma]\label{lem:filter}
  Let $m, \nvecs \in \Nplus$ and suppose that $w_0, \alpha \in [m]$
  are such that $w_0 > \ln \nvecs$ and $w_0 \ge \alpha^2 \ge 4$.
  Further, let $\vecr(1), \ldots, \vecr(\nvecs)$ be integer vectors,
  each of the form $\vecrofdef{\indexvec}{m}$. Then there exists a
  vector $\vecrdef{m}$ of positive integers
  $\ri{i} \le \Floor{ \frac{\log m}{\log \alpha} } - 1$ such that for
  all $\indexvec \in [\nvecs]$ at least one of the following holds:
  \ifthenelse{\boolean{false}} {\begin{enumerate}}
    {\begin{enumerate}[topsep=3pt,itemsep=0pt]}
    \item $\Setsize{ \setdescr[:]{ i \in [m] }{ \riof{\indexvec}{i} \le \ri{i} } } \ge w_0$ 
    \eqcomma
    \item $\Setsize{ \setdescr[:]{ i \in [m] }{\riof{\indexvec}{i} \le \ri{i} + 1} } \le 
      O(\alpha \cdot w_0)$ \eqperiod
  \end{enumerate}
\end{lemma}

\ifthenelse{\boolean{false}}
{
\begin{proof}[Proof sketch.]
}
{
\begin{proof}
}
  We first define 
  a weight function $W(\vecr)$ for vectors $\vecrdef{m}$ as
\begin{equation}
  \label{eq:weight-function}
  W(\vecr) = \sum_{i \in [m]} \alpha^{-\ri{i}}
  \eqperiod
\end{equation}
In order to establish the lemma,
  it is sufficient to show that there
  exist constants $\cgamma$ and $\cgammap$ 
  and 
  a vector $r = (r_1, \ldots, r_m)$
  such that for all $\indexvec \in [\nvecs]$
  the implications
 \ifthenelse{\boolean{false}}{\vspace{-5pt}}{}\begin{subequations}
  \begin{align}
    \label{eq:filter-low-weight}
    W(\vecrof{\indexvec}) \ge \frac{ \cgammap w_0}{ \alpha } 
    &\ \Rightarrow \  
      \setsize{ \set{ i \in [m] \mid \riof{\indexvec}{i} \le  \ri{i}} } \ge w_0 \eqcomma
    \\
    \label{eq:filter-high-weight}
    W(\vecrof{\indexvec}) \le \frac{ \cgammap w_0}{ \alpha } 
        &\ \Rightarrow \  
      \setsize{ \set{ i \in [m] \mid \riof{\indexvec}{i} \le  \ri{i} + 1 } } \le 
      \cgamma \alpha w_0  
  \end{align}
 \end{subequations}
  hold.
  Let $\limval = \Floor{ \frac{ \log m }{ \log \alpha } } - 1$ and let $\distfilter$ be a 
  probability distribution
  on $[\limval]$ given by
  $\prob {\rb = i} = \beta \cdot \alpha^{-i}$
  for all $i \in [\limval]$,
  where $\beta = \frac{\alpha - 1}{1 - \alpha^{-\limval}}$. 
\ifthenelse{\boolean{false}}{}{Note that 
  \begin{equation}
  \beta \sum_{i \in [\limval]} \alpha^{-i} 
    = \frac{\alpha - 1}{1 - \alpha^{-\limval}} 
      \left( \frac{{1 - \alpha^{-\limval}}}{\alpha - 1} \right) = 1
  \end{equation}
  and thus $\distfilter$ is a valid distribution.}
  Let us write 
  $\vecrb = (\rbi{1}, \ldots, \rbi{m})$
  to denote a random vector with coordinates sampled  independently
  according to~$\distfilter$.
  We claim that for every $\indexvec \in [\nvecs]$ the implications
  \refeq{eq:filter-low-weight}
  and~\refeq{eq:filter-high-weight}
  are true
  asymptotically almost surely.
\ifthenelse{\boolean{false}}{The proof of this fact follows by applying Chernoff bounds as 
in~\cite{Razborov03ResolutionLowerBoundsWFPHP}. 
}
{Let us  proceed to verify this.
  \begin{enumerate}
  \item 
    Suppose that 
    $W(\vecrof{\indexvec}) \ge \frac{ \cgammap w_0}{ \alpha}$. 
    We wish to show     that 
    $\setsize{ \setdescr[:]{ i \in [m] }{ \ri{i} \ge \riof{\indexvec}{i} } }
    \ge w_0$.
    Observe that 
    coordinates larger than $\limval$ contribute only
    \begin{equation}
      \sum_{\riof{\indexvec}{i} > \limval}{\alpha^{-\riof{\indexvec}{i}}} \le 
      m \cdot \alpha^{-\limval-1} < \alpha
    \end{equation}
to $W(\vecrof{\indexvec})$,
    and hence the weight function  truncated at $\limval$ is
    \begin{equation}
      \label{eq:truncated-weight}
      \sum_{\riof{\indexvec}{i} \le \limval} { \alpha^{-\riof{\indexvec}{i}} } \ge 
      \frac{ \cgammap w_0}{ \alpha } - \alpha \ge (\cgammap - 1)\frac{ w_0}{ \alpha } \eqcomma
  \end{equation}
  since $ w_0 \geq \alpha^2$.
  Note that for every coordinate $i$ with $\riof{\indexvec}{i} \le \limval$
  we have that 
  $\prob{\rbi{i} \ge \riof{\indexvec}{i}} \ge \beta \cdot \alpha^{-\riof{\indexvec}{i}}$.
  Consider the 
  random
  set
  $\goodset{\indexvec} = \set{ i \in [m] \mid 
    \riof{\indexvec}{i} \le \limval \text{ and } \rbi{i} \ge \riof{\indexvec}{i} }$. 
  We can appeal to~\refeq{eq:truncated-weight}
  to derive  that 
\begin{align}
  \Expectation{ \setsize{ \goodset{\indexvec} } } 
  &=
  \sum_{\riof{\indexvec}{i} \le \limval}
  \prob{\rbi{i} \ge \riof{\indexvec}{i}}\nonumber\\
  &\geq 
  \sum_{\riof{\indexvec}{i} \le \limval}
  \beta
  { \alpha^{-\riof{\indexvec}{i}} } \nonumber\\
  &\geq
  \beta (\cgammap - 1) \frac{ w_0}{ \alpha }
  \geq
  \frac{\cgammap - 1}{2} w_0 
\end{align}
is  a lower bound  on the expected size of $\goodset{\indexvec}$.
As the events $\rbi{i} \ge \riof{\indexvec}{i}$ are independent, by
the multiplicative Chernoff bound, that is, \refthm{thm:Chernoff}, we
get that
\begin{subequations}
\begin{align}
  \Prob{\abs{ \goodset{\indexvec} } < w_0} 
  &\le
  \Prob{ \setsize{ \goodset{\indexvec} } - 
  \Expectation{ \setsize{ \goodset{\indexvec} } } \le w_0  - \Expectation{ \setsize{ \goodset{\indexvec} } } }\\
  &=
  \Prob{ \Expectation{ \setsize{ \goodset{\indexvec} } } - \setsize{ \goodset{\indexvec} }
   \ge \Expectation{ \setsize{ \goodset{\indexvec} } } - w_0  }\\
  &\le
  \exp
  \bigl( - 
  \frac{\bigl(\Expectation{ \setsize{ \goodset{\indexvec} } } - w_0 \bigr)^2}{2 \Expectation{ \setsize{ \goodset{\indexvec} } } }
  \bigr)\\  
  &=
  \exp 
  \left( - 
  \frac{\Expectation{ \setsize{ \goodset{\indexvec} } }^2 - 2\Expectation{ \setsize{ \goodset{\indexvec} } } w_0 + w_0^2}{2 \Expectation{ \setsize{ \goodset{\indexvec} } } }
  \right)\\
  &\le
  \exp
  \left( - 
  \frac{\Expectation{ \setsize{ \goodset{\indexvec} } } - 2 w_0}{ 2 }
  \right)\\
  &\le
  \exp
  \left( - 
  \frac{ (\cgammap - 5)}{ 4 } w_0
  \right)\\
  &\le
  \exp( -2 w_0 )\\
  &\le
  \nvecs^{-2} \eqcomma
\end{align}
\end{subequations}
where the second to last inequality holds for $\cgammap\geq 13$. 

\item Suppose that $W(\vecrof{\indexvec}) \le \frac{ \cgammap w_0}{ \alpha }$.
  Now we need to show 
  that 
  $\setsize{ \setdescr[:]{ i \in [m] }{ \ri{i} \ge \riof{\indexvec}{i} - 1 } }
  \le \cgamma \alpha w_0$
  holds asymptotically almost surely.
  Note that
\begin{subequations}
  \begin{align}
	\prob{\rbi{i} \ge \riof{\indexvec}{i} - 1} 
	&= \beta \sum_{j=\riof{\indexvec}{i} - 1}^\limval \alpha^{-j} \\
	&= \frac{\alpha - 1}{1 - \alpha^{-\limval}}
	\left( 
	\frac{\alpha^{-\riof{\indexvec}{i} + 2} -
	\alpha^{-\limval}}{\alpha - 1} 
	\right)\\
	&= \frac{\alpha^{-\riof{\indexvec}{i} + 2} - 
	\alpha^{-\limval}}{1 - \alpha^{-\limval}}\\
	&= \frac{\alpha^{\limval - \riof{\indexvec}{i} + 2} - 1}{\alpha^{\limval} - 1}\\
	&\le \frac{\alpha^{\limval - \riof{\indexvec}{i} + 2} }{\alpha^{\limval}/2} 
	= 2 \alpha^{2-\riof{\indexvec}{i}} \eqperiod
  \end{align}
\end{subequations}
  Similar to the previous case, let
  $\betterset{\indexvec} = \set{ i \in [m] \mid  \rbi{i} \ge \riof{\indexvec}{i} - 1 }$.
  We can upper-bound the expected cardinality of 
  $\betterset{\indexvec}$ by
  \begin{align}
    \Expectation{ \setsize{ \betterset{\indexvec} } } 
    =
    \sum_{i\in [m]}
    \prob{\rbi{i} \ge \riof{\indexvec}{i} - 1}    
    \le 
    2 \alpha^2 W(\vecrof{\indexvec}) \le 2\cgammap \alpha w_0 \eqperiod
  \end{align}
  Again, we apply the Chernoff bound in \refthm{thm:Chernoff} and conclude that
  \begin{align}
\nonumber
    \Prob{ \setsize{ \betterset{\indexvec} } \ge \cgamma \alpha w_0 }
    &\le
    \PROB{ \setsize{ \betterset{\indexvec} } - 
    \Expectation{ \setsize{ \betterset{\indexvec} } }  
    \ge \cgamma \alpha w_0 - 2\cgammap \alpha w_0 }\\
    &\le
    \exp\left( - \frac{(\cgamma - 2\cgammap)^2 (\alpha w_0)^2}{4\cgammap \alpha w_0 + (\cgamma - 2\cgammap) \alpha w_0} \right)\\
\nonumber
    &\le
    \exp( - \alpha w_0 )\\
\nonumber
    &\le
    \nvecs^{-2}
  \end{align}
where the second to last inequality holds for $\cgamma$ sufficiently larger than $\cgammap$, say $\cgamma \geq 5 \cgammap$.
\end{enumerate}
}A union bound argument over all vectors in 
$\setdescr[:]{ \vecrof{\indexvec} }{ \indexvec \in [\nvecs] }$ 
for both cases
shows
that for $\cgammap \geq 13$ and $\cgamma \geq 5 \cgammap$
there exists a choice of $\vecrdef{m}$ 
such that both implications~\refeq{eq:filter-low-weight} and~\refeq{eq:filter-high-weight} hold.
\end{proof}

\squeezesubsection{Graph Closure}

\newcommand{\Vexp}{V_{\mathsf{exp}}}

A key concept in our work is that of a \emph{closure} of a vertex set,
which seems to have originated
in~\cite{ABRW04Pseudorandom,AR03LowerBounds}. In later arguments we
will want to remove a set~$T$ of vertices along with its
neighbourhood~$N(T)$ from an expanding graph~$G$ while maintaining
some guarantee of expansion. The closure of~$T$ is a suitably small
set $S \supseteq T$ such that the graph $G\setminus (S \cup N(S))$ is
a reasonably good expander; see \reflem{lem:closure_expansion} for the
formal guarantees. For intuition, it is useful to think of~$S$ as being
constructed by an iterative
process~$T = S_0 \subset S_1 \subset \cdots \subset S_t = S$
where~$S_i = S_{i-1} \cup S'_i$ for~$S'_i$ any small set of vertices
in~$G\setminus (S_{i-1} \cup N(S_{i-1}))$ that is not sufficiently
expanding. It turns out to be more convenient to work with the
somewhat technical definition that follows.

In order to have a definition that makes sense for both expanders and
bipartite expanders, we define $\Vexp(G)$ to be the set of vertices of
$G$ that expand, that is, if $G=(V,E)$ is an expander then
$\Vexp(G) = V$, and if $G=(U\disjointunion V, E)$ is a bipartite
expander then $\Vexp(G) = U$.

\begin{definition}[Closure]\label{def:closure}
  For an expander graph $G$ and
  vertex sets $S \subseteq \Vexp(G) $ and $U \subseteq V(G)$,
  we say that
  the set $S$ is \emph{\contained{U}{r}{\almostcontained}} if
  $\setsize{S} \le r$ 
  and 
  $
  \Setsize{
    \uniqueNeigh{S} \setminus U
  }
  <
  \almostcontained \cdot \setsize{S}$. 

  For any expander graph $G$  and any set $T \subseteq \Vexp(G)$       of size $\setsize{T} \leq r$, we let
  $\closure[r, \almostcontained]{T}$
  denote an arbitrary but fixed maximal set
  such that
  $T \subseteq \closure[r, \almostcontained]{T} \subseteq \Vexp(G)$
  and  $\closure[r, \almostcontained]{T}$ is
  \contained{\neigh{T}}{r}{\almostcontained}.
\end{definition}

Note that the closure of any set $T$ of size 
$\setsize{T} \le r$ as defined above does indeed exist, since~$T$ itself is
\contained{\neigh{T}}{r}{\almostcontained}.

\begin{lemma}\label{lem:closure_size}
  Suppose that $G$ 
  is an $(r, \gdegmax, \expansionconst )$-boundary expander
  and that
  $T \subseteq \Vexp(G)$ has size $\setsize{T} \le \setsizemax \le r$. 
  Then
  $\setsize{ \closure[r, \almostcontained]{T} } < 
  \frac{\setsizemax \GdegMax}{\expansionconst -
    \almostcontained}$.
\end{lemma}

\begin{proof}
  By definition we have that 
  $ \Setsize{ 
  \uniqueNeigh{\closure[r, \almostcontained]{T}}  \setminus \neigh{T} } < 
  \nu \cdot \setsize{\closure[r, \almostcontained]{T}}$.
  Furthermore, since 
  $\setsize{\closure[r, \almostcontained]{T}} \leq r$ by definition,
  we can use the expansion property of the graph to derive 
  the inequality
  $\Setsize{
    \uniqueNeigh{\closure[r, \almostcontained]{T}} \setminus \neigh{T}
  } 
  \ge \setsize{\uniqueNeigh{\closure[r, \almostcontained]{T}}} - \setsize{ \neigh{T} }
  \ge  
  \expansionconst \cdot \setsize{\closure[r, \almostcontained]{T}} - 
  \setsizemax  \GdegMax 
  $. Note that we also use the fact that the neighbourhood of $T$ is of size at most 
  $\setsizemax \GdegMax$. 
  The conclusion follows by combining both statements.
\end{proof}

  Suppose $G$ is an excellent boundary expander and that $T \subseteq \Vexp(G)$
  is not too large. Then \reflem{lem:closure_size} shows that the
  closure of $T$ is not much larger.  
  And if the closure is not too large, then
  after removing the closure and its neighbourhood from the graph
  we are still left with a decent expander,
  a fact which will play a key role in the technical arguments in later
  sections. The following lemma makes this intuition precise.

\begin{lemma}\label{lem:closure_expansion}
  For $G$ an 
  $(r, \gdegmax, \expansionconst )$-boundary expander,
  let $T \subseteq \Vexp(G)$ be such that $\setsize{T} \le r$ and
  $\setsize{ \closure[r, \almostcontained]{T} } \le r/2$,
  let $G' = G \setminus \bigl( \closure[r, \almostcontained]{T} \cup 
  \neigh{\closure[r, \almostcontained]{T}} \bigr)$ and
  $\Vexp(G') = \Vexp(G) \cap V(G')$.
  Then any set $S \subseteq \Vexp(G')$ of size $\setsize{S} \le r/2$
  satisfies
$    \abs{ \uniqueNeigh[G']{S}} \ge 
    \almostcontained \setsize{S} $. \end{lemma}

\begin{proof}
  Suppose the set $S \subseteq \Vexp(G')$ is of size $\setsize{S} \le r/2$ and 
  does not satisfy $    \abs{ \uniqueNeigh[G']{S}} \ge 
    \almostcontained \setsize{S} $.
  Since
  $\closure[r, \almostcontained]{T}$ is also of size at most
  $r/2$, we have that
the set
  $(\closure[r, \almostcontained]{T} \cup S)$ is 
  \contained{\neigh{T}}{r}{\almostcontained} in\ifthenelse{\boolean{false}}{}{ the graph} $G$.
  But this contradicts the maximality of $\closure[r, \almostcontained]{T}$.
\end{proof}

\section{Lower Bounds for Weak Graph FPHP Formulas }
\label{sec:GFPHP}

We now proceed to establish lower bounds on the length of resolution
refutations of 
functional pigeonhole principle formulas defined over bipartite graphs.
We write
$G = (\Vp \disjointunion \Vh, E)$
to denote the graph over which the formulas are defined and $\matchings$ 
to denote
the set of partial matchings on $G$
(also viewed as partial mappings of $\Vp$ to $\Vh$).
Let us start by making more
precise some of the technical notions 
discussed in the introduction
(which were originally defined 
in~\cite{Razborov01ImprovedResolutionLowerBoundsWPHP}).

For a 
clause~$C$ and a 
pigeon~$i$ we denote the 
set of
holes $j$ with
the property that $C$ is satisfied if $i$ is matched to~$j$ by
\begin{equation}
	\neigh[C]{i} = \setdescr{ j \in \Vh }{ e=\set{i, j} \in 
	E~\text{and}~\matchboolof{\set{e}}{C} = 1 }
\end{equation}
and we define
the \introduceterm{$i$th pigeon degree 
  $\deg[C]{i}$ 
  of $C$} 
as  
\begin{equation}
  \label{eq:pig-deg}
  \deg[C]{i} 
  = \setsize{ \neigh[C]{i} } 
  \eqperiod
\end{equation}
We think
of a pigeon $i$ with large $\deg[C]{i}$ as a pigeon on which the
derivation has not made any significant progress up to the point of
deriving~$C$, since the clause rules out very few holes. 
The pigeons with high enough pigeon degree in a
clause are the \emph{\heavy pigeons} of the clause
as defined next.

\begin{definition}[Pigeon weight, \pseudowidth and \fakeaxiom{s}]
  \label{def:pseudo-width}
  Let $C$ be a clause and
  let $\vecthrdef{m}$ and $\vecadvdef{m}$ be two 
  vectors 
  of positive integers
  such that
  $\vecthr$ is elementwise greater than $\vecadv$. We say that pigeon $i$
  is \introduceterm{$\vecthr$\nobreakdash-\sheavy for $C$} 
  if $\deg[C]{i} \ge \thr{i}$
  and that pigeon~$i$ is 
  \introduceterm{$(\vecthr, \vecadv)$-\heavy for $C$}
  if $\deg[C]{i} \ge \thr{i} - \adv{i}$. When  $\vecthr$ and $\vecadv$ 
  are understood from context, which is most often the case, 
  we omit the parameters and just refer to 
  \introduceterm{\superheavy} and \introduceterm{\heavy} pigeons.
  Pigeons that are not \heavy are referred to as 
  \emph{\light pigeons}.
  The set of pigeons that are \sheavy for $C$ is denoted by
  \begin{equation*}
    \Pfat{C} = \setdescr{i \in [m]}{\deg[C]{i} \ge \thr{i}} 
  \end{equation*}
  and the set of pigeons that are \heavy for $C$ is denoted by 
  \begin{equation*}
    \Pthick{C} = \setdescr{i \in [m]}{\deg[C]{i} \ge \thr{i} - \adv{i}} 
    \eqperiod
  \end{equation*}
  The 
  \introduceterm{\pseudowidth}  
  of~$C$ 
  is the number of \heavy pigeons in $C$
  and the 
  \pseudowidth of 
  a resolution refutation $\refpi$, 
  denoted by $\width{\refpi}$, is 
  $\max_{C \in \refpi}{\width{C}}$.
  Finally, we will refer to clauses~$C$ with precisely~$w_0$~\sheavy
  pigeons, \ie such that 
  $\setsize{ \Pfat{C} } = w_0$, as
  \introduceterm{\fakeaxiom{}s}.
\end{definition}
Note that 
according to \refdef{def:pseudo-width}
\sheavy pigeons are also \heavy.
Making the connection back to 
\ifthenelse{\boolean{false}}
{the introduction,} 
{our informal discussion in the introduction,} 
the ``fake axioms'' mentioned there are nothing other than
\fakeaxiom{s}.

Now that we have all the notions needed, let us give a detailed proof
outline.  Given a short resolution refutation~$\refpi$ of the
formula~$\GFPHP$, we use the Filter lemma (\reflem{lem:filter}) to get
a filter vector $\vecthrdef{m}$ such that each clause either has many
\sheavy pigeons or there are not too many \heavy pigeons (for an
appropriately chosen vector $\vecadv$).  Clearly, clauses that fall
into the second case of the filter lemma have bounded \pseudowidth. On
the other hand, clauses in the first case may have very large
\pseudowidth. In order to obtain a proof of low \pseudowidth, these
clauses are strengthened to \fakeaxiom{s} and added to a special set
$\Ax$. This then gives a refutation $\refpi'$ that refutes the formula
$\GFPHP \cup \Ax$ in bounded \pseudowidth. The following lemma
summarizes the upper bound on \pseudowidth that we obtain.

\begin{lemma}\label{lem:fphp_upper_bound}
  Let $G = (\Vp \disjointunion \Vh, E)$
  be a bipartite graph with
  $\setsize{ \Vp } = m$ and $\setsize{ \Vh } = n$;
  let $\refpi$ be a resolution refutation of $\GFPHP$;
  let $w_0, \alpha \in [m]$ be such that 
  $w_0 > \log \length{\refpi}$ and
  $w_0 \ge \alpha^2 \ge 4$,
  and
  let $\vecadvdef{m}$ be defined by 
  $\adv{i} = \frac{\deg[G]{i} \log \alpha}{\log m}$.
  Then
  there exists 
  an integer vector $\vecthrdef{m}$, with 
  $\adv{i} < \thr{i} \le \deg[G]{i}$ for all 
  $i \in \Vp$, a set of \fakeaxiom{s} $\Ax$
  with 
  $\abs{\Ax} \le \length{\refpi}$, and a resolution refutation 
  $\refpi'$ of $\GFPHP \cup \Ax$ such that
  $\width{\refpi'} = \bigoh{\alpha \cdot w_0}$.
\end{lemma}
As 
mentioned
above, this 
upper bound
is a straightforward application of 
\reflem{lem:filter}. 
\ifthenelse{\boolean{false}}
{We defer the formal proof to a later point in this section.}
{We defer the formal proof to \refsec{sec:fphp_upper_bound}.}
What we will need from \reflem{lem:fphp_upper_bound} is that a
resolution refutation of $\GFPHP$ in length less than $2^{w_0}$ can be
transformed into a refutation of $\GFPHP \cup \Ax$ in \pseudowidth at
most~$\bigoh{\alpha \cdot w_0}$.

The second step in the proof is to show that 
any resolution 
refutation 
$\refpi$ of $\GFPHP \cup \Ax$
requires large \pseudowidth. The 
high-level 
idea is to define a  progress measure
on clauses $C \in \refpi$
by
counting the number of matchings on $\Pthick{C}$ 
that do not 
satisfy $C$. We then show that in order to increase this progress measure
we need large \pseudowidth.
The following lemma states the \pseudowidth lower bound.

\begin{restatable}{lemma}{lemmafphplower}
  \label{lem:fphp_lower_bound}\RestateRemark
	Let $\diff \le 1/\cdiff$
        and
        $m,n,r,\GdegMax \in \N$;
	let $G = (\Vp \disjointunion \Vh, E)$
        with
	$\setsize{ \Vp } = m$ and $\setsize{ \Vh } = n$
	be an $(r, \GdegMax, (1 - 2 \diff) \GdegMax)$-boundary expander,
	and let $\vecadvdef{m}$ be defined by 
	$\adv{i} = \cdiff \deg[G]{i} \diff$.
	Suppose that 
	$\vecthrdef{m}$ is an integer vector 
	such that $\adv{i} < \thr{i} \le \deg[G]{i}$ for all $i \in
         \Vp$. Let $w_0$
	be an arbitrary parameter and $\Ax$ be
	an arbitrary set of \fakeaxiom{s} with
	$\abs{\Ax} \le \left(1 + \diff \right)^{w_0} $.
	Then every resolution refutation $\refpi$ 
	of $\GFPHP \cup \Ax$
	must satisfy $\width{\refpi} \ge r \diff/4$.
\end{restatable}
In 
one
sentence, the lemma states that if 
the set of ``fake axioms'' 
$\Ax$ is not too large, 
then resolution requires large \pseudowidth
to refute $\GFPHP \cup \Ax$. Note that this lemma holds for
any filter vector and not just for the one 
obtained
from 
\reflem{lem:fphp_upper_bound}.

In order to prove \reflem{lem:fphp_lower_bound}, we wish to define a
progress measure on clauses that indicates how close the derivation is
to refuting the formula (\ie it should be small for axiom clauses but
large for contradiction). A first attempt would be to define the
progress of a clause $C$ as the number of ruled-out matchings (\ie
matchings that do not satisfy~$C$) on the pigeons mentioned
by~$C$. This definition does not quite work since we will want to keep
track of \emph{which} matchings are ruled out---unless there are more
than~$\dimLi$ holes to which a pigeon can be mapped without
satisfying~$C$. Then we will want to think of $C$ as ruling out
\emph{all holes} for this pigeon. Since the pigeon degree of a \light
pigeon~$i$ is at most $\thr{i} - \adv{i}$, such a pigeon will
certainly have at least $ \deg[G]{i} - d_i + \delta_i \geq \dimLi$
holes to which it can be mapped, and the ``lossy counting'' will thus
ensure that all holes are considered as ruled out.

We realize this ``lossy counting'' through a linear space $\linspace$,
in which each partial matching~$\matcha$ is associated with a
subspace~$\linmapof{\matcha}$.  Roughly speaking, the
progress~$\linmapof{C}$ of a clause~$C$ is then defined to be the span
of all partial matchings that are ruled out by~$C$.  We design the
association between matchings and subspaces so that the contradictory
empty clause~$\emptycl$ has $\linmapof{\emptycl} = \linspace$ but so
that the span of all the axioms
$\spans{\setdescr{ \linmapof{A}}{A \in \GFPHP \cup \Ax}}$ is a proper
subspace of $\linspace$. For this argument we crucially rely on the
gap ($r_i$ vs $r_i + 1$) between the two cases of the filter lemma
(\reflem{lem:filter}). This implies that in a refutation~$\refpi$ of
$ \GFPHP \cup \Ax $ there must exist a resolution step deriving a
clause~$C$ from clauses $C_0$ and~$C_1$ such that the linear space of
the resolvent $\linmapof{C}$ is not contained in
$\spans{\linmapof{C_0}, \linmapof{C_1}}$.  But the main technical
lemma of this section (\reflem{lem:fphp_span}) says that for any
derivation in low \pseudowidth the linear space of the resolvent is
contained in the span of the linear spaces of the clauses being
resolved. Hence, in order for~$\refpi$ to be a refutation it must
contain a clause with large \pseudowidth, and this establishes
\reflem{lem:fphp_lower_bound}.

So far our argument follows that of Razborov very closely, but
it turns out we cannot realize this proof idea if we only keep track
of \heavy and \light pigeons. Let us attempt a proof of the claim in
\reflem{lem:fphp_span} that low-width resolution steps cannot increase
the span to illustrate what the problem is.
The interesting case is when there is a pigeon $i$ that is \heavy
for~$C_0$ or~$C_1$
but not for their resolvent~$C$. 
Then, following Razborov, for any matching $\matcha$ on the \heavy
pigeons of $C$ that  fails to satisfy~$C$, 
we need to be able to extend~$\matcha$ in at
least $\dimLi$ different ways to a matching including also pigeon~$i$
that falsifies either~$C_0$ or~$C_1$. If this can be done, then we think
of $C_0$ and~$C_1$ as together ruling out (essentially) all holes
for~$i$, and  the linear space associated with $C$ will be contained
in the span of the spaces for $C_0$ and $C_1$.
However, the problem is that $\matcha$ can send all \heavy pigeons to
the neighbourhood of pigeon~$i$. In this scenario, there might be very
few holes, or even no holes, to which~$i$ can be mapped when
extending~$\matcha$, and even our lossy counting will not be able to
pick up enough holes for the argument to go through.
We resolve this problem by not only considering the \heavy pigeons but
a larger set of \introduceterm{relevant} pigeons including all pigeons~$i'$
that can become overly constrained when some matching on the \heavy pigeons
shrinks the neighbourhood of~$i'$ too much.
Formally, the \introduceterm{closure} of the set of \heavy pigeons,
as defined in \refdef{def:closure},
is the notion we need.

\squeezesubsection{Formal Statements of Graph FPHP Formula Lower Bounds}
\label{sec:fphp_results}

Deferring the proofs of all technical lemmas for now, 
let us state our lower bounds for graph FPHP formulas and see how they
follow from  
\reftwolems{lem:fphp_upper_bound}{lem:fphp_lower_bound}
above.

\begin{theorem} \label{thm:main}
  Let
  $m = \setsize{ U }$ and $n = \setsize{ V }$ and suppose that
  $
  G = (U 
  \disjointunion
  V, E)
  $
  is an
  $\bigl( 
  r, 
  \GdegMax, 
  \bigl(1 - \frac{ \log \alpha}{ 2 \log m} \bigr)\GdegMax \bigr)
  $-boundary 
  expander
  for $\alpha \in [m]$ such that 
  $8 \le \frac{ \alpha^3 }{ \log \alpha} = \Littleoh{\frac{ r }{ \log m }}$.
  Then resolution requires length
  $\exp \left( \BIGOMEGA {\frac{ r \log^2\alpha }{ \alpha \log^2 m }} \right)$
  to refute $\GFPHP$.
\end{theorem}

\ifthenelse{\boolean{false}}{}{
As promised in \refsec{sec:filterandclosure}, let us briefly discuss
the parameter~$\alpha$.}
Note that, on the one hand, the larger $\alpha$ is,
the more relaxed we can be \wrt the expansion requirements, and hence
the set of formulas to which the lower bound applies becomes larger. 
On the other hand, the strength of the lower bound deteriorates\ifthenelse{\boolean{false}}{}{ quickly }with~$\alpha$. Hence, we need to choose~$\alpha$ carefully to
find a good compromise between these two concerns.
  
\begin{proof}[Proof of \refthm{thm:main}]
  Let $\diff = \frac{\log \alpha}{\cdiff \log m}$ and let
  $w_0 = \frac{\varepsilon_0 r \diff}{ \alpha }$ for some small enough
  $\varepsilon_0 > 0$. We note that the choice of parameters and the
  condition on~$\alpha$ ensure that $4 \le \alpha^2 \le w_0$.
  Furthermore, in terms of $\diff$, the graph
  $G$ is an $(r, \GdegMax, (1 - 2 \diff) \GdegMax)$-boundary expander.
  
  We proceed by contradiction. Suppose $\refpi$ is a resolution refutation 
  with $\length{\refpi} < 2^{ \varepsilon' w_0 \diff }$ for a small enough constant
  $\varepsilon'  > 0$.
  Applying \Reflem{lem:fphp_upper_bound} we get a 
  set of \fakeaxiom{s} $\Ax$ with $\setsize{\Ax} \le \length{\refpi}$
  and a
  resolution refutation~$\refpi'$ of
  $\GFPHP \cup \Ax$ such that 
  $\width{\refpi'} \le K  \alpha  w_0$ for some large enough
  constant~$K$.  
  
Note that $\setsize{\Ax} \le \length{\refpi} < 2^{ \varepsilon' w_0
    \diff } \le (1 + \diff)^{w_0} $ for $\varepsilon' < 1/2$. 
  Applying
  \reflem{lem:fphp_lower_bound} to $\refpi'$
  yields
  a \pseudowidth
  lower bound of $r \diff / 4$. We conclude that 
  \begin{equation}
    r \diff/4 \le \width{\refpi'} 
    \le K  \alpha  w_0 
    = \varepsilon_0  K r \diff
    \eqperiod
  \end{equation}
  Choosing $\varepsilon_0 < \frac{1}{4K}$ yields a contradiction.
\end{proof}

The following corollary summarizes our claims for random graphs.

\begin{corollary}\label{cor:random_main}
  Let $m$ and $n$ be 
  positive integers and let 
  $\funcdescr{\gdegmax}{\Nplus}{\Nplus}$ and
  $\funcdescr{\exponentalpha}{\Nplus}{[0,1]}$ 
  be any monotone functions of $n$ such that 
  $n < m \le n^{\left( \exponentalpha/\constFPHPnew \right)^2 \log n}$ 
  and 
  $n\ge 
  \GdegMax \ge 
  \left( \frac{\constFPHPnew \log m}{\exponentalpha \log n} \right)^2$. 
  Then asymptotically almost surely resolution requires length 
  $\exp{\bigl(\Omega\bigl(n^{1-\exponentalpha}\bigr)\bigr)}$ 
  to refute $\GFPHP$ 
  for $G \sim \Gdist \bigl(m,n,\GdegMax \bigr)$.
\end{corollary}

\begin{proof}
  Let us assume that
  $n^{\left( \exponentalpha/\constFPHPnew \right)^2 \log n}$ and
  $\bigl( ({\constFPHPnew \log m}) / ({\exponentalpha \log n} )
  \bigr)^2$ are integers.  Observe that if
  $G \sim \Gdist \left(m,n,\GdegMax \right)$ for
  $\GdegMax > \bigl( ({\constFPHPnew \log m}) / ({\exponentalpha \log
    n}) \bigr)^2 $, then
  we can sample a graph
  $G'$ from the distribution $\Gdist 
  \bigl(m,n, ( 
      ({\constFPHPnew \log m}) / ({\exponentalpha \log n}) 
      )^2 \bigr)
  $
  by choosing a random subset of appropriate size of each neighbourhood
  of a left vertex of~$G$ (and applying a restriction zeroing out the other
  edges). 
  Hence, we can restrict our attention
  to the case where 
  $\GdegMax= 
  \bigl( ({\constFPHPnew \log m}) / ({\exponentalpha \log n}) \bigr)^2
  $.
  Also, it is sufficient to prove the claim for 
  $m = n^{\left( \exponentalpha/\constFPHPnew \right)^2 \log n}$,
  since choosing $m$ smaller can only make the formula less
  constrained and hence makes the lower bound easier to obtain.

  We want to apply \reflem{lem:expander} for 
  $\chi = \alpha = n^{\exponentalpha/4}$
  and
  $\diff = \frac{\log \alpha}{\cdiff \log m}$.
  In order to do so, we need to verify the inequalities
  \begin{subequations}
    \begin{align}
      \label{eq:check_diff}
      \diff &< 1/2 \eqcomma \\
      \label{eq:check_diff-chi}
      \diff \ln \chi &\geq 2 \eqcomma \\
      \label{eq:check_diff-gdegmax-chi}
      \diff \gdegmax \ln \chi &\geq 4 \ln m \eqperiod
    \end{align}
  \end{subequations}
  For \refeq{eq:check_diff} we observe that 
  $\diff = \frac{16}{\exponentalpha \log n}$
  and
  since
  $n < n^{\left(\exponentalpha/\constFPHPnew \right)^2 \log n}$
  we see that
  $\frac{1}{\log n} < \left(\frac{\exponentalpha}{\constFPHPnew}\right)^2$. 
  Hence, the first condition
  holds for $n$ large enough.
  To check \refeq{eq:check_diff-chi}, we compute
  \begin{equation}
    \diff \ln \chi 
    =
    \frac
    { 16  }
    { \varepsilon \log n}
    \frac{\varepsilon \ln n}{4}
    \ge 2 \eqperiod
  \end{equation}
  For
  \eqref{eq:check_diff-gdegmax-chi},
  we observe that $\gdegmax = \log m$ and hence
  \begin{equation}
    \diff \gdegmax \ln \chi 
    =
    \frac{4}{\log e} \log m
    =  4 \ln m  \eqperiod
  \end{equation}
  We conclude 
  that
  asymptotically almost surely,
  $G \sim \Gdist \left(m,n,\GdegMax \right)$ is an
  $\bigl( n^{1 - \exponentalpha/2}, \GdegMax, 
  (1 - 2 \diff)\GdegMax \bigr)$-boundary expander.
  \refthm{thm:main} then gives a length lower bound
  of $\exp\bigl( \Omega \bigl( n^{1 - \exponentalpha} \bigr) \bigr)$,
  as required.
\end{proof}

The following two corollaries are simple consequences of 
\refcor{cor:random_main}, optimizing
 for different parameters. The first corollary gives the
strongest lower bounds, while the second
minimizes the degree.

\begin{corollary}
  Let $m, n$ be such that
  $m \le n^{\littleoh{\log n}}$.
  Then asymptotically almost surely
  resolution requires length 
  $\exp \bigl( \Bigomega{n^{1-\littleoh{1}}} \bigr)$ to refute $\GFPHP$ 
	for $G \sim \Gdist \left(m,n,\log m\right)$.
\end{corollary}

\begin{proof}
  Let $m = n^{f(n)}$, where $f(n) = \littleoh{\log n}$. Applying
  \refcor{cor:random_main} for 
  $\exponentalpha = 16 \sqrt{\frac{f(n)}{\log n}} = \littleoh{1}$
  we get the desired statement.
\end{proof}

\begin{corollary}[Restatement of \refth{th:informal-poly-npigeons}]
  Let $k$ and $n$ 
  be positive integers and
  let $m = n^k$
  and
  $\exponentalpha \in \Rplus$.
  Then asymptotically almost surely
  resolution requires length 
  $\exp \bigl(\Bigomega{n^{1-\exponentalpha}}\bigr)$ to refute $\GFPHP$ 
  for 
  $G \sim \Gdist 
  \left(m, n, \left( \frac{16 k}{\exponentalpha} \right)^2 \right)$.
\end{corollary}

\begin{proof}
  We appeal to \refcor{cor:random_main} with 
  $\gdegmax = \left( \frac{16 k}{\varepsilon} \right)^2$,
  $m = n^{k}$ and 
  $\exponentalpha$ constant.
  A short calculation shows that all conditions are met.
\end{proof}

Our final corollary shows that we can get 
meaningful lower bounds even
for
a weakly exponential number of pigeons. Unfortunately,
the statement does not hold for random graphs.

\begin{corollary}\label{cor:main_subexp}
  Let $\kappa < 3/2 - \sqrt{2}$ and $\exponentalpha > 0$ be constant and
  $n$ be integer. Then
  there is a family of 
  explicitly constructible
  graphs $G$ with 
  $m = 2^{\bigomega{n^\kappa}}$ and
  left degree $\Bigoh{\log^{1 / \sqrt{\kappa}}(m)}$
  such that resolution requires
  length
  $
  \exp
  \bigl(
  \Bigomega{n^{1-2\sqrt{\kappa}(2 - \sqrt{\kappa}) - \exponentalpha}}
  \bigr)
  $
  to refute $\GFPHP$.
\end{corollary}

\begin{proof}
  Let $G$ be the graph from Corollary \ref{cor:boundary_subexp} with
  $\nu = \frac{ 2\sqrt{\kappa} }{ 1 - 2\sqrt{\kappa} }$.
  An appeal to \refth{thm:main} using the graph~$G$ yields 
  the desired lower bound.
\end{proof}

\squeezesubsection{A \PSEUDOWIDTH Upper Bound for Graph FPHP Formulas
  with Extra Axioms}
\label{sec:fphp_upper_bound}

Let us now prove \reflem{lem:fphp_upper_bound}.  For this proof, let
us identify $\Vp$ with $[m]$.  For every clause $C$ in the refutation
$\refpi$, let $\vecrofdef{C}{m}$ be the vector where each coordinate
is given by
\begin{align}
  \riof{C}{i} = \FLOOR{\frac{\deg[G]{i} - \deg[C]{i}}{\adv{i}}} + 1 \eqcomma
\end{align}
where~$\adv{i} = \frac{\deg[G]{i} \log \alpha}{\log m}$ as in
\reflem{lem:fphp_upper_bound} and~$\deg[C]{i}$ denotes the $i$th
pigeon degree of~$C$, that is, the number of holes pigeon~$i$ may go
to such that the clause~$C$ is satisfied under the corresponding
partial assignment (see \refeq{eq:pig-deg}). We apply the filter lemma
(\reflem{lem:filter}) to the set of vectors
$\set{\vecrof{C} \mid C \in \refpi}$.  Denote by $\vecrdef{m}$ a
vector as guaranteed to exist by \reflem{lem:filter}.  Let
\begin{equation}
  \label{eq:def-d-i}
  \thr{i} = \deg[G]{i} - \ceiling{\adv{i} \ri{i}} + 1 \eqperiod
\end{equation}
A short calculation establishes that $\thr{i}$ is the smallest integer such that 
$\Floor{\frac{\deg[G]{i} - \thr{i}}{\adv{i}}} + 1 \le \ri{i}$.

Note that every pigeon $i \in [m]$
such that $\riof{C}{i} \le \ri{i}$
is 
\sheavy
for $C$. Also, every 
\heavy pigeon of a clause $C$
satisfies that $\riof{C}{i} \le \ri{i} + 1$. 

To obtain a refutation $\refpi'$ that satisfies the conclusions of the lemma,
we consider every clause $C \in \refpi$ and either there is a
\fakeaxiom which is a strengthening of~$C$ that can be added to~$\Ax$ or
the \pseudowidth of~$C$ is small enough and the clause can thus remain
in~$\refpi'$.  More concretely, we make a case distinction whether
$\vecrof{C}$ satisfies case 1 of \reflem{lem:filter} or only case
2. In the former case
there is a \fakeaxiom which is a strengthening of~$C$, while in the
other
the \pseudowidth of $C$ is bounded: \begin{enumerate}
  \item
    $C$ satisfies 
    $\Abs{ \setdescr{ i \in [m]}{ \riof{C}{i} \le \ri{i} } } \ge
    w_0$: As every pigeon $i \in [m]$
    with $\riof{C}{i} \le \ri{i}$
    also satisfies $\deg[C]{i} \ge \thr{i}$, there is a \fakeaxiom
    which is a strengthening of~$C$ that can be added to~$\Ax$. This
    reduces the \pseudowidth of this clause to $w_0$.
  \item
    $C$ satisfies 
    $\Abs{ \setdescr{ i \in [m] }{ \riof{C}{i} \le \ri{i} + 1 } } \le
    O(\alpha \cdot w_0)$:  As every \heavy pigeon always
    satisfies $\riof{C}{i} \le \ri{i} + 1$, the \pseudowidth of $C$ is
    $\bigoh{\alpha \cdot w_0}$.
\end{enumerate}
This concludes the proof as 
$\setsize{\Ax} \leq \length{\refpi}$ 
and the \pseudowidth of $\refpi'$ is
$\bigoh{\alpha \cdot w_0}$ by construction.

\squeezesubsection{A \PSEUDOWIDTH Lower Bound for
  Graph FPHP Formulas with Extra Axioms}
\label{sec:fphp_lower_bound}

We continue to the proof of \reflem{lem:fphp_lower_bound}. For
convenience, we restate the lemma here.


\lemmafphplower*

Using \refdef{def:closure}, we define the set of
\introduceterm{relevant} pigeons of a clause~$C$ as
\begin{equation}
  \closure{C} = \closure[r, (1 - 3\diff)\gdegmax]{\Pthick{C}} 
  \eqcomma
\end{equation}
where 
$\Pthick{C} $
denotes the set of
$(\vecthr, \vecadv)$-\heavy pigeons for~$C$ as defined in
\refdef{def:pseudo-width}.
By definition,
the closure of a set $T$ contains $T$ itself 
but is only defined if $\setsize{T} \le r$.
However, if $ \Setsize{\Pthick{C}} \geq r \geq r \diff/4 $ then we
already have the lower bound claimed in the lemma, and so we may
assume that the closure is well-defined for all clauses in the
refutation~$\refpi$.  This implies, in particular, that for every
clause $C \in \refpi$ we have $\Pthick{C} \subseteq \closure{C}$.

Let us next construct the linear space $\linspace$ and describe how
matchings are mapped into it.  Fix a field $\mathbb{F}$ of characteristic~$0$ and for each pigeon
$i \in \Vp$ let $\linspace[i]$ be a linear space over $\mathbb{F}$ of
dimension $\dimLi$. Let $\linspace$ be the tensor product
$\linspace = \bigtensor_{i \in \Vp} \linspace[i]$ and denote by
$\linmap[i]: \Vh \mapsto \linspace[i]$ a function with the
property\footnote{It is readily seen that such functions exist by,
  e.g., considering mappings that send each hole~$j \in \Vh$ to a
  uniformly sampled unit vector.} that any subset of holes
$J \subseteq \Vh$ of size at least $\dim(\linspace[i])$ spans
$\linspace[i]$. In other words, for $J$ as above
we have that $\linspace[i] = \spans{\linmapof[i]{j}:j\in J}$.
This is how we will realize the idea of ``lossy counting''. For
$J \subseteq \Vh$ such that $\setsize{J} \le \dim(\linspace[i])$ we
have exact counting
$ \dim(\spans{\setdescr{\linmap[i](j)}{j \in J}}) = \setsize{J} $ as
otherwise we can construct a set of size~$\dim(\linspace[i])$ that
does not span the entire space, but
when $ \setsize{J} > \dim(\linspace[i]) $ gets large enough we have
$ \dim(\spans{\setdescr{\linmap[i](j)}{j \in J}}) = \dim(\linspace[i])
$. We crucially rely on this notion of ``lossy counting'' in the proof
of \reflem{lem:fphp_span}.

In order to map functions $\Vp \mapsto \Vh$ into $\linspace$, we define $\linmap : \Vh^{\Vp} \mapsto \linspace$ by 
$\linmapof{ j_1, \ldots, j_m} = \bigtensor_{i\in \Vp}
\linmapof[i]{j_i}$,
where will we abuse notions slightly in that we identify a vector with the
$1$\nobreakdash-di\-men\-sion\-al space spanned by this vector.
For a partial function $\matcha : \Vp \mapsto \Vh$, 
we
let $\linmapof{\matcha}$ 
be the span of all total extensions of $\matcha$ (not necessarily matchings), 
or equivalently
\begin{equation}
  \label{eq:lanbda-map-def}
  \linmapof{\matcha} = \bigtensor_{i \in \dom{\matcha}} 
  \linmapof[i]{\matcha_i} \tensor 
  \bigtensor_{ i \not \in \dom{\matcha}}\linspace[i] \eqperiod
\end{equation}
Recall that $\matchings$ is the set of all partial matchings on the
graph $G$ and that we interchangeably think of partial matchings as
partial functions $\funcdescr{\matcha}{\Vp}{\Vh}$ or as Boolean
assignments as defined in \refeq{eq:matching-to-assignment}.
For each clause $C$, we are interested in the partial matchings
$\matcha \in \matchings$ with domain $ \dom{\matcha} = \closure{C} $
such that $\matcha$ does not satisfy $C$.  We refer to the set of such
matchings as the \introduceterm{zero space} of $\clc$ and denote it by
\begin{equation}\label{eq:def_null}
  \Null{C}
   = \set{\matcha \in \matchings \mid \dom{\matcha} = \closure{C} \wedge
     \restrict{C}{\matcha} \neq 1}
   \eqperiod
\end{equation}
We associate $C$ with the linear space
\begin{equation}
  \linmapof{C} = 
  \spans{ \setdescr{\linmapof{\matcha} }{ \matcha \in \Null{C}}} \eqperiod  
\end{equation}
Note that contradiction is mapped to $\linspace$, \ie
$\linmapof{\bot} = \linspace$. 

We assert that the span of the axioms
$\spans{\setdescr{\linmapof{A}}{A \in \GFPHP \cup \Ax}}$ 
is a proper subspace of $\linspace$.

\begin{lemma} \label{clm:fphp_spanL} \label{lem:fphp_spanL}
  If $\setsize{\Ax} \le (1 + \diff)^{w_0}$, then 
  $\spans{\setdescr{\linmapof{A}}{A \in \GFPHP \cup \Ax}} \subsetneq \linspace$.
\end{lemma}

Accepting this claim without proof for now, this implies that in
$\refpi$ there is some resolution step deriving $C$ from $C_0$
and~$C_1$ where the subspace of the resolvent is not contained in the
span of the subspaces of the premises, or in other words
$ \linmapof{C}
\nsubseteq 
\spans{\linmapof{C_0}, \linmapof{C_1}}
$.
Our next lemma, which is the heart of the argument,
says that this cannot happen 
as long as the closures of the clauses are small. 

\begin{lemma}\label{lem:fphp_span}
  Let $C$ be a clause derived from clauses $C_0$ and $C_1$ by the
  resolution rule. If it holds that
  \begin{align*}
    \maxofexpr{
    \setsize{\closure{C_0}},
    \setsize{\closure{C_1}},
    \setsize{\closure{C}}
    }
    \leq r/4 \eqcomma
  \end{align*}
  then
  $\linmapof{C} \subseteq \spans{\linmapof{C_0}, \linmapof{C_1}}$.
\end{lemma}

Since contradiction cannot be derived while the closure is of size at
most $r/4$, any refutation~$\refpi$ must contain a clause~$C$ with
$\setsize{ \closure{C} } > r/4$.  But then \reflem{lem:closure_size}
implies that $C$ has \pseudowidth at least $r \diff / 4$, and
\reflem{lem:fphp_lower_bound} follows.  All that remains for us is to
establish \reftwolems {lem:fphp_spanL} {lem:fphp_span}.

\begin{proof}[Proof of  \reflem{lem:fphp_spanL}]
  We need to show that the
  axioms $\GFPHP \cup \Ax$ 
  do not span all of $\linspace$. We start with the axioms in $\GFPHP$.

  Let $A$ be pigeon axiom $\pigeonclause{i}$ as in~\eqref{eq:axiom-pigeon}
  or a functionality axiom $\functionclause{i}{j}{j'}$ 
  as in~\eqref{eq:axiom-functionality}. 
  Note that~$i$ is a \heavy pigeon for $A$. Clearly, there are no pigeon-to-hole
  assignments for pigeon $i$ that do not satisfy $A$. Thus there are no 
  matchings on $\closure{A}$ that do not satisfy $A$. We conclude that 
  $\linmapof{A} = \emptyset$.
  If instead $A$ is a hole axiom $\holeclause{i}{i'}{j}$ as
  in~\eqref{eq:axiom-hole}, 
  then we can  observe that
  $\deg[G]{i} - 1 \ge \thr{i} - \adv{i}$ since $\adv{i} = \cdiff \diff
  \deg[G]{i} \geq 2 \diff \gdegmax \ge 1$ (by boundary expansion). 
  This
  implies that $A$ has two \heavy pigeons.
  Observe that there are no matchings on these two pigeons
  that do not satisfy $A$. Thus $\Null{A} = \emptyset$ and
  we conclude that $\linmapof{A} = \emptyset$.

  Now consider the \fakeaxiom{}s in $\Ax$. We wish to show that any
  $A \in \Ax$ can only span a very small fraction of~$\linspace$.  We
  can estimate the number of dimensions $\linmapof{A}$ spans by
  \begin{align}
    \dim \linmapof{A}
    \le \prod_{i \notin \Pfat{A}} \dim \linspace[i] 
    \prod_{i \in \Pfat{A}} (\deg[G]{i} - \deg[A]{i})
    \le \prod_{i \notin \Pfat{A}} \dim \linspace[i] 
    \prod_{i \in \Pfat{A}}  (\deg[G]{i} - \thr{i}) \eqcomma
  \end{align}
  since all pigeons~$i \in \Pfat{A}$ are \sheavy, that is, it holds
  that~$\deg[A]{i} \geq \thr{i}$.
  Hence the fraction of the space $\linspace$ that $A$ may span is bounded by
  \begin{equation}
    \frac{\dim \linmapof{A}}{\dim \linspace} 
    \le \prod_{ i \in \Pfat{A} } \frac{\deg[G]{i} - \thr{i}}{\dimLi}
    \le \left( 1 - \diff \right)^{w_0} \eqperiod
  \end{equation}
  As $\setsize{\Ax} \le  \left( 1 + \diff \right)^{w_0}$
  we can conclude that
  not all of $\linspace$ is spanned by the axioms.
\end{proof}

\begin{figure}
  \begin{center}
    \includegraphics{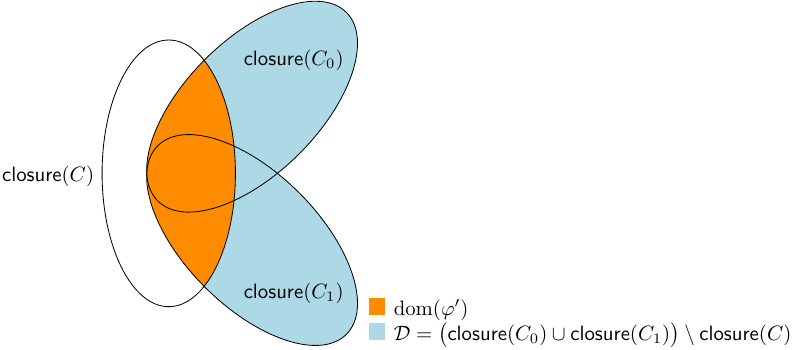}
  \end{center}
  \caption{Depiction of relations between 
    $\closure{C}, 
    \closure{C_i}, i = 1,2, \dom{\matchb}$ and $\diffclosure$ in proof
    of \reflem{lem:fphp_span}.}
  \label{fig:span}
\end{figure}

\begin{proof}[Proof of \reflem{lem:fphp_span}]
  Let us write 
  $S_{01} = \closure{C_0} \cup \closure{C_1}$ and $S = \closure{C}$.
  In order to establish
  the lemma,
  we need to show   for all $\matcha \in Z(C)$ that
  \begin{equation}
    \linmapof{\matcha} \subseteq 
    \spans{ \linmapof{C_0}, \linmapof{C_1} } 
    \eqperiod
  \end{equation}
  Denote by $\matchb$ the restriction of $\matcha$ to the domain 
  $S \cap S_{01}$
  and note
  that $C$ is not satisfied
  under the Boolean assignment associated with~$\matchb$.
  Also, observe that  if a matching $\matche$ extends a matching $\matchf$, then 
  $\linmapof{\matche}$ is contained in $\linmapof{\matchf}$. This is so since
  for any pigeon
  $i \in \dom{\matche} \setminus \dom{\matchf}$
  we have  from~\refeq{eq:lanbda-map-def}
  that $\matchf$ picks up the whole subspace~$\linspace[i]$
  while $\matche$ only gets a single vector.  
  Thus, if we can show that
  $\linmapof{\matchb} \subseteq \spans{ \linmapof{C_0}, \linmapof{C_1}
  }$, then we are done as $\matcha$ extends $\matchb$ and hence
  $\linmapof{\matcha} \subseteq \linmapof{\matchb}$. For the following
  argument it may be helpful to refer to the illustration in
  \reffig{fig:span}.

  Let $\diffclosure = S_{01} \setminus S$ and denote by 
  $\matchings[\diffclosure]$ 
  the set of matchings
  that extend $\matchb$ to the domain $\diffclosure$ and do not satisfy $C$. 
  Since each matching $\matchc \in \matchings[\diffclosure]$ 
  fails to satisfy $C$,
  by the soundness of the resolution rule we have that it also fails
  to satisfy either     
  $C_0$ or~$C_1$. Assume without loss of generality that $\matchc$
  does not satisfy~$C_0$ and denote
  by $\matchd$ the restriction of $\matchc$ to the domain of
  $\closure{C_0}$. 
  From \refeq{eq:def_null} we see that
  $\matchd \in Z(C_0)$ 
  and therefore
  $\linmapof{\matchc} \subseteq \linmapof{\matchd} \subseteq
  \linmapof{C_0}$.
  
  So far we have argued that for all matchings $\matchc \in \matchings[\diffclosure]$ 
  it holds that
  $\linmapof{\matchc}$ is contained in $\spans{\linmapof{C_0}, \linmapof{C_1}}$.
  Let $\linmapof{\matchings[\diffclosure]} = 
  \spans{\linmapof{\matchc} \mid \matchc \in \matchings[\diffclosure]}$.
  If we can show that the set of matchings
  $\matchings[\diffclosure]$ 
  is large enough for 
  $\linmapof{\matchings[\diffclosure]} = \linmapof{\matchb}$
  to hold, then the lemma follows.
  In other words, we want to show that
  $\linmapof{\matchings[\diffclosure]}$ 
  projected to 
  $\linspace[\diffclosure] = 
  \bigtensor_{i \in \diffclosure} \linspace[i]$
  spans all of the space $\linspace[\diffclosure]$.

  To argue this, note first
  that $\diffclosure$ is completely outside the $\closure{C}$.
  Furthermore,
  by assumption we have
  $\setsize{\closure{C}} \leq  r/4$ and  
  $
  \setsize{ \diffclosure } 
  \leq
  \setsize{ S_{01} } 
  \leq 
  r/2
  $. 
  An application of
  \reflem{lem:closure_expansion} 
  now tells us that
  \begin{equation}
    \abs{ \uniqueNeigh[G \setminus (\closure{C} \cup \neigh{\closure{C}})]
      {\diffclosure} } \ge 
    (1 - 3 \diff) \GdegMax \setsize{\diffclosure} 
    \eqperiod
  \end{equation}    
  By an averaging argument,
  there must exist a pigeon $i_1 \in \diffclosure$
  that has more than $(1 - 3 \diff) \GdegMax$ 
  unique neighbours in 
  $\uniqueNeigh[G \setminus 
  (\closure{C} \cup \neigh{\closure{C}})]{\diffclosure}$.
  The same argument applied to   
  $\diffclosure \setminus \set{ i_1 }$
  show that some pigeon $i_2$
  has more than $(1 - 3 \diff) \GdegMax$ 
  unique neighbours on top of the neighbours reserved for
  pigeon~$i_1$.
  Iterating this argument, 
  we derive by induction that for each pigeon 
  $i \in \diffclosure$  we can find
  $(1 - 3 \diff) \GdegMax$ distinct holes
  in $\neigh{\diffclosure}$.
  Since all pigeons in  $\diffclosure$ are \light in~$C$, 
  it follows that at most $\thr{i} - \adv{i}$ mappings of pigeon~$i$
  can satisfy the clause $C$. Hence, there are at least
  \begin{align}
    (1 - 3 \diff) \GdegMax - (\thr{i} - \adv{i})
    &\geq 
      (1 - 3 \diff) \deg[G]{i} - \thr{i} + 4 \diff \deg[G]{i}
      \nonumber \\
    &\geq
      \deg[G]{i} - \thr{i} + \adv{i}/4
  \end{align}
  many holes to which  each pigeon in~$\diffclosure$
  can be sent, independently of all other pigeons in~$\diffclosure$,
  without satisfying $C$. As we 
  have that
  $\dim(\linspace[i]) = \dimLi$, we conclude that
  $\linmapof{\matchings[\diffclosure]}$ projected to
  $\linspace[\diffclosure]$ spans the whole space.
  This concludes the proof of the lemma.
\end{proof}

\newcommand{\closureconstant}{20}
\newcommand{\closureconplus}{28}

\section{Lower Bounds for 
  Perfect Matching Principle
    Formulas}
\label{sec:pm}

In this section, we show that the perfect matching principle formulas
defined over even highly unbalanced bipartite graphs 
require exponentially long resolution refutations if the graphs are
expanding enough.

Just as in~\cite{Razborov04ResolutionLowerBoundsPM}, our proof is by an
indirect reduction to the  FPHP lower bound, and therefore there is a
significant overlap in concepts and notation with~\refsec{sec:GFPHP}.
However, since there are also quite a few subtle shifts in meaning, 
we restate all definitions in full below to make the exposition in
this section self-contained and unambiguous.

We first review some useful notions
from~\cite{Razborov01ImprovedResolutionLowerBoundsWPHP}. 
Let
$G = (\V, E)$
denote the graph over which the formulas are defined.
For a 
clause~$C$ and a 
vertex $v \in V(G)$,
let the \introduceterm{clause-neighbourhood of $v$ in $C$}, denoted by $\neigh[C]{v}$,  be
the vertices $u \in V(G)$ with the property that $C$ is satisfied if
$v$ is matched to $u$, that is,
\begin{equation}
	\neigh[C]{v} = \setdescr{ u \in \V }{ e=\set{u, v} \in 
	E~\text{and}~\matchboolof{\set{e}}{C} = 1 } 
   \eqperiod
\end{equation}
For a set $V \subseteq V(G)$ let $\neigh[C]{V}$ be the union of the
clause-neighbourhoods of the vertices in $V$, \ie
$
  \neigh[C]{V} = \bigcup_{v \in V}\neigh[C]{v}
$ and let the \introduceterm{$v$th vertex degree of $C$} be
\begin{equation}
  \deg[C]{v} = \setsize{ \neigh[C]{v} } 
\eqperiod
\end{equation}
We think of a vertex $v$ with large 
degree~$\deg[C]{v}$ 
as a vertex on which the derivation has
not made any progress up to the point of deriving $C$, since the
clause rules out very few
neighbours.
The vertices with high enough vertex degree in a clause are the 
\emph{\heavy vertices} of the clause as defined next.

\begin{definition}[Vertex weight, \pseudowidth and \fakeaxiom{s}]
  \label{def:onto-pseudo-width}
  Let $\vecthrdef{m + n}$ and $\vecadvdef{m + n}$ be two vectors such that
  $\vecthr$ is elementwise greater than $\vecadv$. We say that a vertex $v$
  is \introduceterm{$\vecthr$-\sheavy for $C$} 
  if $\deg[C]{v} \ge \thr{v}$
  and that vertex~$v$ is
  \introduceterm{$(\vecthr, \vecadv)$-\heavy for $C$}
  if $\deg[C]{v} \ge \thr{v} - \adv{v}$. When $\vecthr$ and $\vecadv$ 
  are understood from context
  we omit the parameters and just refer to 
  \introduceterm{\superheavy} and \introduceterm{\heavy} vertices.
  Vertices that are not \heavy are referred to as 
  \emph{\light vertices}.
  The set of vertices that are \sheavy for $C$ is denoted by
  \begin{equation}
    \Vfat{C} = \setdescr{v \in \V}{\deg[C]{v} \ge \thr{v}}
  \end{equation}
  and the set of \heavy vertices for  $C$ is denoted by 
  \begin{equation}
    \Vthick{C} = \setdescr{v \in \V}{\deg[C]{v} \ge \thr{v} - \adv{v}} 
    \eqperiod
  \end{equation}
  The 
  \introduceterm{\pseudowidth}
  $\width{C} = \setsize{ \Vthick{C} }$ of a clause~$C$
  is the number of \heavy vertices in it, and
  the \pseudowidth of a resolution 
  refutation~$\refpi$ is
  $\width{\refpi} =
  \max_{C \in \refpi}{\width{C}}$.
  We refer to clauses $C$
  with precisely~$w_0$~\sheavy vertices
  as \introduceterm{\fakeaxiom{s}}.
\end{definition}

To a large extent, the proof of the
lower bounds for perfect matching formulas follows the general idea of the proof of \refthm{thm:main}:
given a short refutation we first apply the filter lemma to obtain a refutation
of small \pseudowidth; we then prove that
in small \pseudowidth contradiction cannot be derived
and can thus conclude that no short refutation exists.
In more detail, 
given a short resolution refutation~$\refpi$,
we use the filter lemma (\reflem{lem:filter}) 
to get a filter vector $\vecthrdef{m+n}$
such that each clause either has many \sheavy vertices or
not too many \heavy vertices (for an appropriately chosen vector $\vecadv$).
Clearly, clauses that fall into the second case of the filter lemma
have bounded \pseudowidth. Clauses in the first case, however,
may have very large \pseudowidth. In order to obtain a proof of
low \pseudowidth, these latter clauses are strengthened to \fakeaxiom{s}
and added to a special set $\Ax$. This then gives a refutation
$\refpi'$ that refutes the formula $\PM \cup \Ax$ in bounded
\pseudowidth
as stated in the next lemma.

\begin{lemma}\label{lem:ontofphp_upper_bound}
  Let $G=( \Vl \disjointunion \Vr, E)$ be a bipartite graph 
  with $\setsize{\Vl} = \numpigeons$ and 
  $\setsize{\Vr} = \numholes$;
  let $\refpi$ be a resolution refutation of $\PM$;
  let $w_0, \alpha \in [\numpigeons + \numholes]$ be such that 
  $w_0 > \log \length{\refpi}$ and
  $w_0 \ge \alpha^2 \ge 4$,
  and let $\vecadvdef{\numpigeons + \numholes}$ be defined by
  $\adv{v} = \frac{\deg[G]{v} \log \alpha}
  {\log (\numpigeons + \numholes)}$ for $v\in V(G)$.
  Then there exists 
  an integer vector $\vecthrdef{\numpigeons + \numholes}$, with 
  $\adv{v} < \thr{v} \le \deg[G]{v}$ for all 
  $v \in V(G)$, a set of \fakeaxiom{s} $\Ax$ with 
  \mbox{$\setsize{\Ax} \le \length{\refpi}$}, and a resolution refutation 
  $\refpi'$ of $\PM \cup \Ax$ such that 
  $\length{\refpi'} \leq \length{\refpi}$ and
  \mbox{$\width{\refpi'} \le \bigoh{\alpha \cdot w_0}$.}
\end{lemma}

The proof of the above lemma is omitted as it is syntactically equivalent
to the proof of \reflem{lem:fphp_upper_bound}.
Until this point, we have almost mimicked the proof of \refthm{thm:main}.
The main differences will appear in the proof of the
counterpart to \reflem{lem:ontofphp_upper_bound}, 
which states a \pseudowidth lower bound.

\begin{lemma}\label{lem:onto_lower_bound}
  Assume for
  $\diff \le 1/64$ and $m,n,r,\gdegmax \in \N$
  that
  $G=(\Vl \disjointunion \Vr, E)$ 
  is an
  $(r, \GdegMax, \left(1 - 2\diff\right)\GdegMax)$-boundary expander
  with $\setsize{\Vl} = m$, $\setsize{\Vr} = n$,
  $\GdegMax \geq {\log m}/{\diff^2}$, and 
  $\minofset{\deg[G]{v}}{v\in \Vr} \geq {r}/{\diff}$.
  Let $\vecadv = (\adv{v} \mid v \in V(G))$ be defined by
  $\adv{v} = \ontocdiff \deg[G]{v} \diff$ and suppose that
  $\vecthr = (\thr{v} \mid v \in V(G))$
  is an integer vector such
  that $\adv{v} < \thr{v} \leq \deg[G]{v}$ for all $v \in V(G)$.
  Fix $w_0$ such that $64 \leq w_0 \leq r \diff - \log n$ 
  and let $\Ax$ be
  an arbitrary set of \fakeaxiom{s} with
  $\setsize{\Ax} \le (1 + 16\diff)^{w_0/8}$. 
  Then every resolution refutation~$\refpi$ of
  $\PM \cup \Ax$ has either length
  $\length{\refpi} \ge 2^{w_0/32}$ or
  \pseudowidth
  $\width{\refpi} \ge r \diff$.
\end{lemma}

The proof of the above lemma is based on a sort of reduction to the
$\GFPHP$ case. The idea, due to Razborov \cite{Razborov04ResolutionLowerBoundsPM}, is to first pick a
partition of the vertices of $G$ that looks random to every clause in the 
refutation and then simulate the $\GFPHP$ lower bound on this partition.
In our setting, however, this process gets quite involved. 
Already implementing the partition idea of Razborov is
non-trivial: for a fixed clause $C$ some
vertices that are \light may be \sheavy with respect to the
partition, and we do not have an upper bound on the \pseudowidth any longer.
The insight needed to solve this issue 
is to show that
by expansion there
are not too many such vertices per clause, and then adapt the closure
definition to take these vertices into account.

Another issue we run into is that the span argument from
\refsec{sec:GFPHP} cannot be applied to all the vertices in the graph
since vertices in $\Vr$ are not guaranteed to be expanding. Instead,
for the vertices in $\Vr$, we need to resort to the span argument from
\cite{Razborov03ResolutionLowerBoundsWFPHP}. Moreover, vertices in the neighbourhood of $\diffclosure$
(as defined in the proof of \reflem{lem:fphp_span})
may already be matched, and we are hence unable to attain enough
matchings. 
Our solution is to consider a ``lazy'' edge removal procedure
from the original matching, which with a careful analysis can be shown
to circumvent the problem\ifthenelse{\boolean{false}}{. 
 We refer to the full-length version of this paper for the proof of \reflem{lem:onto_lower_bound}.
}{---see \refsec{sec:onto_lower_bound} for
details.}

\squeezesubsection{Formal Statements of Perfect Matching Formula Lower Bounds}
\label{sec:onto_results}

\ifthenelse{\boolean{false}}{
Let us state the lower bounds we obtain for the perfect matching formulas. }{
Let us state our lower bounds for the perfect matching formulas and
defer the proof of \reflem{lem:onto_lower_bound} to Section~\ref{sec:onto_lower_bound}.
}
\begin{theorem} \label{thm:onto_main}
  Let
  $G = (U \disjointunion V, E)$ be a  bipartite graph with
  $m = \setsize{U}$ and $n = \setsize{V}$.
  Suppose that  $G$ is an
  $\left(r, \GdegMax, \left(1 - 2\diff \right) \GdegMax \right)$-boundary
  expander for 
  $\GdegMax \geq \frac{\log(\numpigeons + \numholes)}{\diff^2}$ 
  and
  $\diff = \frac{\log \alpha}{ 64 \log (\numpigeons + \numholes)}$
  where
  $\alpha \geq 2$ and
  $\frac{\alpha^3}{\log \alpha} = \LITTLEOH{\frac{r}{\log (\numpigeons
      + \numholes)}}$,
  which furthermore satisfies the degree requirement
  $\minofset{\deg[G]{v}}{v \in V} \geq r/\diff$.
  Then resolution requires length 
  $\exp \left( 
    \BIGOMEGA{\frac{ r \log^2\alpha }
      { \alpha \log^2 (\numpigeons + \numholes) }} \right)$
  to refute the perfect matching formula $\PM$ defined over~$G$.
\end{theorem}

We remark that this 
theorem also holds if we replace the minimum degree constraint of~$V$
with an expansion guarantee from~$V$ to~$U$.
We state the theorem
in the above form as we want to apply it to the graphs
from~\cite{GUV09Unbalanced} for which we have no expansion guarantee
from~$V$ to~$U$.

\begin{proof}[Proof of Theorem \ref{thm:onto_main}]
  Let
  $w_0 = \frac{\varepsilon_0 r \diff}{ \alpha }$, 
  for some small enough $\varepsilon_0 > 0$. 
  Suppose for the sake of contradiction that $\refpi$ is a resolution refutation 
  of~$\PM$ such that 
  $\length{\refpi} < (1 + 16\diff)^{w_0/8}$. Since $w_0 > \log \length{\refpi}$, by \reflem{lem:ontofphp_upper_bound} we have that
  there exists 
  an integer vector $\vecthrdef{\numpigeons + \numholes}$, with 
  $\adv{v} < \thr{v} \le \deg[G]{v}$,
  a set of \fakeaxiom{s} $\Ax$
  with 
  $\setsize{\Ax} \le \length{\refpi} < (1 + 16\diff)^{w_0/8}$, 
  and a resolution refutation~$\refpi'$ of $\PM \cup \Ax$ such that 
  $\length{\refpi'} \leq \length{\refpi}$  
  and
  $\width{\refpi'} \le K{\alpha  w_0}$ for some large enough constant $K$.
  Since $\length{\refpi'} < (1 + 16\diff)^{w_0/8} \leq 2^{w_0/32}$,
  by \reflem{lem:onto_lower_bound}, we have that 
  $\width{\refpi'} \ge r \diff \ge {\alpha w_0}/{\varepsilon_0}$. 
  Choosing $\varepsilon_0 < 1/K$, we get a contradiction
  and, thus, $\length{\refpi} \geq (1 + 16\diff)^{w_0/8} = 
  \exp\left(\BIGOMEGA{\frac{r \diff^2}{\alpha}}\right)$.
\end{proof}

As in \refsec{sec:GFPHP}, we have a general statement for random graphs.

\begin{corollary}\label{cor:onto_random_main}
  Let $m$ and $n$ be positive integers, 
  let
  $\funcdescr{\gdegmax}{\Nplus}{\Nplus}$ and
  $\funcdescr{\exponentalpha}{\Nplus}{[0,1]}$ be any monotone
  functions of $n$ such that 
  $n^3 < m \le n^{\left( \exponentalpha/\constonto \right)^2 \log n}$ 
  and 
  $n \ge 
  \gdegmax \ge 
  \log (m+n) 
  \left( \frac{\constonto \log (m + n)}{\exponentalpha \log n} \right)^2$. 
  Then asymptotically almost surely resolution requires length 
  $\exp{\bigl(\Omega\bigl(n^{1-\exponentalpha}\bigr)\bigr)}$ 
  to refute $\PM$ 
  for $G \sim \Gdist \bigl(m,n,\GdegMax \bigr)$.
\end{corollary}

\ifthenelse{\boolean{false}}{
\begin{proof}[Proof sketch.]
  It suffices to prove the claim for
  $m = n^{\left( \exponentalpha/\constonto \right)^2 \log n}$ and
  $\gdegmax = \log (m + n) \cdot \bigl( ( \constonto \log (m + n))/
  (\exponentalpha \log n) \bigr)^2$.  By applying
  \reflem{lem:expander} for $\chi = \alpha = n^{\exponentalpha/4}$ and
  $\diff = \frac{\log \alpha}{\ontocdiff \log m}$, we conclude that
  asymptotically almost surely,
  $G \sim \Gdist \left(m,n,\GdegMax \right)$ is an
  $\bigl( n^{1 - \exponentalpha/2}, \GdegMax, (1 - 2 \diff)\GdegMax
  \bigr)$-boundary expander.  Furthermore, by the Chernoff inequality
  (\refthm{thm:Chernoff}) asymptotically almost surely all right
  vertices have degree at least
  $n \cdot \frac{64 \log (m + n)}{\exponentalpha\log n}$. Thus,
  \refthm{thm:onto_main} gives a length lower bound of
  $\exp\bigl( \Omega \bigl( n^{1 - \exponentalpha} \bigr) \bigr)$ as
  claimed.
\end{proof}
}
{
\begin{proof}
  For simplicity, let us assume that $m^{+} = n^{\left( \exponentalpha/\constonto \right)^2 \log n}$
  and $\gdegmax^{-} = \log (m + n) \cdot \bigl( ( \constonto \log (m + n))/
  (\exponentalpha \log n) \bigr)^2$ are
  integers.
  It suffices to prove the claim for 
  $m = m^{+}$ and $\gdegmax = \gdegmax^{-}$.
  Indeed, if $G \sim \Gdist \left(m,n,\GdegMax \right)$,
  for 
  $\GdegMax > \gdegmax^{-} $,
  we can sample a random subgraph
  $G' \sim \Gdist 
  \bigl(m,n,\gdegmax^{-} \bigr)$ of $G$
  by choosing a random subset of appropriate size of each neighbourhood
  of a left vertex and applying a restriction zeroing out the other
  edges.
  Furthermore,
  as for smaller $m$ the formula gets less constrained
  and hence the lower bound is easier to obtain, it suffices to prove it for
  $m = m^{+}$.

  We want to apply \reflem{lem:expander} for 
  $\chi = \alpha = n^{\exponentalpha/4}$
  and
  $\diff = \frac{\log \alpha}{\ontocdiff \log m}$,
  and towards this end we argue that
  the inequalities
  \begin{subequations}
    \begin{align}
      \label{eq:onto_check_diff}
      \diff &< 1/2 \eqcomma \\
      \label{eq:onto_check_diff-chi}
      \diff \ln \chi &\geq 2 \eqcomma \\
      \label{eq:onto_check_diff-gdegmax-chi}
      \diff \gdegmax \ln \chi &\geq 4 \ln m
    \end{align}
  \end{subequations}
  all hold.
  First observe that
  $\diff = \frac{32}{\exponentalpha \log n}$
  and 
  $n < n^{\left(\exponentalpha/\constonto\right)^2 \log n}$,
  from which we conclude that
  $\frac{1}{\log n} < \left(\frac{\exponentalpha}{\constonto}\right)^2$. 
  Hence, the first inequality~\refeq{eq:onto_check_diff}
  holds for $n$ large enough.
  A simple calculation
  \begin{equation}
  \diff \ln \chi 
  =
  \frac
  { 32  }
  { \varepsilon \log n}
  \frac{\varepsilon \ln n}{4}
   \ge 2
  \end{equation}
  shows that~\refeq{eq:onto_check_diff-chi} is also true.
  Finally, for
  \eqref{eq:onto_check_diff-gdegmax-chi},
  we observe that $\gdegmax \ge \log^2 m$ and hence
  \begin{equation}
    \diff \gdegmax \ln \chi 
    \ge
    \frac{8}{\log e} \log^2 m
    \ge  4 \ln m  \eqperiod
  \end{equation}
  We conclude 
  that
  asymptotically almost surely
  $G \sim \Gdist \left(m,n,\GdegMax \right)$ is an
  $\bigl( n^{1 - \exponentalpha/2}, \GdegMax, (1 - 2 \diff)\GdegMax
  \bigr)$-boundary expander.  Furthermore, by the Chernoff inequality
  (\refthm{thm:Chernoff}) asymptotically almost surely all right
  vertices have degree at least
  $n \cdot \frac{64 \log (m + n)}{\exponentalpha\log n}$. Thus,
  \refthm{thm:onto_main} gives a length lower bound of
  $\exp\bigl( \Omega \bigl( n^{1 - \exponentalpha} \bigr) \bigr)$
  as claimed.
\end{proof}
}

The following corollary is a simple consequence of 
\refcor{cor:onto_random_main},
optimizing for the strongest lower bounds.

\begin{corollary}[Restatement of \refth{th:informal-quasipoly-npigeons}]
  Let $m, n$ be such that
  $m \le n^{\littleoh{\log n}}$.
  Then asymptotically almost surely resolution requires length 
  $\exp \bigl( \Bigomega{n^{1-\littleoh{1}}} \bigr)$ to refute $\PM$ 
  for $G \sim \Gdist \left(m,n, 8 \log^2 m\right)$.
\end{corollary}

\begin{proof}
  Let $m = n^{f(n)}$, where $f(n) = \littleoh{\log n}$. Applying
  \refcor{cor:onto_random_main} for 
  $\exponentalpha = 128 \sqrt{\frac{f(n)}{\log n}} = \littleoh{1}$,
  we get the desired statement.
\end{proof}

Our final corollary shows that we even get meaningful lower bounds
for highly 
unbalanced
bipartite graphs. As was the case for $\GFPHP$, 
the required expansion is too strong to hold for random graphs with such large imbalance, 
but does hold for explicitly constructed graphs from~\cite{GUV09Unbalanced}.

\begin{corollary}[Restatement of \refth{th:informal-exponential-npigeons}]
\label{cor:onto_main_subexp}
	Let $\kappa < 3/2 - \sqrt{2}$ and $\exponentalpha > 0$ be constants, and let
	$n$ be an integer. Then there is a family of 
	(explicitly constructible) graphs $G$ with $ m = 2^{\bigomega{n^\kappa}}$ and 
	left degree $O(\log^{1 / \sqrt{\kappa}}(m))$,
	such that resolution requires
	length
	$\exp(\Omega(n^{1-2\sqrt{\kappa}(2 - \sqrt{\kappa}) - \exponentalpha}))$
	to refute $\PM$.
\end{corollary}

\begin{proof}
  Let $G$ be the graph from Corollary~\ref{cor:boundary_subexp} with
  $\nu = \frac{ 2\sqrt{\kappa} }{ 1 - 2\sqrt{\kappa} }$. In order to apply
  \refthm{thm:onto_main} we need to satisfy the minimum right degree constraint.
  A simple way of doing this is by adding $n^2$ edges to $G$ such that
  each vertex on the right has exactly $n$ incident edges added while each vertex
  on the left has at most one incident edge added.  This
  will leave us with a graph which has large enough right degree while
  each left degree increased by at most one. The additional edges may
  reduce the boundary expansion a bit, but 
  a short calculation shows that by choosing $\diff = \frac{\log
    \alpha}{128 \log (m + n)}$ in Corollary~\ref{cor:boundary_subexp}, 
  we can still guarantee the needed
  boundary expansion for \refthm{thm:onto_main}.
  The corollary bound follows.
\end{proof}

\squeezesubsection{Defining Pigeons and Holes}
\label{sec:onto_random}

As stated earlier, we prove the $\PM$ lower bound by simulating the
$\GFPHP$ lower bound from \refsec{sec:GFPHP} on a partition
$\Vp \disjointunion \Vh$ of the vertices of $G$. As the notation
suggests, we think of the vertices in $\Vp$ as pigeons and of the
vertices in $\Vh$ as holes, that is, we will ignore all edges
contained in $\Vp$ or in $\Vh$ so that pigeons may only fly to holes
(formally we restrict the variables corresponding to such edges
to~$0$).

Let us first motivate the properties---captured in \reflem{lem:partition}---that  
such a partition must satisfy in order for the $\GFPHP$ simulation to
go through.
To begin with, recall that in the proof of \reflem{lem:fphp_spanL}
we show that a \fakeaxiom 
only spans an exponentially small fraction of the linear space
$\linspace$. 
The argument crucially relies on the fact that there are many \sheavy 
pigeons in every \fakeaxiom. 
To make this work over the partition $\Vp
\disjointunion \Vh$, we require that a constant fraction of the
\sheavy vertices of every
\fakeaxiom are in $\Vp$
and that \sheavy vertices remain \sheavy with
respect to this partition.
This first issue is addressed by Property~\ref{item:thick} of \reflem{lem:partition}
whereas the second issue is guaranteed by the other properties:
Property \ref{item:graph} ensures that for every vertex
roughly half of its neighbours are in $\Vh$ while Properties \ref{item:clause-r}
and \ref{item:clause-l} ensure that most clause-neighbourhoods behave
in the same manner, \ie up to a small set of vertices per clause
every clause-neighbourhood of a vertex has roughly half of its vertices in $\Vh$.
Combining these arguments, we can bound the fraction of the space spanned by a
\fakeaxiom.

The other main step of the $\GFPHP$ lower bounds is
\reflem{lem:fphp_span} which state that in low \pseudowidth the linear
space associated with a resolvent never leaves the span of the
premises.  This argument relies on the expansion guarantee of the
underlying graph and the fact that \light pigeons are unconstrained.
The required graph expansion (see \reflem{lem:expansion_Gprime}) 
will follow from 
Property~\ref{item:graph} and Properties
\ref{item:graph}--\ref{item:clause-l} are used to argue that
\light pigeons are also unconstrained with repect to the
partition.

\begin{lemma}\label{lem:partition}
  Let $G = ( \Vl \disjointunion \Vr, E)$ be an $(r, \GdegMax, \left( 1 - 2 \diff \right)
  \GdegMax)$-boundary expander for $\diff \leq 1/4$ and $\setsize{\Vl} \geq 4$. Fix $w_0$ such
  that $64\leq w_0 \leq r$ and let $\Ax$ be a set of \fakeaxiom{s}
  of size $\setsize{\Ax} \le \exp({w_0 / 32})$. Moreover, suppose that
  $\GdegMax \geq \log \setsize{\Vl}/\diff^2$ and
  $\minofset{\deg[G]{v}}{v \in \Vr} \geq (\log \setsize{\Vr} + w_0) / \diff^2$.    
  If $\refpi$ is a resolution refutation of $\PM \cup \Ax$ with $\length{\refpi} \leq
  \exp({w_0/32})$, then there exists a vertex partition $\V(G) = \Vp\, \dot \cup\, \Vh$ such that
  \begin{enumerate}
  \item for every $A \in \Ax$: \label{item:thick}\begin{align*}
      \abs{ \Vfat{A} \cap \Vp } \ge w_0 / 4 \eqcomma
    \end{align*}
  \item for every $v \in \V$: \label{item:graph}\begin{align*}
      \Abs{ \abs{ \neigh[G]{v} \cap \Vh } -
      1/2\abs{ \neigh[G]{v} } } \le
      4 \diff \abs{ \neigh[G]{v} }  \eqcomma
    \end{align*}
  \item for every $C \in \refpi$ and for every $v \in \Vr$: \label{item:clause-r}\begin{align*}
      \Abs{ \abs{ \neigh[C]{v} \cap \Vh } -
      1/2\abs{ \neigh[C]{v} } } \le
      4 \diff \abs{ \neigh[G]{v} }  \eqcomma
    \end{align*}
  \item for every $C \in \refpi$ there is a set of vertices $\Vbad{C}
    \subseteq \Vl$, satisfying $\setsize{ \Vbad{C} } \le w_0/8$, such
    that for every $v \in \Vl \setminus
    \Vbad{C}$: \label{item:clause-l}\begin{align*}
      \Abs{ \abs{ \neigh[C]{v} \cap \Vh } -
      1/2\abs{ \neigh[C]{v} } } \le
      4 \diff \GdegMax \eqperiod
    \end{align*}
    \label{item:badset}
  \end{enumerate}
\end{lemma}

The analogue of above lemma in
\cite{Razborov04ResolutionLowerBoundsPM} is Claim 19.  The main
difference is that in our setting Property~\ref{item:clause-l} does
not always hold for all vertices in the graph while in Razborov's
setting the corresponding property always holds.

The following is an auxiliary claim used to argue
Property~\ref{item:badset}, that is, that the error set $\Vbad{C}$ is
small. The claim states that if $G$ is a good expander and for a fixed
clause $C$ there are many vertices $v \in \Vl$ such that
$\setsize{\neigh[C]{v} \cap \Vh}$ does not behave as expected, then
there is a large set of vertices $\Vbad[\star]{C}$ whose
clause-neighbourhood in $\Vh$ (\ie the set
$\neigh[C]{\Vbad[\star]{C}} \cap \Vh$) deviates from its expected
size.

\begin{claim}\label{cl:non_expansion}
  Let $G=( \Vl \disjointunion \Vr, E)$ be an 
  $(r, \GdegMax, \left(1 - 2\diff \right)\GdegMax)$-boundary 
  expander.
  Fix any partition $V(G) = \Vp \disjointunion \Vh$ and any clause $C$.
  Let
  \begin{equation*}
    \Vbad{C} = \set{v \in \Vl :
      \Abs{ \abs{ \neigh[C]{v} \cap \Vh } - 
        1/2\abs{ \neigh[C]{v} } } > 
      4 \diff \GdegMax} \eqperiod
  \end{equation*}
  If $\setsize{ \Vbad{C} } > w_0/8$, then there is a set
  of vertices $\Vbad[\star]{C} \subseteq \Vbad{C}$, with 
  $\setsize{ \Vbad[\star]{C} } = w_0/16$, such that
  \begin{align*}
    \Abs{ \abs{ \neigh[C]{\Vbad[\star]{C} } \cap \Vh } -
    1/2\abs{ \neigh[C]{ \Vbad[\star]{C} } }} >
    2 \diff \GdegMax \setsize{ \Vbad[\star]{C} } \eqperiod
  \end{align*}
\end{claim}

\begin{proof}
  Denote by $\Vbad[+]{C}$ ($\Vbad[-]{C}$ respectively) the vertices
  in $\Vbad{C}$ that have more neighbours (fewer neighbours respectively) in 
  $\Vh$ than the expected $1/2\Abs{ \neigh[C]{v} }$. As
  $\setsize{ \Vbad{C} } > w_0/8$, one of the sets 
  $\Vbad[+]{C}$ or~$\Vbad[-]{C}$
  is of cardinality at least $w_0/16$. 
  \begin{itemize}
  \item[] \textbf{Case 1}: Suppose 
    $\Vbad[+]{C} \geq w_0/16$ and let
    $\Vbad[\star]{C}$ be any subset of $\Vbad[+]{C}$ of 
    size $w_0/16$.
    As boundary expansion of $G$ guarantees 
    that $\Vbad[\star]{C}$ has at most 
    $2 \diff \GdegMax \setsize{ \Vbad[\star]{C} }$ \emph{edges} 
    to non-unique neighbours in~$G$ we derive
    \begin{align}
      \abs{ \neigh[C]{ \Vbad[\star]{C} } \cap \Vh } 
      &\ge \sum_{v \in\Vbad[\star]{C}} 
        \setsize{ \neigh[C]{v} \cap \uniqueneigh{\Vbad[\star]{C}} \cap \Vh }\\
      &\ge \sum_{v \in\Vbad[\star]{C}} \setsize{ \neigh[C]{v} \cap \Vh } - 
        2 \diff \GdegMax \setsize{\Vbad[\star]{C}}\\
      &> \sum_{v \in \Vbad[\star]{C}} 
        \bigl( 1/2\setsize{ \neigh[C]{v} } + 
        4 \diff \GdegMax  \bigr) -
        2 \diff \GdegMax \setsize{\Vbad[\star]{C}}\\
      &\ge 1/2 \abs{ \neigh[C]{ \Vbad[\star]{C} } } + 
        2 \diff \GdegMax \setsize{\Vbad[\star]{C}} \eqcomma
    \end{align}
    where the strict inequality follows by definition of $\Vbad[+]{C}$.
  \item[] \textbf{Case 2}: Suppose 
    $\Vbad[-]{C}  \geq w_0/16$ and let
    $\Vbad[\star]{C}$ be any subset of $\Vbad[-]{C}$ of 
    size $w_0/16$. Similar to the previous case we can conclude that
    \begin{align}
      \abs{ \neigh[C]{ \Vbad[\star]{C} } \cap \Vh } 
      &\le \sum_{v \in \Vbad[\star]{C}} \setsize{ \neigh[C]{v} \cap \Vh }\\
      &< \sum_{v \in \Vbad[\star]{C}} 
        \bigl( 1/2\setsize{ \neigh[C]{v} } -
        4 \diff \GdegMax  \bigr) \\
      &\leq 1/2 \abs{ \neigh[C]{ \Vbad[\star]{C} } } -
        2 \diff \GdegMax \setsize{\Vbad[\star]{C}} \eqcomma
    \end{align} 
    where the last inequality uses that  $\Vbad[\star]{C}$ has at most 
    $2 \diff \GdegMax \setsize{ \Vbad[\star]{C} }$ edges incident
    to non-unique neighbours in~$G$.
\end{itemize}
  Combining both cases yields the claim.
\end{proof}

\begin{proof}[Proof of Lemma~\ref{lem:partition}]
  Pick a partition $\V = \mathbf{\Vp} \disjointunion \mathbf{\Vh}$
  uniformly at random.  In what follows we show that Property 1 holds
  with probability at least $3/4$ and Properties 2, 3 and 4 each hold
  with probability at least $7/8$.  Hence by a union bound there
  exists a partition that satisfies all four properties
  simultaneously.
  
  For the first property, since
  $\Expectation{ \setsize{ \Vfat{A} \cap \mathbf{\Vp} } } = w_0/2$, by
  the multiplicative Chernoff bound, that is, \refthm{thm:Chernoff}, we have
  that
  \begin{align}
    \Prob{ \setsize{ \Vfat{A} \cap \mathbf{\Vp} } \le w_0/4 } 
    \le \exp\left(-w_0 / 16\right) \eqperiod
  \end{align}
  Since $\setsize{\Ax} \leq \exp(w_0/32)$ and $w_0 \geq 64$,
  a union bound over $\Ax$ gives 
  us that Property 1 holds except with probability 
  $\exp(-w_0/32)\leq 1/4$.
	
  To analyse Properties 2 and 3,
  let $C$ either be a clause in $\refpi$
  or be the graph $G$ (\ie the clause that
  contains all variables)
  and fix an arbitrary $v \in V(G)$. By Chernoff bound\ifthenelse{\boolean{false}}{}{~(Theorem \ref{thm:Chernoff})} we get that 
  \begin{align}
  \nonumber
    \Prob{ \Setsize{ \setsize{ \neigh[C]{v} \cap \mathbf{\Vh} } - 
    1/2 &\setsize{\neigh[C]{v}} }
    \ge
    4 \diff \abs{ \neigh[G]{v} } }\\\nonumber
    &\le 
      2 \exp \left( -  
      \frac{(4\diff \abs{ \neigh[G]{v} })^2}{ \setsize{\neigh[C]{v}} + 
      4 \diff \abs{ \neigh[G]{v} } } 
      \right) \\		
    &\le \exp \left( - 8
      \diff^2 \abs{ \neigh[G]{v} }
      + 1\right)\eqcomma
  \end{align}
  where the last inequality holds as 
  $ \setsize{\neigh[C]{v}} \leq  \setsize{\neigh[G]{v}}$ and $\diff \leq 1/4$.

  By a union bound argument over the clauses in $\refpi$ and $v \in \Vr$,
  we have that Property~3 holds except with probability~$1/8$.
  For Property~2, we need to analyse vertices in $\Vl$ and in $\Vr$ separately.
  On the one hand, since $\minofset{\deg[G]{v}}{v\in \Vl} \geq (1-2\diff)\GdegMax \geq 
  \frac{\log \setsize{\Vl}}{2\diff^2}$ and $\setsize{\Vl} \ge 4$,
  a union bound over 
  $v\in \Vl$ shows
  that Property~2 holds for all vertices $\Vl$ except with 
  probability~$1/16$.
  On the other, as
  $\minofset{\deg[G]{v}}{v\in \Vr} \geq (\log \setsize{\Vr} + w_0 )/\diff^2$,
  a union bound yields that Property~2 holds for all $v \in \Vr$
  except with probability~$1/16$.

  To obtain Property~4, fix a clause $C$ and consider the set 
  ${\mathbf{\Vbad{ \mathnormal{C} } }}$ that contains
  all vertices $v\in \Vl$ satisfying
  \begin{align}
    \Abs{ \abs{ \neigh[C]{v} \cap {\mathbf{\Vh}} } - 
    1/2\abs{ \neigh[C]{v} } } > 
    4 \diff \GdegMax \eqperiod
  \end{align}
  We want to show that it is unlikely 
  that $\setsize{\mathbf{\Vbad{ \mathnormal{C} } }} \geq w_0/8$.
  Note that such a large $\mathbf{\Vbad{ \mathnormal{C} } }$
  implies by Claim~\ref{cl:non_expansion} that there
  is a set $S \subseteq \Vl$ of size $w_0/16$
  such that 
  $\Setsize{ \setsize{ \neigh[C]{S} \cap \mathbf{\Vh} } - 
    1/2\abs{ \neigh[C]{ S }}}
  \ge
  2 \diff \GdegMax \setsize{ S }$.
  By a union bound over all such sets $S$
  and applying Chernoff bound\ifthenelse{\boolean{false}}{}{~(\refthm{thm:Chernoff})} 
  we have that
  \begin{align}\nonumber
    \Prob{\setsize{&\mathbf{\Vbad{ \mathnormal{C} } }} \geq w_0/8}\\
      & \leq
        \binom{\setsize{\Vl}}{w_0/16} 
        \max_{\substack{S \subseteq \Vl:\\ \setsize{S}=w_0/16}}
    \Prob{\Setsize{ \setsize{ \neigh[C]{S} \cap \mathbf{\Vh} } - 
    1/2\abs{ \neigh[C]{ S }}}
    \ge
    \diff \GdegMax w_0/8 }\\ 
    \label{eq:part_one}
      & \leq
        \setsize{\Vl}^{w_0/16} \cdot 
        2 \exp\left( -
        \frac{(\diff \GdegMax w_0/8 )^2}{ \GdegMax w_0/16 + 
        \diff \GdegMax w_0/8  }\right) \\ 
    \label{eq:part_two}
      & \leq
        \exp\left({- \diff^2 \GdegMax w_0/8 }+ 1 + 
        \log \setsize{\Vl} \cdot  w_0/16\right) \\ 
    \label{eq:part_three}
      & \leq
        \exp\left(-\log \setsize{\Vl} \cdot w_0/16 + 1 \right)
        \eqcomma
  \end{align}
  where for \refeq{eq:part_one} we observe that
  $\abs{ \neigh[C]{ S }} \leq \GdegMax \setsize{S}$, for 
  \refeq{eq:part_two} we need that $\diff \leq 1 / 4$
  and for \refeq{eq:part_three} that
  ${\GdegMax \geq \log \setsize{\Vl}/\diff^2}$. 
  By a union bound over all clauses in~$\refpi$ we see
  that Property 4 holds except with probability~$1/8$.
\end{proof}

Let $\Vp \disjointunion \Vh$ be a partition of $\V(G)$ as 
guaranteed to exist by Lemma~\ref{lem:partition}. For an overview of the
vertex sets and how they relate we refer to \reffig{fig:onto_span}.
The following lemma shows that the vertices in $\Vl$ expand into the
set $\Vr \cap \Vh$. Let $\Gprime = G \setminus (\Vr \cap \Vp)$ 
with vertex partition $(\Vl \disjointunion (\Vr\setminus\Vp))$.
\begin{lemma}\label{lem:expansion_Gprime}
  The graph $\Gprime$ is an
  $(r, (1 + \GprimeConstantDeg \diff) \GdegMax/2, 
  (1 - \GprimeConstantExp \diff) \GdegMax/2)$-boundary 
  expander. 
\end{lemma}

\begin{proof}
  By Lemma \ref{lem:partition}, Property~\ref{item:graph}, every vertex in $\Vp\cap \Vl$ 
  has degree at most
  $(1 + \GprimeConstantDeg \diff) \setsize{\neigh[G]{v}}/2$ and at least 
  $(1 - \GprimeConstantDeg \diff) \setsize{\neigh[G]{v}}/2$.
  By the expansion guarantee of $G$, we know that 
  $\setsize{\neigh[G]{v}} \ge (1 - 2\diff) \GdegMax$.
  Therefore all sets of size 1 are good enough boundary expanders. 
  We continue by induction
  on the size of the set. Let $S$ be a set of vertices of size at most $r$. 
  In the original graph $G$,
  this set $S$ has at least $(1 - 2\diff) \GdegMax \abs{S}$ many unique 
  neighbours. Thus, there is a
  vertex $v$ in $S$ that has at least $(1 - 2\diff) \GdegMax$ many unique 
  neighbours in $G$. Further,
  by Lemma \ref{lem:partition}, Property~\ref{item:graph}, the vertex $v$ has at least 
  $(1 - \GprimeConstantDeg \diff) \GdegMax/2$
  many neighbours
  in $\Vr \cap \Vh$. Hence $v$ has at least 
  $(1 - \GprimeConstantExp \diff) \GdegMax/2$
  many unique neighbours in $\Vh$. From the induction 
  hypothesis on $S\setminus\set{v}$, it follows that $S$
  has the required number of unique neighbours in $\Vh$.
\end{proof}
 
\squeezesubsection{\PSEUDOWIDTH Lower Bound}
\label{sec:onto_lower_bound}

We start by setting up the notation we will need to prove 
Lemma~\ref{lem:onto_lower_bound}.

Let $C$ be a clause in $\refpi$, let
 $\Vbad{C} = \setdescr[:]{v \in \Vl }{
			\Abs{ \abs{ \neigh[C]{v} \cap \Vh } - 
			1/2\abs{ \neigh[C]{v} } } >
			4 \diff \GdegMax} $
and $\VbadThick{C} = (\Vthick{C} \cap \Vl) \cup \Vbad{C}$.
The closure of $C$ is a subset of $\Vl$
in the graph $\Gprime$, defined by
\begin{equation}
    \closure{C} = \closure[r, (1 - 20\diff) \gdegmax / 2]{\VbadThick{C}} 
    \eqperiod
\end{equation}
We define the closure only on $\Vl$ as we only have an expansion guarantee
from $\Vl$ into $\Vr \cap \Vh$. As the concept of closure only makes sense
on vertex sets which are expanding, we do not define it on $\Vr$.
The set of relevant vertices of a clause $C$ are the vertices in
$\closure{C} \cup \Vthick{C}$.
With this definition at
hand we proceed to set up the linear spaces that realize the lossy counting 
(see \refsec{sec:GFPHP}). Let us stress the fact that only vertices
in $\Vp$ are associated with a linear space.

  Fix a field $\mathbb{F}$ of characteristic 0 and 
  for each vertex $v \in \Vp$ let $\linspace[v]$ be a linear space over $\mathbb{F}$ of 
  dimension $\ontodimLi$. Let
  $\linspace = \bigtensor_{v \in \Vp} \linspace[v]$ and denote by
  $\linmap[v]: \Vh \mapsto \linspace[v]$ a function with the 
  property that any image of a subset 
  $S \subseteq \Vh$ of size $\setsize{S} \ge \dim(\linspace[v])$ 
  spans $\linspace[v]$, \ie
  $\spans{\linmapof[v]{u} : u \in S} = \linspace[v]$.

  Let $\matchings$ be the set of partial matchings in $G$ that contain
  no edges from $\Vp \times \Vp$. To map partial matchings 
  $\matcha \in \matchings$ into $\linspace$,
  we define
  $\linmap : \matchings \mapsto \linspace$ by
  \begin{align}\label{eq:onto_lambda_def}
    \linmapof{\matcha} = 
    \bigtensor_{v \in \vmatch{\matcha} \cap \Vp} \linmapof[v]{\matcha_v} \tensor 
    \bigtensor_{ v \in \Vp \setminus \vmatch{\matcha}} \linspace[v] \eqperiod
  \end{align}
  Recall that each partial matching $\matcha \in \matchings$ has an
  associated partial boolean assignment as defined in
  \refeq{eq:matching-to-assignment}.  For each clause $C$, we are
  interested in the partial matchings $\matcha \in \matchings$ that
  match all of $\closure{C} \cup \Vthick{C}$ such that
  ${\matcha}$ does not satisfy $C$. We refer to the set of
  such matchings as the \introduceterm{zero space} of $\clc$ and
  denote it by
  \begin{align}
    \Null{C} = \set{\matcha \in \matchings \mid 
    \vmatch{\matcha} \supseteq (\closure{C} \cup \Vthick{C}) 
    \wedge \restrict{C}{\matcha}\neq 1} \eqperiod
  \end{align}
  We associate $C$ with the linear space
  \begin{align}
    \linmapof{C} = \spans{ \linmapof{\matcha} \mid \matcha \in \Null{C}} 
    \eqperiod
  \end{align}
  Note that contradiction is mapped to $\linspace$, \ie
  $\linmapof{\bot} = \linspace$.
  
  The following lemma asserts that the span of the axioms
  $\spans{\setdescr{\linmapof{A}}{A \in \PM \cup \Ax}}$
  is a proper subspace of $\linspace$.
  
  \begin{lemma}\label{lem:onto_subspace}
    If $\setsize{\Ax} \le (1 + 16 \diff)^{w_0/8}$, then
    $\spans{\setdescr{\linmapof{A}}{A \in \PM \cup \Ax}} \subsetneq \linspace$.
  \end{lemma}
  
  Deferring the proof of this lemma for now, note 
  this implies that in the refutation
  $\refpi$ there is a resolution step deriving~$C$ from~$C_0$
  and~$C_1$ where the subspace of the resolvent is not contained
  in the span of the subspaces of the premises, or in other words
  $\linmapof{C} 
  \not\subseteq \spans{\linmapof{C_0}, \linmapof{C_1}}$.
  The following lemma, which is the heart of the argument, says that this
  cannot happen while the sets
  of relevant vertices of the clauses are small.

  \begin{lemma}\label{lem:onto_span}
    Let $C$ be derived from $C_0$ and $C_1$.
    If $
    \maxofexpr{
    \setsize{ \closure{C_0} \cup \Vthick{C_0} }, 
    \setsize{ \closure{C_1} \cup \Vthick{C_1} }, 
    \setsize{ \closure{C} \cup \Vthick{C} }
    } 
    \le r/4
    $,
    then
    $\linmapof{C} \subseteq \spans{\linmapof{C_0}, \linmapof{C_1}}.$
  \end{lemma}

  Deferring the proof of Lemma~\ref{lem:onto_span} to Section~\ref{sec:onto_span}, 
  we proceed to show how \reflem{lem:onto_lower_bound} follows from 
  what we have established so far.

  \begin{proof}[Proof of \reflem{lem:onto_lower_bound}]
  \reflem{lem:onto_subspace} and \reflem{lem:onto_span} imply that
  contradiction cannot be derived while the set of relevant
  vertices is of size at most $r/4$ and hence
  any refutation~$\refpi$ must contain a clause~$C$
  with
  $\setsize{\closure{C} \cup \Vthick{C}} \geq r/4$.
  If for such a clauses $C$ it holds that
  $\setsize{\Vthick{C}} \ge r\diff$, then \reflem{lem:onto_lower_bound} follows.
  Otherwise,
  recall that 
  $\closure{C} = \closure[r,\nu]{\VbadThick{C}}$, 
  for $\nu =  (1 - 20\diff)\gdegmax/2$, and that
  $\Gprime$ is an
  $(r, \gdegmax', c)$-boundary 
  expander by Lemma~\ref{lem:expansion_Gprime}, 
  where $\gdegmax' = (1 + \GprimeConstantDeg \diff) \GdegMax/2$
  and $c = (1 - \GprimeConstantExp  \diff)\GdegMax/2$. Thus we can 
  apply
  Lemma~\ref{lem:closure_size} to $G'$ and get that
  $\setsize{\VbadThick{C}}
  \ge \min\set{r, (r/4 - r \diff)\cdot \left( c - \nu \right) / \gdegmax'} 
  \ge 3r\diff / 2 $.
As by definition $\VbadThick{C} = (\Vthick{C} \cap \Vl) \cup \Vbad{C}$
  and by property~\ref{item:badset} of Lemma~\ref{lem:partition} we have that
  $\setsize{ \Vbad{C} } \leq w_0/8$, we conclude that
  \begin{equation}
    \width{\refpi} \ge \setsize{\Vthick{C}} 
    \ge \setsize{\Vthick{C} \cap \Vl} 
    \ge \setsize{ \VbadThick{C} } - \setsize{ \Vbad{C} }
    \ge 3r\diff/2  - w_0/8 
    \ge r\diff \eqperiod
  \end{equation}
  This completes the proof of \reflem{lem:onto_lower_bound}.
  \end{proof}

  \begin{proof}[Proof of \reflem{lem:onto_subspace}]
    Suppose $A$ is a vertex axiom $P^v$ 
    or a functionality axiom $F^{v}_{w,w'}$ as in 
    \refeq{eq:axiom-pigeon}
    and \refeq{eq:axiom-functionality}.
    Observe that $v$ is a \heavy vertex for $A$.
    Clearly, there are no matchings on $v$ that do not
    satisfy $A$. We conclude that $\linmapof{A} = \emptyset$.
  
    Let us consider $A \in \Ax$. 
    These axioms may span a part of the space $\linspace$ but 
    the fraction of the space $\linspace$ they span 
    is sufficiently small. 
    We first estimate the dimension of $\linmapof{A}$.
    By definition $\Vbad{A} = \setdescr[:]{v \in \Vl}{
			\Abs{ \abs{ \neigh[A]{v} \cap \Vh } - 
			1/2\abs{ \neigh[A]{v} } } > 
			4 \diff \GdegMax} $
    and by property~\ref{item:badset} of \reflem{lem:partition}
    it holds that
    $\setsize{\Vbad{A}} \leq w_0/8$. We partition $\Vp$ into two sets 
    $U = \Vp \setminus \bigl( \Vthick{A} \setminus \Vbad{A} \bigr)$ and
    $W = \Vp \cap \bigl( \Vthick{A} \setminus \Vbad{A} \bigr)$.
    Note that all vertices $v \in W$ satisfy that
    $\Abs{ \abs{ \neigh[A]{v} \cap \Vh } - 
			1/2\abs{ \neigh[A]{v} } } \le
			4 \diff \GdegMax$.
    Using property~\ref{item:graph} of \reflem{lem:partition}
    we get that
    \begin{align}
      \dim \linmapof{A} 
      &\le \prod_{v \in U}
      \dim \linspace[v] 
      \cdot \prod_{v \in W} 
      \bigl(\abs{ \neigh[G]{v} \cap \Vh } 
      - \abs{ \neigh[A]{v} \cap \Vh } \bigr)
      \\
      &\le \prod_{v \in U} \dim \linspace[v] 
      \cdot \prod_{v \in W}
      \bigl(
      1/2\abs{ \neigh[G]{v} } + 4 \diff \abs{ \neigh[G]{v} } ~
      -  \nonumber\\
      &\qquad\qquad\qquad\qquad\qquad\qquad
      1/2\abs{ \neigh[A]{v} } + 4 \diff \abs{ \neigh[G]{v} }
      \bigr)
      \\
      &= \prod_{v \in U} \dim \linspace[v] 
      \cdot \prod_{v \in W} 
      \bigl(
      1/2\bigl(\abs{ \neigh[G]{v} } - \abs{ \neigh[A]{v} } \bigr) 
      + 
      8 \diff \abs{ \neigh[G]{v} } 
      \bigr)
      \\
      &\le \prod_{v \in U} \dim \linspace[v] 
      \cdot \prod_{v \in W}  
      \bigl(
      1/2(\deg[G]{v} - \thr{v}) 
      + 
      \adv{v}/8
      \bigr)
      \\
      &\le \prod_{v \in U} \dim \linspace[v] 
      \cdot \prod_{v \in W} 
      \bigl( \dim \linspace[v] - \adv{v}/8 \bigr) \eqcomma
    \end{align}
    where the second to last inequality follows from 
    the fact that $\adv{v} = \ontocdiff \diff \abs{ \neigh[G]{v} } $ 
    and the last inequality from the definition of $\dim \linspace[v]$.
    
    Note that by property~\ref{item:thick} of \reflem{lem:partition}, 
    $\setsize{\Vp \cap \Vthick{A}} \geq w_0/4$ and hence 
    $\setsize{W} \geq w_0/8$.
    We conclude that the fraction of the space $\linspace$ 
    that $A$ spans is bounded by
    \begin{align}
      \frac{\dim \linmapof{A} }{\dim \linspace} 
      \le \prod_{ v \in W } \frac{\dim \linspace[v] - 
      \adv{v}/8}{\dim \linspace[v]}
      \le \left( 1 - 16{\diff} \right)^{w_0/8} \eqperiod
    \end{align}
    Along with the assumption on $\setsize{\Ax}$, this shows that not
    all of $\linspace$ is spanned by the axioms.
  \end{proof}

\squeezesubsection{Proof of \reflem{lem:onto_span}}
\label{sec:onto_span}

\begin{figure}
  \begin{center}
    \includegraphics[width=\linewidth,height=\textheight,keepaspectratio]{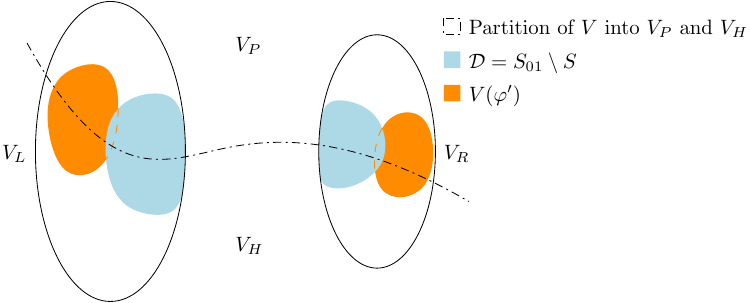}
    \caption{Depiction of relations between $\Vl, \Vr, \Vp, \Vh$ and the vertex sets in the proof of
    \reflem{lem:onto_span}}
    \label{fig:onto_span}
  \end{center}
\end{figure}

  For conciseness of notation, let us write 
  \begin{align}
    S_{01} =
    (\closure{C_0} \cup \closure{C_1}) \cup 
    (\Vthick{C_0} \cup \Vthick{C_1})
  \end{align}
  and 
  $S = \closure{C} \cup \Vthick{C}$.
  In order to establish \reflem{lem:onto_span}, we need to show
  for all $\matcha \in Z(C)$ that
  \begin{align}
    \linmapof{\matcha} \subseteq 
    \spans{ \linmapof{C_0}, \linmapof{C_1} } \eqperiod
  \end{align}
  Denote by $\matchb$ the restriction of $\matcha$ to the edges with at least
  one vertex in
  $S \cap S_{01}$ and note that $C$ is not satisfied 
  by the partial assignment associated with~$\matchb$. 
  Also, observe that  if a matching $\matche$ extends a matching $\matchf$, then 
  $\linmapof{\matche}$ is a subspace of $\linmapof{\matchf}$. This is so since
  for any vertex
  $v \in \Vp \cap \bigl( V(\matche) \setminus V(\matchf) \bigr)$
  we have  from~\refeq{eq:onto_lambda_def}
  that $\matchf$ picks up the whole subspace~$\linspace[v]$
  while $\matche$ only gets a single vector.
  Thus, if we can show that
  $\linmapof{\matchb} \subseteq \spans{ \linmapof{C_0}, \linmapof{C_1} }$,
  the statement follows since $\matcha$ extends $\matchb$ and hence 
  $\linmapof{\matcha} \subseteq \linmapof{\matchb}$. For the following
  argument it may be helpful to refer to \reffig{fig:onto_span}.
  
  Let $\diffclosure = S_{01} \setminus S$ and for a set of matchings
  $\mathcal{N} \subseteq \matchings$ let 
  $\linmapof{ \mathcal{N} } = 
  \spans{ \setdescr{ \linmapof{\matchc} }
  { \matchc \in \mathcal{N} } } $.
  In the following we show that there exists a set of matchings
  $\matchings[\diffclosure] \subseteq \matchings$ 
  that do not satisfy $C$, that cover $S_{01}$
  and such that
  \begin{equation} \label{eq:good_set}
    \linmapof{\matchb} \subseteq 
    \linmapof{ \matchings[\diffclosure] } \eqperiod
  \end{equation}

  Before arguing the existence of such a set
  $\matchings[\diffclosure]$ let us
  argue that this would imply the lemma.
  Observe that by soundness of resolution,
  no matching in $\matchings[\diffclosure]$ can satisfy 
  both $C_0$ and $C_1$ simultaneously. Fix $\matchc \in \matchings[\diffclosure]$. Without
  loss of generality, assume that $C_0$ is not satisfied. Denote
  by $\matchd \subseteq \matchc$ all edges in $\matchc$ with at least one
  vertex in 
  $\closure{C_0} \cup \Vthick{C_0}$. Clearly, $\matchd \in Z(C_0)$ and hence
  $\linmapof{\matchc} \subseteq \linmapof{\matchd} \subseteq \linmapof{C_0}$.
  Thus, for all matchings $\matchc \in \matchings[\diffclosure]$ we have that
  $\linmapof{\matchc} \subseteq \spans{\linmapof{C_0}, \linmapof{C_1} }$.
  Combining with \refeq{eq:good_set}, we get that
  \begin{equation}
    \linmapof{\matchb} \subseteq \linmapof{ \matchings[\diffclosure] } \subseteq 
    \spans{\linmapof{C_0}, \linmapof{C_1} }
  \end{equation}
  and hence the lemma follows.

  In the remainder, we show how to construct the set $\matchings[\diffclosure]$.
  Observe that all vertices $v \in \diffclosure$ are \light vertices of $C$. 
  Using property~\ref{item:clause-r}
  from \reflem{lem:partition} we get that for all $v_r \in \diffclosure \cap \Vr$
  there are at most
  \begin{equation}\label{eq:satC-r}
    	\Setsize{\neigh[C]{v_r} \cap \Vh } 
    	\le 1/2 
    	\Abs{ \neigh[C]{v_r} } + 4 \diff  \Abs{ \neigh[G]{v_r} }
    	\le 1/2 \left(\thr{v_r} - \adv{v_r} + 
    	8 \diff  \Abs{ \neigh[G]{v_r} } \right)
  \end{equation}
  mappings of $v_r$ to a vertex in $\neigh[G]{v_r} \cap \Vh$ 
  that satisfy the clause $C$.
  Similarly, using property~\ref{item:clause-l} from \reflem{lem:partition}
  and the fact that $\diffclosure \cap \Vbad{C} = \emptyset$
  we see that
  for all $v_\ell \in \diffclosure \cap \Vl$
  there are at most
  \begin{equation}\label{eq:satC-l}
        \Setsize{\neigh[C]{v_\ell} \cap \Vh } 
    	\le 1/2 
    	\Abs{ \neigh[C]{v_\ell} } + 4 \diff  \Abs{ \neigh[G]{v_\ell} }
    	\le 1/2 \left(\thr{v_\ell} - \adv{v_\ell} + 
    	8 \diff \gdegmax \right)
  \end{equation}
  mappings of $v_\ell$ to a vertex in $\neigh[G]{v_\ell} \cap \Vh$ 
  that satisfy the clause $C$.
  
  For a set of vertices
  $W \subseteq \Vp \cup \Vh$, let
  $
    \linspace[W] = 
    \bigtensor_{w \in W \cap \Vp}
    \linspace[w] $ 
  and 
  for a set $U \subseteq V(G)$
  let $\linmapl{U}$ be the projection of $\linmap$ to the space
  $\linspace[U]$ or in other words
  \begin{align}
    \linmaplof{U}{\matche} = 
    \bigtensor_{v \in V(\matche) \cap \Vp \cap U}
    \linmapof[v]{\matche_v} \tensor \bigtensor_{v \in (\Vp \cap U)
    \setminus V(\matche)} \linspace[v] \eqperiod
  \end{align}
  We extend the notation to sets of matchings as previously 
  for $\linmap$.
  In order to establish \refeq{eq:good_set}, we have to argue that 
  $\linmaplof{\diffclosure \setminus
    V(\matchb)}{\matchings[\diffclosure]}$ 
  spans the space 
  $\linspace[\diffclosure \setminus \vmatch{\matchb}]$.
  At this point, we deviate from the $\GFPHP$ proof.
  Note that we only have expansion for the vertices $\Vl$ into $\Vh$
  but $\diffclosure$ may also contain vertices from $\Vr$. Thus we
  cannot apply the argument from \refsec{sec:GFPHP} to all vertices. 

  Instead, we split the argument into 2 seperate parts.
  First, by an argument similar to the lower bound proof of the
  $\GFPHP$ formulas, we show that vertices in $\diffclosure \cap \Vl$ 
  can be matched in many ways. This will in particular imply that 
  $\linmaplof{(\diffclosure \cap \Vl) \setminus V(\matchb)}
  {\matchings[\diffclosure]}$ spans all of
  $\linspace[(\diffclosure \cap \Vl) \setminus \vmatch{\matchb}]$. 
  After that we consider the vertices in $\diffclosure \cap \Vr$. 
  As these
  vertices have very high degree, there are always enough
  neighbours they can be matched to and therefore
  $\linmaplof{(\diffclosure \cap \Vr) \setminus \vmatch{\matchb}}
  {\matchings[\diffclosure]}$ spans all of
  $\linspace[(\diffclosure \cap \Vr) \setminus \vmatch{\matchb}]$. 
  Note that this second argument is essentially the span argument from
  \cite{Razborov03ResolutionLowerBoundsWFPHP}.
  
  Consider the vertex set $\diffclosure \cap \Vl$. Note that 
  $\diffclosure \cap \Vl$ is completely
  outside the $\closure{C}$. Since, by assumption, the cardinality of $\closure{C}$ 
  is upper bounded by $r /
  4$ and $\abs{\diffclosure \cap \Vl} \le \abs{S_{01}} \le r / 2$, by 
  Lemma~\ref{lem:closure_expansion}
  we get that 
  \begin{equation}
  \begin{aligned}
    \abs{ 
    \uniqueNeigh[ {\Gprime \setminus 
    (\closure{C} \cup \neigh[\Gprime]{ {\closure{C}} })} ]
    {\diffclosure \cap \Vl} } 
    \ge 
    1/2(1 - \closureconstant \diff)
    \GdegMax \setsize{\diffclosure \cap \Vl}
    \eqperiod
  \end{aligned}
  \end{equation}
  By an averaging 
  argument, there is a $v \in
  \diffclosure \cap \Vl$ that has at least 
  $(1 - \closureconstant \diff)\GdegMax/2$ unique
  neighbours in 
  $\uniqueNeigh[ {\Gprime \setminus (\closure{C} \cup 
  \neigh[\Gprime]{\closure{C}})} ]
  {\diffclosure \cap \Vl}$.
  By iterating this argument on $(\diffclosure \cap \Vl) \setminus \set{ v }$
  we get a partition 
  $V_{v_1} \disjointunion V_{v_2} \ldots 
  \disjointunion V_{v_{\setsize{\diffclosure \cap \Vl}}}$
  of the neighbourhood $\diffclosure \cap \Vl$. The key properties
  of this partition are that every vertex 
  $v_\ell \in \diffclosure \cap \Vl$ can independently 
  be matched 
  to any vertex in $V_{v_\ell}$ and each set is of size at least
  $\setsize{V_{v_\ell}} \ge (1 - \closureconstant \diff)\GdegMax/2$.
  Using~\refeq{eq:satC-l}, we have that each vertex 
  $v_\ell \in \diffclosure\cap\Vl$ can be matched to at least
  \begin{equation}
  \begin{aligned}
  	1/2 (1 - \closureconstant \diff)\GdegMax - 
  	1/2 (\thr{v_\ell} - \adv{v_\ell} + 8 \diff \GdegMax)
  	&= 1/2 \left( \GdegMax - \thr{v_\ell} + \adv{v_\ell} - 
  	\closureconplus \diff \GdegMax \right)\\
  	&\ge 1/2 \left(\deg[G]{v_\ell} - \thr{v_\ell} + \adv{v_\ell}/2 \right)
  \end{aligned}
  \end{equation}
many vertices in $V_{v_\ell}$ without satisfying $C$. Denote these vertices
  by $V'_{v_\ell}$.
  As in section \refsec{sec:GFPHP}, we would like to conclude that
  every vertex has many choices of vertices it can independently be mapped to and
  therefore there are enough matchings to span the
  space $\linspace[\diffclosure \cap \Vl]$.
  Unfortunately this argument does
  not work since vertices in $V'_{v_\ell}$ can be matched in
  $\matchb$ and are hence not available to be matched to
  $v_\ell$, so there might be too few matchings of $v_\ell$ to span the whole space $\linspace[v_\ell]$.

  We could attempt to overcome this problem by removing all edges in
  $\matchb$ with a vertex in one of the sets $V'_{v_\ell}$.
  This allows us to independently match all the vertices in
  $\diffclosure \cap \Vl$ to sufficiently many
  neighbours.
  Regrettably, this edge removal strategy turns out to be too aggressive: it can
  occur that
  a vertex from $S_{01} \cap \Vr$,
  previously matched by $\matchb$,
  now has no neighbour available to be matched to.
  Fortunately, this only happens to vertices that were matched in
  $\matchb$. The solution that suggests itself is to remove
  edges from $\matchb$ in a ``lazy'' manner: only remove an edge $\set{u,v}$ from
  $\matchb$ when one of the vertices should be matched to some
  $v_\ell \in \Vl$. This ensures that no vertex in $\Vr$ that was previously
  matched by $\matchb$ is suddenly unmatched.
  This is the main idea of 
  Algorithm~\ref{alg:extendmatchings} which takes care of the necessary
  edge removals.

  Let
  $\matchings[\diffclosure \cap \Vl] = $ 
  \textsc{ExtendMatching}$(\diffclosure \cap \Vl, \matchb, 
  V'_{v_1}, \ldots, V'_{v_{\setsize{\diffclosure \cap \Vl}}})$. 
  Note that the
  algorithm terminates on this input as the sets
  $V'_{v_1}, V'_{v_2}, \ldots, V'_{v_{\setsize{\diffclosure \cap \Vl}}}$ are disjoint.
  Let us establish some claims regarding 
  $\matchings[\diffclosure \cap \Vl]$.

  \SetKwComment{Comment}{$\triangleright$\ }{}
  \begin{algorithm}[t]
  \caption{Extend Matching}
  \label{alg:extendmatchings}
    \SetKwFunction{extmatch}{ExtendedMatching}
  \SetKwProg{myproc}{procedure}{}{}
\myproc{\extmatch{$\setalg, \matchc, V_{v_1}, V_{v_2}, \ldots, 
    V_{v_{\setsize{\setalg}}}$} \Comment*[f]{extend $\matchc$ to domain $\setalg$}}{%
    \If{$\setalg \setminus V(\matchc) \neq \emptyset$ \Comment*[r]{still need to extend $\matchc$}}{
     $\matchings \gets \emptyset$\;
     $v_\ell \gets_{\text{any}} \setalg \setminus V(\matchc)$\;
       \For{$w \in V_{v_\ell}$ \Comment*[r]{$v_\ell$ can be matched to $w$}}{
           $\matchd \gets \matchc$\;
         \If{$\exists w'$ such that $\set{w, w'} \in \matchc$}{
              $\matchd \gets \matchd \setminus \set{w, w'}$ \Comment*[r]{remove $w$ from the matching}
          }
          $\matchd \gets \matchd \cup \set{v_\ell, w}$ \Comment*[r]{match $v_\ell$ to $w$}
          $\matchings = \matchings \cup\,$\extmatch{$\setalg,\matchd, V_{v_1}, V_{v_2}, \ldots,
          V_{v_{\setsize{\setalg}}}$}\;
    }
    \textbf{return} $\matchings$\;
  }
  \Else{
    \textbf{return} $\matchc$\;
    }
    }
\end{algorithm}

The first claim states that the algorithm cannot remove edges from
$\matchb$ with 
a vertex in $S \cap S_{01} \cap \Vl$.
This is important as we want
to get matchings that are defined on all of $S_{01} \cap \Vl$.
As the algorithm only tries to match vertices in 
$\diffclosure \cap \Vl = (S \setminus S_{01}) \cap \Vl$,
we must ensure that the edges in $\matchb$ with an endpoint in
$S \cap S_{01} \cap \Vl$ are not erased.
Note that all edges that are removed by the algorithm have an endpoint in the
neighbourhood of $\diffclosure \cap \Vl$. Hence it suffices to show
that the vertices from $S \cap S_{01} \cap \Vl$ are not matched to
a vertex in the neighbourhood of $\diffclosure \cap \Vl$.

\begin{claim}\label{clm:neighbours}
  The matching $\matchb$ contains no edge $\set{w, w'}$ such that
  $$w \in \neigh[
  {\Gprime \setminus (\closure{C} \cup \neigh[\Gprime]{\closure{C}})}]
  {\diffclosure \cap \Vl}$$
  and
  $w' \in S \cap S_{01} \cap \Vl$.
\end{claim}

\begin{proof}
  Suppose there is an edge $\set{w, w'} \in \matcha$ for $w, w'$ as
  in the lemma statement. As 
  $S \cap S_{01} \cap \Vl\subseteq \closure{C}$, we see that 
  $w \in \neigh[\Gprime]{\closure{C}}$. 
  But this is a contradiction since $w$ is not in the graph
  $\Gprime \setminus (\closure{C} \cup \neigh[\Gprime]{\closure{C}})$. 
\end{proof}

Next, we consider edges in $\matchb$ with a vertex in the set $\Vp \cap \Vr$.
Observe that if the algorithm removed such an edge, then the linear
space associated with the new matching would differ from the original
space in a non-trivial way. Fortunately, this cannot happen. 

\begin{claim}\label{clm:dom_matchings}
  All matchings $\matchc \in \matchings[\diffclosure \cap \Vl]$ 
  cover the set
  $S_{01} \cap \Vl$
  and an edge $e \in \Vl \times (\Vp \cap \Vr)$ is contained in
  $\matchc$ if and only if it is contained in $\matchb$. Furthermore,
  if a vertex $v \in \Vr$ is matched in $\matchb$, then it is matched
  in every $\matchc \in \matchings[\diffclosure \cap \Vl]$.
\end{claim}

\begin{proof}
  By Claim~\ref{clm:neighbours}, 
  Algorithm~\ref{alg:extendmatchings} never removes edges from $\matchb$
  that are incident to
  a vertex in $S_{01} \cap S \cap \Vl$. As 
  $\matchb$ covers all of $S_{01} \cap S \cap \Vl$,
  it follows that every
  $\matchc \in \matchings[\diffclosure \cap \Vl]$ also covers the set
  $S \cap S_{01} \cap \Vl$. 
  Furthermore, the algorithm ensures that every 
  $\matchc \in \matchings[\diffclosure \cap \Vl]$
  covers the set $\diffclosure \cap \Vl = (S_{01} \setminus S) \cap
  \Vl$.
  Combining these statements we see that every matching
  $\matchc \in \matchings[\diffclosure \cap \Vl]$ covers
  $S_{01} \cap \Vl$.
  
  We observe that all edges in $\matchb$ that
  may be deleted by the algorithm must have an endpoint in one of the
  sets $V'_{v_\ell}$ and all these sets are contained in
  $\Vh \cap \Vr$. As the graph is bipartite (with bipartition
  $\Vl \disjointunion \Vr$) and the set $\matchings$ does not contain
  matchings with edges from $\Vp \times \Vp$, we see that vertices from
  $\Vp \cap \Vr$ can only be matched to vertices in
  $\Vh \cap\Vl$.
  Therefore the algorithm cannot change edges in $\matchb$ with
  an endpoint in $\Vp \cap \Vr$.
  This implies that if an edge $e \in \Vl \times (\Vp \cap \Vr)$ is in $\matchb$, then it is also in $\matchc$.
  For the other direction, observe that since the algorithm can only add
  edges to $\matchc$ with an endpoint in $\Vh \cap \Vr$, and since the graph
is bipartite, the algorithm does not add an edge from $ \Vl \times (\Vp \cap \Vr)$.

  Finally, the fact that all matched vertices $v \in \Vr$ in $\matchb$ are also matched
  in every $\matchc \in \matchings[\diffclosure \cap \Vl]$ 
  follows from the ``lazy'' removal of edges from $\matchb$.
\end{proof}

We can now show that our set of matchings spans the appropriate space
when projected to $\Vl$.
Note that for a matching $\matche$ it holds that
$\linmapof{\matche} = \linmaplof{U}{\matche} \tensor \linmaplof{\Vp
  \setminus U}{\matche}$ for any set  $U$, but the same does not hold
for sets of matchings: span does not commute with tensor.

\begin{claim}\label{clm:span_matchings}
  $
    \linmaplof{\Vl}{\matchb}
    \subseteq
    \linmaplof{\Vl}{\matchings[\diffclosure \cap \Vl]}
  $
\end{claim}

\begin{proof}
  Let us write
  \begin{align}
    \linmaplof{\Vl}{\matchb} 
    &=
    \linmaplof{V(\matchb) \cap \Vl}{\matchb} \tensor \linspace[\Vl
      \setminus V(\matchb)]\\
    &=
    \linmaplof{V(\matchb) \cap \diffclosure \cap \Vl}{\matchb} \tensor
      \linmaplof{(V(\matchb) \cap \Vl) \setminus
      \diffclosure}{\matchb} \tensor\nonumber\\
    &\qquad
      \linspace[(\diffclosure \cap \Vl) \setminus V(\matchb)] \tensor
      \linspace[\Vl \setminus (\diffclosure \cup V(\matchb))] 
      \eqperiod
  \end{align}
  Note that no matching $\matchc \in \matchings[\diffclosure \cap
  \Vl]$ covers any of the vertices in $\Vl \setminus (\diffclosure \cup V(\matchb))$.
  This holds as the algorithm can only add edges from the set
  $(\diffclosure \cap \Vl) \times (\Vh \cap \Vr)$.
  Hence we can write
  \begin{align}
    \linmaplof{\Vl}{\matchings[\diffclosure \cap \Vl]}
    &=
    \linmaplof{\Vl \cap (\diffclosure \cup V(\matchb))}
      {\matchings[\diffclosure \cap \Vl]} \tensor
      \linspace[\Vl \setminus (\diffclosure \cup V(\matchb))]
      \eqperiod
  \end{align}
  Thus we can ignore the space 
  $\linspace[\Vl \setminus (\diffclosure \cup V(\matchb))]$
  for the remainder of this argument.
  From the algorithm it should be evident that
  \begin{align}\label{eq:d_not_def}
    \linmaplof{(\diffclosure \cap \Vl) \setminus V(\matchb)}
      {\matchings[\diffclosure \cap \Vl]}
    &=
      \linspace[(\diffclosure \cap \Vl) \setminus V(\matchb)]
  \end{align}
  as every vertex in 
  $v \in (\diffclosure \cap \Vl) \setminus V(\matchb)$
  is independently matched to every vertex in $V'_{v}$ of size 
  $\setsize{V'_{v}} \ge \ontodimLi$. As the dimension of
  $\dim(\linspace[v]) = \ontodimLi$, we conclude that
  $\linspace[(\Vl \cap \diffclosure) \setminus V(\matchb)]$ is spanned.
  
  To continue the argument, we need the following equivalence relation
  on matchings. Two matchings
  $\matchc, \matchd \in \matchings[\diffclosure \cap \Vl]$ are
  equivalent on a vertex set $V$ if they match the vertices in $V$ in
  the same way, that is, for $v \in V$ we have that
  $\matchc_v = \matchd_v$. We denote the equivalence class with
  respect to the vertex set $V$ over
  $\matchings[\diffclosure \cap \Vl]$ of a matching
  $\matchc\in \matchings[\diffclosure \cap \Vl]$ by $\set{\matchc}_V$.

  We want to show that for every
  $\matchc \in \matchings[\diffclosure \cap \Vl]$ it holds that
  \begin{align}\label{eq:d_def}
    \linmaplof{V(\matchb) \cap \diffclosure \cap \Vl}{\matchb}
    &\subseteq
      \spans{\linmaplof{V(\matchb) \cap \diffclosure \cap \Vl}
      {\matchd} \mid \matchd \in
      \set{\matchc}_{(\diffclosure \cap \Vl) \setminus V(\matchb)}} \eqperiod
  \end{align}
  Note that in combination with \refeq{eq:d_not_def} we get that
  \begin{align} \label{eq:d}
    \linmaplof{\diffclosure \cap \Vl}{\matchb}
    &=
      \linmaplof{(\diffclosure \cap \Vl) \setminus V(\matchb)}
      {\matchb} \tensor
      \linmaplof{V(\matchb) \cap \diffclosure \cap \Vl}{\matchb}\\
    &=
      \linspace[(\diffclosure \cap \Vl) \setminus V(\matchb)] \tensor
      \linmaplof{V(\matchb) \cap \diffclosure \cap \Vl}{\matchb}\\
    &\subseteq
      \linmaplof{\diffclosure \cap \Vl}{
      \matchings[\diffclosure \cap \Vl]}
      \eqperiod
  \end{align}

  We prove \refeq{eq:d_def} by induction on subsets of
  $V(\matchb) \cap \diffclosure \cap \Vl$. The statement clearly holds
  for the empty set. 
  Fix $U \subseteq V(\matchb) \cap \diffclosure \cap \Vl$ and a vertex
  $u \in U$. By induction, we may assume that
  \begin{align}\label{eq:ih}
    \linmaplof{U \setminus \set{u}}{\matchb}
    &\subseteq
      \spans{
      \linmaplof{U
      \setminus \set{u}}
      {\matchd} \mid \matchd \in 
      \set{\matchc}_{(\diffclosure \cap \Vl) \setminus V(\matchb)}} 
      \eqperiod
  \end{align}
  We want to show that the statement also holds for the set $U$.
  Note that 
    $\linmaplof{U}{\matchb} 
    = 
    \linmaplof{U \setminus \set{u}}{\matchb} \tensor 
    \linmapof[u]{\matchb_u}$. Further,
  \begin{align}
    &\spans{\linmaplof{U}{\matchd}
      \mid \matchd \in \set{\matchc}_{(\diffclosure \cap \Vl)
      \setminus V(\matchb)}
      } =\\
    &\spans{\linmaplof{U \setminus \set{u}}{\matchd} ~\tensor
      \nonumber\\
    &\qquad
    \spans{\linmapof[u]{\matche} \mid \matche \in
      \set{\matchd}_{((\diffclosure \cap \Vl) \setminus
    V(\matchb)) \cup (U \setminus \set{u})}}
      \mid \matchd \in \set{\matchc}_{(\diffclosure \cap \Vl)
      \setminus V(\matchb)}
      } \eqperiod
  \end{align}
  Suppose that for every $\matchd \in \set{\matchc}_{(\Vl \cap \diffclosure)
      \setminus V(\matchb)}$ it holds that
  \begin{align}\label{eq:linspace_u_spanned}
    \linmapof[u]{\matchb_u} \subseteq \spans{\linmapof[u]{\matche_u} 
    \mid \matche \in \set{\matchd}
    _{((\diffclosure \cap \Vl) \setminus V(\matchb))\cup (U \setminus
    \set{u})}}
    \eqperiod
  \end{align}
  Then, continuing from above, we see that
  \begin{align}
    \spans{\linmaplof{U}{\matchd}
      \mid &\matchd \in \set{\matchc}_{(\diffclosure \cap \Vl)
      \setminus V(\matchb)}
      } \\
    &\supseteq
      \spans{\linmaplof{U \setminus \set{u}}{\matchd}
      \mid \matchd \in \set{\matchc}_{(\diffclosure \cap \Vl)
      \setminus V(\matchb)}
      }
      \tensor \linmapof[u]{\matchb_u}\\
    &\supseteq
      \linmaplof{U \setminus \set{u}}{\matchb} \tensor 
      \linmapof[u]{\matchb_u}\\
    &= \linmaplof{U}{\matchb} \eqcomma
  \end{align}
  where the second inclusion holds by the induction hypothesis
  \refeq{eq:ih}. 
  Thus, to show the statement for $U$ we just need to show 
  \refeq{eq:linspace_u_spanned}.
  To this end, fix a matching 
  ${\matchd \in \set{\matchc}_{(\diffclosure \cap \Vl) \setminus V(\matchb)}}$.
  Note that if there is a matching
  $\matche \in \set{\matchd}_{((\diffclosure \cap \Vl) \setminus
    V(\matchb)) \cup (U \setminus \set{u})}$ 
  such that $\matche_u = \matchb_u$, then we are done.
  Otherwise, Algorithm~\ref{alg:extendmatchings} removed the edge
  that mached the vertex $u$ in $\matchb$. Hence the vertex $u$ is
  matched by the procedure to at least $\setsize{V'_{u}} \ge \ontodimLi$
  different vertices. As the dimension of $\linspace[u] = \ontodimLi$,
  we see that all of the space is spanned. We conclude that
  \refeq{eq:linspace_u_spanned} holds.

  What remains is to argue that for every $\matchc \in
  \matchings[\diffclosure \cap \Vl]$ it holds that
  \begin{align}\label{eq:not_d}
    \linmaplof{(V(\matchb) \cap \Vl) \setminus \diffclosure}{\matchb}
    &\subseteq
      \spans{\linmaplof{(V(\matchb) \cap \Vl) \setminus \diffclosure}
      {\matchd} \mid \matchd \in
      \set{\matchc}_{\diffclosure \cap \Vl}} \eqperiod
  \end{align}
  The argument goes along the same lines as for the vertices in 
  $V(\matchb) \cap \diffclosure \cap \Vl$ and we thus omit it.

  We can then combine \refeq{eq:d} and \refeq{eq:not_d} to conclude
  the claim.
\end{proof}

  Observe that the matchings in $\matchings[\diffclosure \cap \Vl]$ are
  not necessarily extensions of $\matchb$. This is not a problem, however,
  since the matchings only
  differ in edges that contain vertices which either do not show up in the linear space or
  for which the whole linear space associated to the vertex is spanned. Furthermore, vertices from $\diffclosure \cap \Vh \cap \Vl$ are matched to
  many vertices
  even though a single vertex would have been sufficient.

  It remains only to show that every matching
  $\matchc \in \matchings[\diffclosure \cap \Vl]$ can be extended in
  many ways to the set $\diffclosure \cap \Vr$.
  Fix a matching 
  $\matchc \in \matchings[\diffclosure \cap \Vl]$ and recall
  that these are defined on $S_{01} \cap \Vl$.
  Note that by Lemma~\ref{lem:partition}, property~\ref{item:graph},
  each $v \in \diffclosure \cap \Vr$ has at least
  \begin{equation}
  \begin{aligned}
  	\Setsize{\neigh[G]{v} \cap \Vh} \ge 
  	1/2 \Setsize{ \neigh[G]{v} } - 4 \diff \Setsize{ \neigh[G]{v} } 
  	\label{eq:largeneighbourhood}
  \end{aligned}
  \end{equation}
  many neighbours in $\Vh$.
  Using~\refeq{eq:satC-r} we can now bound the number of matchings that do not 
  satisfy $C$.

  Note that the matching $\matchc$ contains at most 
  $\setsize{S_{01}} \leq r/2$ many edges. 
  Since $G$ is bipartite, this implies that for any $v \in \Vr$ at most 
  $r / 2$ neighbours are already matched.
  Observe that some vertex
  $v \in \diffclosure \cap \Vh \cap \Vr$ 
  may have been matched by Algorithm~\ref{alg:extendmatchings}. As
  these vertices are not associated with a linear space, we only need
  to match these vertices with a single vertex and hence we can just leave them
  matched as in $\matchc$.
  Further, by Claim~\ref{clm:dom_matchings},
  we see that the vertices in $\diffclosure \cap \Vp \cap \Vr$ were
  not matched by Algorithm~\ref{alg:extendmatchings}. All these will
  be matched in many ways as needed:
  If $v \in \diffclosure \cap \Vr$ is not matched by $\matchc$, then 
  by~\refeq{eq:largeneighbourhood} and~\refeq{eq:satC-r} it can 
  be matched to at least
  \begin{align}
  \nonumber
  	1/2 \big( \Setsize{ \neigh[G]{v} } - 
  	&8 \diff \Setsize{ \neigh[G]{v} }  - 
  	\thr{v} + \adv{v} - 8 \diff  \Abs{ \neigh[G]{v} } - r \big)\\
    \nonumber
  	&= 1/2 \left( \deg[G]{v} - \thr{v} + \adv{v} - 
  	16\diff\deg[G]{v} - r\right) \\
  	&\ge 1/2 \left( \deg[G]{v} - \thr{v} + \adv{v} - 
  	17\diff\deg[G]{v} \right)\label{eq:rsmall} \\
  	\nonumber
  	&\ge 1/2 \left( \deg[G]{v} - \thr{v} + \adv{v}/2\right)
  \end{align}
  many vertices without satisfying the clause $C$. 
  Note that in~\refeq{eq:rsmall} we used the assumption 
  that $\deg[G]{v} \ge r \diff$ for $v \in \Vr$. 
  As we have that
  $\dim(\linspace[v]) = \ontodimLi$, we conclude that
  the extensions of $\psi$ can span the linear space
  $\linspace[(\diffclosure \cap \Vr) \setminus \vmatch{\matchb}]$.
  Hence, by extending each $\matchc \in \matchings[\diffclosure \cap \Vl]$,
  we get a set of matchings $\matchings[\diffclosure]$, which do not satisfy the
  clause $C$, are defined on $S_{01}$ and 
  $\linmapof{\matchb} \subseteq \linmapof{\matchings[\diffclosure]}$.
  This establishes the lemma.

\section{Concluding Remarks}
\label{sec:conclusion}

In this work, we extend the \pseudowidth method developed by 
Razborov \cite{Razborov04ResolutionLowerBoundsPM,Razborov03ResolutionLowerBoundsWFPHP} 
for proving lower bounds on severely overconstrained CNF formulas in
resolution.
In particular, we establish that pigeonhole principle formulas and
perfect matching formulas over highly unbalanced bipartite graphs 
remain exponentially hard for resolution even when these graphs are sparse.
This resolves an open problem in~\cite{Razborov04ResolutionLowerBoundsPM}.

The main technical difference in our work compared to
\cite{Razborov04ResolutionLowerBoundsPM,Razborov03ResolutionLowerBoundsWFPHP} 
goes right to the heart of the proof, where one wants to argue that
resolution in small \pseudowidth cannot make progress towards a
derivation of contradiction. Here Razborov uses the global symmetry
properties of the formula, whereas we resort to a local argument based on
graph expansion. This argument needs to be carefully combined with a
graph closure operation as 
in~\cite{ABRW04Pseudorandom,AR03LowerBounds} to ensure that the
residual graph always remains expanding as matched pigeons and their
neighbouring holes are removed. It is this change of perspective that
allows us to prove lower bounds for sparse bipartite graphs
with the size~$m$ of the left-hand side (\ie the number of pigeons)
varying all the way from linear to exponential in the size~$n$ of the
right-hand size (\ie the number of pigeonholes),
thus covering the full range between
\cite{BW01ShortProofs} on the one hand and
\cite{Raz04Resolution,Razborov04ResolutionLowerBoundsPM,Razborov03ResolutionLowerBoundsWFPHP}
on the other.

One shortcoming of our approach is that the sparse expander graphs are
required to have very good expansion---for graphs of left
degree~$\gdegmax$, the size of the set of unique neighbours of any not
too large left vertex set has to scale like
$(1 - \littleoh{1})\gdegmax$.  We would like to prove that graph PHP
formulas are hard also for graphs with constant expansion
$(1 - \varepsilon)\gdegmax$ for some $\varepsilon > 0$, but there
appear to be fundamental technical barriers to extending our lower
bound proof to this setting---in particular, since we can only obtain
a very small gap in the filter lemma.

Another intriguing problem left over
from~\cite{Razborov04ResolutionLowerBoundsPM} 
is to determine the true resolution complexity of weak PHP formulas
over complete bipartite graphs
$K_{m,n}$ as $m \to \infty$.
The best known upper  bound from~\cite{BP97Weak}
is
$\exp \bigr( \Bigoh{\sqrt{n \log n}} \bigr)$,
whereas the lower bound 
in~\cite{Razborov04ResolutionLowerBoundsPM,Razborov03ResolutionLowerBoundsWFPHP}
is
$\exp \bigl( \Bigomega{\sqrt[3]{n}} \bigr)$.
It does not seem unreasonable to hypothesize that 
$\exp \bigl( \Bigomega{\sqrt[2]{n}} \bigr)$
should be the correct lower bound (ignoring lower-order terms), but
establishing such a lower bound again appears to require substantial
new ideas.

We believe that one of the main contributions of our work is that it
again demonstrates the power of Razborov's \pseudowidth method, and we
are currently optimistic that it could be useful for solving other
open problems for resolution and other proof systems. Unfortunately
\pseudowidth is a quite difficult method to apply; we leave it as an
open problem to obtain simplifications of the current approach.

For resolution, an interesting question mentioned
in~\cite{Razborov04ResolutionLowerBoundsPM} is whether \pseudowidth
can be useful to prove lower bounds for formulas that encode the
Nisan--Wigderson generator~\cite{ABRW04Pseudorandom,Razborov15PseudorandomGeneratorsHard}.
Since the clauses in such formulas encode local constraints, we hope
that techniques from our paper could 
be helpful.
Another long-standing open problem is to prove lower bounds on proofs
in resolution that
\mbox{$k$-clique} free sparse graph
do not contain
\mbox{$k$-cliques}, where the expected length lower bound would be
$n^{\bigomega{k}}$.
Here we only know weakly exponential lower bounds for quite dense
random  graphs~\cite{BIS07IndependentSets,Pang21},
although an asymptotically optimal
$n^{\bigomega{k}}$ lower bound 
has been
established in the sparse
regime for the restricted subsystem of regular 
resolution~\cite{ABdRLNR18Clique}.

Finally, we want to highlight that for the stronger proof system
\introduceterm{polynomial calculus}
\cite{ABRW02SpaceComplexity,CEI96Groebner}
no lower bounds on proof size are known for PHP formulas with
$m \geq n^2$ pigeons. It would be very interesting if some kind of
``pseudo-degree'' method could be developed that would finally lead to
progress on this problem.

\section*{Acknowledgements}

First and foremost, we are most grateful to Alexander Razborov for
many discussions about \pseudowidth, graph closure, and other
mysteries of the universe. We also thank Paul Beame and Johan Håstad
for useful discussions, and Jonah Brown-Cohen for helpful references
on expander graphs. We gratefully acknowledge the feedback from
participants of the Dagstuhl workshop 19121 \emph{Computational
  Complexity of Discrete Problems} in March 2019 and would like to
thank the referees for their thorough reviews.

\printbibliography

\end{document}